\documentclass[11pt]{article}
\usepackage[paper=letterpaper,margin=1in]{geometry}
\usepackage[T1]{fontenc}
\usepackage{lmodern}
\usepackage{amsmath}
\usepackage{amssymb}
\usepackage{amsthm}
\usepackage{enumitem}
\usepackage{bm}
\usepackage[ruled,vlined,linesnumbered]{algorithm2e}
\SetKwFor{InfLoop}{repeat}{}{}

\numberwithin{equation}{section}
\usepackage{xspace}
\usepackage{bm}

\usepackage[capitalize]{cleveref}

\def\showauthornotes{1}
\usepackage{color}

\usepackage{array}

\usepackage{graphicx}

\usepackage{flafter}

\usepackage{tikz}
\usetikzlibrary{calc}
\usetikzlibrary{positioning}
\usetikzlibrary{matrix}
\usetikzlibrary{shapes}
\usetikzlibrary{plotmarks}
\usetikzlibrary{decorations.pathmorphing}
\usetikzlibrary{decorations.pathreplacing}
\usetikzlibrary{patterns}
\usetikzlibrary{arrows,positioning} 
\tikzstyle{vertex}=[circle, draw,fill=gray!20, inner sep=0pt, minimum size=16pt]
\tikzstyle{svertex}=[circle, draw,fill=gray!20, inner sep=0pt, minimum size=8pt]
\tikzstyle{ssquare}=[draw,fill=gray!20, inner sep=0pt, minimum size=8pt]
\tikzset{
        >=stealth',
            pil/.style={
           ->,
           thick,
           shorten <=2pt,
           shorten >=2pt,}
}

\newtheorem{theorem}{Theorem}[section]
\newtheorem{lemma}[theorem]{Lemma}
\newtheorem{definition}[theorem]{Definition}
\newtheorem{remark}[theorem]{Remark}
\newtheorem{corollary}[theorem]{Corollary}

\newtheorem{claim}[theorem]{Claim}
\newtheorem{proposition}[theorem]{Proposition}
\newtheorem*{claim*}{Claim}
\newtheorem*{remark*}{Remark}

\newtheorem*{fact*}{Fact}
\newtheorem{invariant}[]{Invariant}

\theoremstyle{definition}

\makeatletter
\newtheorem*{rep@theorem}{\rep@title}
\newcommand{\newreptheorem}[2]{\newenvironment{rep#1}[1]{ \def\rep@title{#2 \ref{##1}} \begin{rep@theorem}} {\end{rep@theorem}}}
\makeatother

\newreptheorem{reduction}{Reduction}
\newreptheorem{lemma}{Lemma}

\colorlet{darkgreen}{green!50!black}

\usepackage{forloop}
\newcommand{\defcal}[1]{\expandafter\newcommand\csname c#1\endcsname{{\mathcal{#1}}}}
\newcounter{ct}
\forLoop{1}{26}{ct}{
  \edef\letter{\Alph{ct}}
  \expandafter\defcal\letter
}

\ifnum\showauthornotes=1
\newcommand{\Authornote}[2]{{\sffamily\small\color{red}{[#1: #2]}}}
\newcommand{\Authornotecolored}[3]{{\sffamily\small\color{#1}{[#2: #3]}}}
\newcommand{\Authorcomment}[2]{{\sffamily\small\color{gray}{[#1: #2]}}}
\newcommand{\Authorstartcomment}[1]{\sffamily\small\color{gray}[#1: }

\newcommand{\Authorfnote}[2]{\footnote{\color{red}{#1: #2}}}
\newcommand{\Authorfixme}[1]{\Authornote{#1}{\textbf{??}}}
\newcommand{\Authormarginmark}[1]{\marginpar{\textcolor{red}{\fbox{\Large #1:!}}}}
\else
\newcommand{\Authornote}[2]{}
\newcommand{\Authornotecolored}[3]{}
\newcommand{\Authorcomment}[2]{}
\newcommand{\Authorstartcomment}[1]{}

\newcommand{\Authorfnote}[2]{}
\newcommand{\Authorfixme}[1]{}
\newcommand{\Authormarginmark}[1]{}
\fi

\newcommand{\Onote}{\Authornote{O}}

\DeclareMathOperator{\poly}{poly}

\renewcommand{\vec}[1]{{\bm{#1}}}
\newcommand{\ie}{i.e.,\xspace}
\newcommand{\eg}{e.g.,\xspace}

\newcommand{\etal}{et al.\xspace}

\newcommand{\E}{{\mathbb{E}}}

\newcommand{\cDD}{\ensuremath{\mathcal{D}}}
\newcommand{\cFF}{\ensuremath{\mathcal{F}}}

\newcommand{\bRn}{\ensuremath{\mathbb{R_{+}}}}

\newcommand{\opt}{{\ensuremath{\mathrm{OPT}}}\xspace}
\newcommand{\lpopt}[1]{{\ensuremath{\mathrm{OPT #1}}}\xspace}

\newcommand{\argmin}{\operatorname{arg\min}}

\newcommand{\lmp}{\ensuremath{\mathsf{LMP}}\xspace}

\newcommand{\hide}[1]{}
\newcommand{\sgraph}[1]{\ensuremath{\mathcal{SG}
\ifthenelse{\equal{#1}{}}{}{(#1)}
}}
\newcommand{\cgraph}[1]{\ensuremath{\mathcal{CG}
\ifthenelse{\equal{#1}{}}{}{(#1)}
}}
\newcommand{\cpath}[2]{\ensuremath{P_{#1}
\ifthenelse{\equal{#2}{}}{}{(#2)}
}}

\newcommand{\lpf}{\textsf{\small LP$(\lambda)$}\xspace}
\newcommand{\jv}{\textsf{\small JV}\xspace}
\newcommand{\jvd}{\textsf{\small JV$(\delta)$}\xspace}
\newcommand{\jvp}[1]{\textsf{\small JV$(#1)$}\xspace}
\newcommand{\dualf}[1][\lambda]{\textsf{\small DUAL$(#1)$}\xspace}

\newcommand{\is}{\mathsf{IS}\xspace}

\newcommand{\dmean}{\ensuremath{\delta_{\textsf{\tiny mean}}}\xspace}
\newcommand{\rmean}{\ensuremath{\rho_{\textsf{\tiny mean}}}\xspace}
\newcommand{\dmed}{\ensuremath{\delta_{\textsf{\tiny med}}}\xspace}
\newcommand{\rmed}{\ensuremath{\rho_{\textsf{\tiny med}}}\xspace}
\newcommand{\sfac}{\ensuremath{\cFF_{\textsf{\tiny S}}}\xspace}
\newcommand{\sfact}[1]{\ensuremath{\sfac^{(#1)}}\xspace}
\newcommand{\dfac}{\ensuremath{\cFF_{\textsf{\tiny D}}}\xspace}

\newcommand{\dclient}{\ensuremath{\cDD_{\textsf{\tiny D}}}\xspace}

\newcommand{\sclient}{\ensuremath{\cDD_{\textsf{\tiny S}}}\xspace}
\newcommand{\sclientt}[1]{\ensuremath{\sclient^{(#1)}}}
\newcommand{\bclient}{\ensuremath{\cDD_{\textsf{\tiny B}}}\xspace}

\newcommand{\cneigh}[1][\gamma]{\ensuremath{N^{(0)}_{#1}}\xspace}

\newcommand{\clients}{\cD}
\newcommand{\facilities}{\cF}

\newcommand{\sweep}{{\normalfont\textsc{Sweep}}\xspace}

\newcommand{\graphupdate}{{\normalfont\textsc{GraphUpdate}}\xspace}
\newcommand{\raiseprice}{{\normalfont\textsc{RaisePrice}}\xspace}
\newcommand{\raiseprices}{{\normalfont\textsc{RaisePrices}}\xspace}
\newcommand{\quasisweep}{{\normalfont\textsc{QuasiSweep}}\xspace}
\newcommand{\quasigraphupdate}{{\normalfont\textsc{QuasiGraphUpdate}}\xspace}

\newcommand{\vfact}[1]{\cV^{(#1)}}
\newcommand{\vfac}{\cV}
\newcommand{\Gcf}{G}
\newcommand{\Gf}{H}
\newcommand{\Gcft}[1]{\Gcf^{(#1)}}
\newcommand{\Gft}[1]{\Gf^{(#1)}}

\newcommand{\timet}[1]{t^{(#1)}}
\newcommand{\stime}{\tau}
\newcommand{\stimet}[1]{\tau^{(#1)}}
\newcommand{\cst}[1]{\cS^{(#1)}}

\newcommand{\aeps}{\epsilon}
\newcommand{\const}[1]{C_{#1}}
\newcommand{\zeps}{\epsilon_z}
\newcommand{\alphat}[1]{\alpha^{(#1)}}
\newcommand{\Nt}[1]{N^{(#1)}}

\newcommand{\zt}[1]{z^{(#1)}}
\newcommand{\solt}[1]{(\alphat{#1}, \zt{#1})}
\newcommand{\rolt}[1]{(\alphat{#1}, \zt{#1},\sfact{#1}, \sclientt{#1} )}
\newcommand{\ist}[1]{\is^{(#1)}}
\newcommand{\sqrtalpha}{\bar{\alpha}}
\newcommand{\sqrtalphat}[1]{\sqrtalpha^{(#1)}}
\newcommand{\betat}[1]{\beta^{(#1)}}

\newcommand{\alphain}{\alpha^{\mbox{\tiny \textsf{in}}}}
\newcommand{\Nin}{N^{\mbox{\tiny \textsf{in}}}}
\newcommand{\alphaout}{\alpha^{\mbox{\tiny \textsf{out}}}}

\newif\ifshort
\shortfalse
\begin{document}
\title{Better Guarantees for $k$-Means and Euclidean $k$-Median by Primal-Dual Algorithms\thanks{Supported by ERC Starting Grant 335288-OptApprox.}} 

\author{Sara Ahmadian  \thanks{Department of Combinatorics and Optimization, University of Waterloo.
    Email:\textit{sahmadian@uwaterloo.ca}.}
\and
Ashkan Norouzi-Fard\thanks{School of Computer and Communication Sciences, EPFL.
Email: \textit{ashkan.norouzifard@epfl.ch}.
}
\and
Ola Svensson\thanks{School of Computer and Communication Sciences, EPFL.
Email: \textit{ola.svensson@epfl.ch}.
}
\and
Justin Ward\thanks{School of Computer and Communication Sciences, EPFL.
Email: \textit{{justin.ward@epfl.ch}}.
}}
\clearpage\maketitle
\thispagestyle{empty}

\begin{abstract}
  Clustering is a classic topic in optimization with $k$-means being one of the
  most fundamental such problems. In the absence of any restrictions on the
  input, the best known algorithm for $k$-means with a provable guarantee is
  a simple local search heuristic yielding an approximation guarantee of
  $9+\epsilon$, a ratio that  is known to be tight with respect to such methods.

  We  overcome this barrier by presenting a new primal-dual approach that
  allows us to (1) exploit the geometric structure of $k$-means and (2) to
  satisfy the hard constraint that at most $k$ clusters are selected without
  deteriorating the approximation guarantee. Our main result is
  a $6.357$-approximation algorithm with respect to the  standard LP relaxation.
  Our techniques are quite general and we also show improved guarantees for the
  general version of $k$-means where the underlying metric is not required to
  be Euclidean and for $k$-median in Euclidean metrics.

\end{abstract}
\clearpage
\setcounter{page}{1}

\section{Introduction}
\label{sec:intro}

Clustering problems have been extensively studied in computer science. They play a central role in many areas, including data science and machine learning, and their study has led to the development and refinement of several key techniques in algorithms and theoretical computer science. Perhaps the most widely considered clustering problem is the \emph{$k$-means} problem: given a set $\cD$ of $n$ points in $\mathbb{R}^{\ell}$ and an integer $k$, the task is to select a set $S$ of $k$ \emph{cluster centers} in $\mathbb{R}^\ell$, so that $\sum_{j \in \cD}c(j,S)$ is minimized, where $c(j,S)$ is the squared Euclidean distance between $j$ and its nearest center in~$S$.  

A commonly used heuristic for $k$-means is Lloyd's algorithm~\cite{Lloyd:2006:LSQ:2263356.2269955}, which is based on iterative improvement.  However, despite its ubiquity in practice, Lloyd's algorithm has, in general, no worst-case guarantee and may not even converge in polynomial time \cite{Arthur:2006:SKM:1137856.1137880,Vattani2011}.  To overcome some of these limitations, Arthur and Vassilvitskii \cite{Arthur:2007:KAC:1283383.1283494} proposed a randomized initialization procedure for Lloyd's algorithm, called $k$-means$++$, that leads to a $\Theta(\log k)$ expected approximation guarantee in the worst case.  Under additional assumptions about the \emph{clusterability} of the input dataset, Ostrovsky \etal~\cite{Ostrovsky2013} showed that this adaptation of Lloyd's algorithm gives a PTAS for $k$-means clustering. However, under no such assumptions, the best approximation algorithm in the general case has for some time remained a $(9+\epsilon)$-approximation algorithm based on local search, due to Kanungo \etal~\cite{DBLP:journals/comgeo/KanungoMNPSW04111}. Their analysis also shows that no natural local search algorithm performing a fixed number of swaps can improve upon this ratio.  This leads to a barrier for these techniques that are rather far away from  the best-known inapproximability result which  only says that it is NP-hard to approximate this problem to within a factor better than $1.0013$~\cite{DBLP:journals/corr/LeeSW15a}.

While the general problem has resisted improvements, there has been significant progress on the
$k$-means problem  under a variety of assumptions.  For example, Awasthi, Blum,
and Sheffet obtain a PTAS in the special case when the instance has certain
stability properties \cite{Awasthi:2010:SYP:1917827.1918395} (see also \cite{Balcan2009}), and there has
been a long line of work (beginning with \cite{matouvsek2000approximate})
obtaining better and better PTASes under the assumption that $k$ is constant.
Most recently, it has been shown that local search gives a PTAS under the
assumption that the dimension $\ell$ of the dataset is constant
\cite{DBLP:journals/corr/Cohen-AddadKM16,DBLP:journals/corr/FriggstadRS16}.
These last results generalize to the case in which the squared distances are
from the shortest path metric on a graph with forbidden minors
\cite{DBLP:journals/corr/Cohen-AddadKM16} or from a metric with constant
doubling dimension \cite{DBLP:journals/corr/FriggstadRS16}. We remark that the
dimension $\ell$ of a $k$-means instance may always be assumed to be at most
$O(\log n)$ by a standard application of the Johnson-Lindenstrauss transform. But, as the results in~\cite{DBLP:journals/corr/Cohen-AddadKM16,DBLP:journals/corr/FriggstadRS16} exhibit doubly-exponential dependence on the dimension, they do not give any non-trivial implications for the general case. Moreover,   
such a doubly-exponential dependence is essentially unavoidable, as the problem is APX-hard in the general case~\cite{DBLP:conf/compgeom/AwasthiCKS15}.

In summary, while $k$-means is perhaps the most widely used clustering problem
in computer science, the only constant-factor approximation algorithm for the
general case is based on simple local search heuristics that, for inherent
reasons, give guarantees that are rather far from known hardness results.
This is in contrast to many other well-studied clustering problems, such
as facility location and $k$-median.  Over the past several decades, a toolbox of core algorithmic techniques such as dual fitting, primal-dual and LP-rounding, has been refined and applied to these problems leading to improved approximation guarantees~\cite{STA97,CS04,
BA10,DBLP:journals/iandc/Li13,LV92B,DBLP:journals/jacm/JainV01,Jain:2002:NGA:509907.510012,LiS16,JMM03}.
In particular, the current best approximation guarantees for both facility location (a
1.488-approximation due to Li~\cite{DBLP:journals/iandc/Li13}) and $k$-median
(a 2.675-approximation due to Byrka \etal~\cite{DBLP:conf/soda/ByrkaPRST15}) are LP-based and 
give significantly better results than
previous local search algorithms \cite{AGKMMP,CG}.  However, such LP-based techniques
have not yet been able to attain similar improvements for $k$-means.  One
reason for this is that they have relied heavily on the triangle inequality,
which does not hold in the case of $k$-means.  

\paragraph{Our results.}
In this work, we overcome this barrier by developing new techniques that allow us to exploit the standard LP formulation for $k$-means.  We significantly narrow the gap between known upper and lower bounds by designing a new primal-dual algorithm for the $k$-means problem.  We stress that our algorithm works in the general case that $k$ and $\ell$ are part of the input, and requires no assumptions on the dataset.
\begin{theorem}
  \label{thm:maineuckmeans}
  For any $\epsilon > 0$, there is a $(\rmean+\epsilon)$-approximation algorithm for the $k$-means problem, where $\rmean \approx 6.357$. Moreover, the integrality gap of the standard LP is at most $\rmean$.
\end{theorem}

We now describe our approach and contributions at a high level.  Given a $k$-means instance, we apply standard discretization techniques (\eg \cite{DBLP:conf/compgeom/FeldmanMS07}) to obtain an instance of the \emph{discrete $k$-means} problem, in which we are given a discrete set $\cF$ of candidate centers in $\mathbb{R}^\ell$ and must select $k$ centers from $\cF$, rather than $k$ arbitrary points in $\mathbb{R}^\ell$.  This step incurs an arbitrarily small loss in the approximation guarantee with respect to the original $k$-means instance.  Because our algorithm always returns a set of centers from the discrete set $\cF$, all of our results also hold for the \emph{exemplar clustering} problem, in which centers must be chosen from the input points in $\clients$.  Specifically, we can simply take $\cF = \clients$.

Using Lagrangian relaxation, we can then consider the resulting discrete problem using the standard linear programming formulation for facility location.  This general approach was pioneered in this context by Jain and
Vazirani~\cite{DBLP:journals/jacm/JainV01} who gave primal-dual algorithms for the $k$-median problem. In their paper, they first present
a \emph{Lagrangian Multiplier Preserving} (\lmp) $3$-approximation algorithm
for the facility location problem. 
Then they run binary search over the opening cost of the facilities and use the aforementioned algorithm to get two
solutions: one that opens more than $k$ facilities and one that opens fewer than
$k$, such that the opening cost of facilities in these solutions are close to each other. These
solutions are then combined to obtain a solution that opens exactly $k$
facilities. This step results in  losing another factor $2$ in the
approximation guarantee, which results in a $6$-approximation algorithm for 
$k$-median. The factor $6$ was later improved by Jain, Mahdian, and Saberi~\cite{Jain:2002:NGA:509907.510012}  who obtained
a $4$-approximation algorithm for $k$-median by developing an $\lmp$
$2$-approximation algorithm for facility location. 

\paragraph{Technical contributions.}
One can see that the same approach gives a much larger constant factor for the $k$-means problem since one cannot anymore rely on the triangle
inequality.  We use two main ideas to overcome this obstacle: (1) we exploit the geometric structure of $k$-means to obtain an improved $\lmp$-approximation, and (2) we develop a new primal-dual algorithm that opens exactly $k$ facilities while losing only an \emph{arbitrarily small factor}.

For our first contribution, we modify the primal-dual algorithm of Jain and Vazirani~\cite{DBLP:journals/jacm/JainV01} into a parameterized version which allows us to  regulate the ``aggressiveness'' of the opening strategy of facilities. By using properties of Euclidean metrics we show that this leads to improved \lmp approximation algorithms for $k$-means. 

By the virtue of~\cite{DBLP:conf/esa/ArcherRS03}, these results  already imply upper bounds on the integrality gaps of the standard LP relaxations, albeit with an \emph{exponential time} rounding algorithm.  Our second and more technical contribution is a new primal-dual algorithm that accomplishes the same task in polynomial time.  In other words, we are able to turn an \lmp approximation algorithm into an algorithm that opens at most $k$ facilities without deteriorating the approximation guarantee. We believe that this contribution is of independent interest. Indeed, all recent progress on the approximation of $k$-median beyond long-standing local search approaches \cite{AGKMMP} has involved reducing the factor $2$ that is lost by Jain and Vazirani when two solutions are combined to open exactly $k$ facilities (i.e. in the rounding of a so-called \emph{bipoint} solution) \cite{LiS16,DBLP:conf/soda/ByrkaPRST15}.  Here, we show that it is possible to reduce this loss all the way to $(1 + \epsilon)$ by fundamentally changing the way in which dual solutions are constructed and maintained.  

Instead of finding two solutions by binary search as in the
framework of~\cite{DBLP:journals/jacm/JainV01}, we consider a sequence of solutions such that the opening costs and also the dual values of any two consecutive solutions are close in $L^{\infty}$-norm.  We show that this latter property allows us to combine two appropriate, consecutive solutions in the sequence into a single solution that opens exactly $k$ facilities while losing only a factor of $1+\epsilon$ (rather than 2) in the approximation guarantee.  Unfortunately, the dual solutions produced by the standard primal-dual algorithm approach are unstable, in the sense that a small change in opening price may result in drastic changes in the value of the dual variables.  Thus, we introduce a new primal-dual procedure which instead iteratively transforms a dual solution for one price into a dual solution for another price.  By carefully constraining the way in which the dual variables are altered, we show that we can obtain a sequence of ``close'' solutions that can be combined as desired.

We believe that this technique may be applicable in other settings, as well.  An especially interesting open question is whether it is possible combine stronger \lmp approximation algorithms, such as the one by Jain, Mahdian, Saberi~\cite{Jain:2002:NGA:509907.510012},  with our lossless rounding to obtain an improved $(2+\epsilon)$-approximation algorithm for $k$-median.

\paragraph{Extensions to other problems.}
In addition to the standard $k$-means problem, we show that our results also extend to the following two problems.  In the first extension, we consider the Euclidean $k$-median problem.  Here we are given a set $\cD$ of $n$ points in $\mathbb{R}^\ell$ and a set $\cF$ of $m$ points in $\mathbb{R}^\ell$ corresponding to facilities.  The task is to select a set $S$ of at most $k$ facilities from $\cF$ so as to minimize $\sum_{j \in \cD}c(j,S)$, where $c(j,S)$ is now the (non-squared) Euclidean distance from $j$ to its nearest facility in $S$.  For this problem, no approximation better than the general $2.675$-approximation algorithm of Byrka \etal~\cite{DBLP:conf/soda/ByrkaPRST15} for $k$-median was known.
\begin{theorem}
  \label{thm:kmedian}
  For any $\epsilon>0$, there is a $(\rmed+\epsilon)$-approximation algorithm for the Euclidean $k$-median problem, where $\rmed \approx 2.633$. Moreover, the integrality gap of the standard LP is at most $\rmed$.
\end{theorem}

In the second extension, we consider a variant of the $k$-means problem in which each $c(j,S)$ corresponds to the squared distance in an arbitrary (possibly non-Euclidean) metric on $\cD \cup \cF$.  For this problem, the best-known approximation algorithm is a 16-approximation due to Gupta and Tangwongsan~\cite{DBLP:journals/corr/abs-0809-2554}.  In this paper, we obtain the following improvement:
\begin{theorem}
  \label{thm:kmeansgeneralmetrics}
  For any $\epsilon>0$, there is a  $(9+\epsilon)$-approximation algorithm for the $k$-means problem in  general metrics. Moreover, the integrality gap of the standard LP is at most $9$.
\end{theorem}
\noindent We remark that the same hardness reduction as used for $k$-median~\cite{Jain:2002:NGA:509907.510012} immediately yields a much stronger hardness result for the above generalization than what is known for the standard $k$-means problem: it is hard to approximate the $k$-means problem in general metrics within a factor $1+8/e - \epsilon \approx 3.94$ for any $\epsilon >0$.

\paragraph{Outline of paper.}
In Section~\ref{sec:prelim} we review the standard LP formulation that we use, as well as its Lagrangian relaxation.  We then in Section~\ref{PDalg} show how to exploit the geometric structure of $k$-means and Euclidean $k$-median to give improved \lmp guarantees.  In Section~\ref{sec:quasi} we show the main ideas behind our new rounding approach by giving an algorithm that runs in quasi-polynomial time. These results are then generalized in Sections~\ref{sec:polyn-time-appr},~\ref{sec:rounding}, and~\ref{sec:polyn-time-algor} to obtain an algorithm that runs in polynomial time.

\section{The standard LP relaxation  and its Lagrangian relaxation} 
\label{sec:prelim}
Here and in the remainder of the paper, we shall consider the
\emph{discrete} $k$-means problem, where we are given a discrete set
$\cF$ of facilities (corresponding to candidate centers).\footnote{As
discussed in the introduction, it is well-known that a
$\rho$-approximation algorithm for this case can be turned into a
$(\rho+\epsilon)$-approximation algorithm for the standard $k$-means
problem, for any constant
$\epsilon > 0$ (see \eg~\cite{DBLP:conf/compgeom/FeldmanMS07}).}  Henceforth, we
will simply refer to the discrete $k$-means problem as the $k$-means
problem.

Given an instance $(\cD, \cF, d, k)$ of the $k$-means problem or the $k$-median
problem, let $c(j,i)$ denote the connection cost of client $j$ if connected to
facility $i$.  That is, $c(j,i) = d(j,i)$ in the case of $k$-median and $c(j,i)
= d(j,i)^2$ in the case of $k$-means.  Let $n = |\cD|$ and $m = |\cF|$.

The standard linear programming (LP) relaxation of these problems has two sets of variables:
a variable $y_i$ for each facility $i\in \cF$ and a variable $x_{ij}$
for each facility-client pair $i\in \cF, j\in \cD$.  The intuition of these
variables is that $y_i$ should  indicate whether facility $i$ is opened
and $x_{ij}$ should indicate whether client $j$ is connected to facility $i$. The
standard LP relaxation can now be formulated as follows.
\ifshort
\begin{center}
  \begin{minipage}[t]{0.5\textwidth}
    \begin{alignat}{2}
     \textstyle \min \  \textstyle \sum_{i\in \cFF, j\in \cDD}&\   x_{ij} \cdot c(j,i) & \notag  \\
      \text{s.t. }  \textstyle \sum_{i\in \cFF} x_{ij} &\textstyle \geq  1 
      \quad \ \ \forall j \in \cDD \label{con:assclient}\\
      \textstyle x_{ij} &\textstyle\leq y_i \quad \ \ \forall j \in \cDD, i\in \cFF \label{con:assopen}     \end{alignat} 
  \end{minipage}
  \begin{minipage}[t]{0.25\textwidth}
    \begin{alignat}{2}
      \  \notag \\
                     \textstyle \sum_{i \in \cFF} y_{i} &\textstyle \le k \label{con:kbnd} \\
      \textstyle x,y & \geq 0\,. \label{con:nonneg}
    \end{alignat} 
  \end{minipage}
\end{center}
\else
\begin{center}
  \begin{minipage}[t]{0.55\textwidth}
    \begin{alignat}{2}
      \min & \quad & \sum_{i\in \cFF, j\in \cDD} x_{ij} \cdot c(j,i) & \notag  \\
      \text{s.t.} && \sum_{i\in \cFF} x_{ij} & \geq  1 
      \quad \ \ \forall j \in \cDD \label{con:assclient}\\
      && x_{ij} &\leq y_i \quad \ \ \forall j \in \cDD, i\in \cFF \label{con:assopen} \\[2mm]
      && \sum_{i \in \cFF} y_{i} & \le k \label{con:kbnd} \\
      && x,y & \geq 0\,. \label{con:nonneg}
    \end{alignat} 
  \end{minipage}
\end{center}
\fi
The first set of constraints says that each client should be connected to at least one
facility; the second set of constraints enforces that clients can only be connected to opened
facilities; and the third constraint says that at most $k$ facilities can be opened. 
We remark that this is a relaxation of the original problem as we have relaxed
the constraint that $x$ and $y$ should take Boolean values to 
a non-negativity constraint. For future reference, we let $\opt_k$ denote the
value of an optimal solution to this relaxation.

A main difficulty for approximating the $k$-median and the $k$-means  problems
is the hard constraint that at most $k$ facilities can be selected, i.e.,
constraint~\eqref{con:kbnd} in the above relaxation.  A popular way of
overcoming this difficulty, pioneered in this context by Jain and
Vazirani~\cite{DBLP:journals/jacm/JainV01}, is to consider the Lagrangian relaxation where we multiply the 
constraint~\eqref{con:kbnd} times a Lagrange multiplier $\lambda$ and move it to the
objective.  This results, for every $\lambda\geq 0$, in the following
relaxation and its dual that we denote by \lpf and \dualf, respectively.  
\begin{center}
\hspace*{-2ex}
\fbox{\begin{minipage}[t][3.5cm]{0.45\textwidth}
  \begin{center}\lpf
\begin{alignat}{2}
  \min & \quad & \sum_{i \in \cFF, j\in \cDD} x_{ij} \cdot c(j,i) + \lambda\cdot\left(\sum_{i\in \cFF} y_i - k\right) & \notag \\[4mm]
  \text{s.t.} && \mbox{\eqref{con:assclient},~\eqref{con:assopen}, and~\eqref{con:nonneg}.} \qquad \qquad \qquad &    \notag
\end{alignat} \end{center}
\end{minipage}}
\fbox{\begin{minipage}[t][3.5cm]{0.5\textwidth}
  \begin{center} \dualf
\begin{alignat}{2}
\max & \quad & \sum_{j\in \cDD} \alpha_j - \lambda \cdot k & \notag \\
\text{s.t.} && \sum_{j\in \cDD} [\alpha_j - c(j,i)]^+& \leq  \lambda \qquad \forall i\in\cFF \label{eq:facopen} \\
&& \alpha & \geq 0 \notag.
\end{alignat}
\end{center}
\end{minipage}}
\end{center}
Here, we have simplified the dual by noticing that the dual variables
$\{\beta_{ij}\}_{i\in \cFF, j\in \cDD}$ corresponding to the
constraints~\eqref{con:assopen} of the primal can always be set  $\beta_{ij}
= [\alpha_j - c(j,i)]^+$; the notation $[a]^+$ denotes $\max(a,0)$.
 Moreover, to see that \lpf remains a relaxation, note that
any feasible solution to the original LP is a feasible solution to  the
Lagrangian relaxation of no higher cost. In other words, for any $\lambda \geq 0$, the optimum value of \lpf is at most $\opt_k$. 

If we disregard the constant term $\lambda\cdot k$ in the objective functions, \lpf
and \dualf become the standard LP formulation and its dual for the facility
location problem where the opening cost  of each facility equals $\lambda$ and the connection costs are defined by $c(\cdot, \cdot)$.
Recall that the facility location problem (with uniform opening costs) is defined as the  problem of selecting a set $S \subseteq \cF$ of facilities to open so as to minimize the opening cost $|S| \lambda$ plus  the connection cost $\sum_{j\in \cDD} c(j, S)$. Jain and
Vazirani~\cite{DBLP:journals/jacm/JainV01} introduced the following method for addressing the $k$-median
problem motivated by simple economics. On the one hand,  if $\lambda$ is
selected to be very small, i.e., it is cheap to open facilities, then a good
algorithm for the facility location problem will open many facilities. On the
other hand, if $\lambda$ is selected to be very large, then a good algorithm
for the facility location problem will open few facilities. Ideally, we want to
use this intuition to find an opening price that leads to the opening of
exactly $k$ facilities and thus a solution to the original, constrained problem.

To make this intuition work, we need the notion of \emph{Lagrangian Multiplier
Preserving} (\lmp)  approximations: we say that a $\rho$-approximation
algorithm is \lmp for the facility location problem with opening costs $\lambda$ if it returns a solution $S\subseteq \cFF$ satisfying
\begin{align*}
  \sum_{j\in \cDD} c(j, S) \leq \rho (\lpopt(\lambda) - |S| \lambda)\,,
\end{align*}
where $\lpopt(\lambda)$ denotes the value of an optimal solution to \lpf without the constant
term $\lambda \cdot k$. The  importance of
this definition becomes apparent when either $\lambda = 0$ or $|S| \leq k$. In those cases, we
can see that the value of the $k$-median or $k$-means solution is at most
$\rho$ times the optimal value of its relaxation \lpf, and thus an $\rho$-approximation with respect to its standard LP relaxation  since $\lpopt(\lambda)
- k\cdot \lambda \leq \opt_k$ for any $\lambda \geq 0$.  For further explanations and applications of this technique, we refer the reader to the excellent text books~\cite{Vazirani:2001:AA:500776} and~\cite{Williamson:2011:DAA:1971947}.

\section{Exploiting Euclidean metrics via primal-dual algorithms}\label{PDalg}

In this section we show how to exploit the structure of Euclidean metrics to
achieve better approximation guarantees. Our \lmp approximation algorithm builds upon the 
primal-dual algorithm for the facility location problem by Jain and Vazirani~\cite{DBLP:journals/jacm/JainV01}.
We refer to their algorithm as the \jv algorithm.  The main modification to
their algorithm is that we allow for a more ``aggressive''  opening strategy of
facilities. The amount of aggressiveness is measured by the parameter $\delta$: we devise an algorithm \jvd  for  each
parameter $\delta \geq 0$, where a smaller $\delta$ results in a more aggressive opening strategy. 
 We first describe \jvd and 
we then optimize $\delta$ for the considered objectives to obtain the claimed approximation guarantees. 

We remark that the result in~\cite{DBLP:conf/esa/ArcherRS03} (non-constructively) upper bounds  the integrality gap of the
standard LP relaxation of $k$-median in terms of the \lmp approximation guarantee of $\jv$. This
readily generalizes to the $k$-means problem and $\jvd$. Consequently, our guarantees
presented here upper bound the integrality gaps as the theorems state in
the introduction.
\ifshort \vspace{-2mm} \fi
\subsection{Description of \jvd}
\label{sec:pddesc}
As alluded to above, the algorithm is a modification of \jv, and Remark~\ref{remark:difference} below highlights the difference. The algorithm consists of two phases:
the dual-growth phase and the pruning phase. 
\ifshort \vspace{-4mm} \fi
\paragraph{Dual-growth phase:} In this stage, we  construct a feasible dual
solution $\alpha$ to \dualf. Initially, we set $\alpha = \vec{0}$ and let $A
= \cD$ denote the set of active clients (which is all clients at first).  We
then repeat the following until there are no active clients, i.e.,  $A
= \emptyset$: increase the dual-variables $\{\alpha_j\}_{j\in A}$ corresponding
to the active clients at a uniform rate until one of the following events occur (if
several events happen at the same time, break ties arbitrarily):
\begin{description}
  \item[\textnormal{Event 1:}] A dual constraint $\sum_{j\in \cD} [\alpha_j - c(j,i)]^+ \leq
    \lambda$ becomes tight for a facility $i\in \cF$. In this case we say
    that facility $i$ is \emph{tight} or \emph{temporarily opened}. We
    update $A$ by removing the active clients with a \emph{tight edge} to
    $i$, that is, a client $j\in A$ is removed if $\alpha_j - c(j,i) \geq
    0$.  For future reference, we say that facility $i$ is the \emph{witness} of these removed
    clients.   \item[\textnormal{Event 2:}] An active client $j\in A$ gets a tight edge, i.e., $\alpha_j
    - c(j,i) = 0$, to some  already tight facility $i$. In this case, we remove
    $j$ from $A$ and let $i$ be its witness. \end{description}
This completes the description of the dual-growth phase. Before proceeding to
the pruning phase, let us remark that the constructed $\alpha$ is indeed
a feasible solution to \dualf by design. It is clear that $\alpha$ is
non-negative. Now consider a facility $i\in \cFF$ and its corresponding dual
constraint  $\sum_{j\in \cDD} [\alpha_j - c(j,i)]^+ \leq \lambda$. On the one hand, the
constraint is clearly satisfied if it never becomes tight during the
dual-growth phase.  On other hand, if it becomes tight, then all clients with
a tight edge to it are removed from the active set of clients by
Event~1. Moreover, if any client gets a tight edge to $i$ in
subsequent iterations it gets immediately removed from the set of active
clients by Event~2. Therefore the left-hand-side of the constraint
will never increase (nor decrease) after it becomes tight so the constraint
remains satisfied. Having proved that $\alpha$ is a feasible solution to \dualf, let us now describe the pruning phase.

\ifshort \vspace{-3mm} \fi
\paragraph{Pruning phase:} After the dual-growth phase (too) many facilities
are temporarily opened. The pruning phase will select a subset of these
facilities to open. In order to formally describe this process, 
 we need the following notation. For a client $j$, let $N(j) = \{i \in \cFF: \alpha_j - c(j,i) > 0\}$ denote the facilities to which client $j$ contributes to the opening cost. Similarly, for $i\in \cFF$, let $N(i) = \{j\in \cDD: \alpha_j - c(j,i) > 0\}$ denote the clients with a positive contribution toward $i$'s opening cost. 
For a temporarily opened facility $i$, let 
\ifshort \vspace{-3mm} \fi
\begin{align*}
  t_i =  \max_{j\in N(i)} \alpha_j\,,
\end{align*}
and by convention let $t_i = 0$ if $N(i) = \emptyset$ (this convention will be useful in future sections and will only be used when the opening cost $\lambda$ of facilities are set to $0$). 
Note that, if $N(i) \neq \emptyset$, then  $t_i$ equals the ``time'' that
facility $i$ was temporarily opened in the dual-growth phase. A crucial property of $t_i$ that follows from the construction of $\alpha$ is the following.
\begin{claim}
  \label{claim:twitness}
  For a client $j$ and its witness $i$, $\alpha_j \geq t_i$. Moreover, for any $j'\in N(i)$ we have $t_i \geq \alpha_{j'}$.
\end{claim}
For the pruning phase, it will be convenient to define the \emph{client-facility graph} $\Gcf$ and the \emph{conflict graph} $\Gf$. The vertex set of $\Gcf$ consist of all the clients and  all the facilities $i$ such that $\sum_{j\in \clients} [\alpha_{j} - c(j,i)]^+ = \lambda$ (\ie the tight or  temporarily open facilities). There is an edge between facility $i$ and client $j$ if $i \in N(j)$. The conflict graph $\Gf$ is defined based on the client-facility graph $\Gcf$ and $t$ as follows:
\ifshort \vspace{-2mm} \fi
\begin{itemize}\ifshort \itemsep0mm \fi
  \item The vertex set consists of all facilities in $\Gcf$.
  \item There is an edge between two facilities $i$ and $i'$ if some client $j$ is adjacent to both of them  in $\Gcf$ and $c(i,i') \leq \delta \min(t_i, t_{i'})$.
\end{itemize}
\ifshort \vspace{-2mm} \fi
The pruning phase now finds a (inclusion-wise) maximal independent set $\is$ of $\Gf$ and opens those 
facilities; clients are connected to the closest
facility in $\is$. 

\begin{remark}
  \label{remark:difference}
  The difference between the original algorithm \jv and our modified \jvd
  is the additional condition $c(i,i') \leq \delta \min(t_i, t_{i'})$ in the
  definition of the conflict graph. Notice that if we select a smaller $\delta$
  we will have fewer edges in $\Gf$. Therefore a maximal independent set will likely
  grow in size, which results in a more ``aggressive'' opening strategy. Adjusting $\delta$ will allow us to achieve better \lmp approximation guarantees.
\end{remark}

\subsection{Analysis of \jvd for the considered objectives}

In the following subsections, we optimize $\delta$ and analyze the guarantees
obtained by the algorithm $\jvd$  for the
objective functions: k-means objective in general metrics, standard $k$-means objective (in Euclidean metrics),   and k-median
objective in Euclidean metrics.  The first analysis is very similar to the
original \jv analysis and also serves as a motivation for the possible
improvements in Euclidean metrics.
\subsubsection{$k$-Means objective in general metrics}
We consider the case when $c(j,i) = d(j,i)^2$ and $d$ forms a general metric.
We let $\delta= \infty$ so \jvd becomes simply the \jv algorithm. We prove the following. \begin{theorem}
Let $d$ be any metric on $\clients \cup \facilities$ and suppose that $c(j,i) = d(j,i)^2$ for every $i \in \facilities$ and $j \in \clients$.  Then, for any $\lambda \geq 0$, Algorithm \jvp{\infty} constructs a solution $\alpha$ to \dualf and returns a set $\is$ of opened  facilities such that
  \begin{align*}
    \sum_{j\in \clients} c(j, \is)  \leq 9 \cdot  (\sum_{j\in \cDD} \alpha_j - \lambda |\is|)\,.
  \end{align*}
  \label{thm:genkmeans}
\end{theorem}
\begin{proof}
  Consider any client $j \in \clients$. We shall prove that 
  \begin{align}
    \label{eq:genkmeans}
\frac{c(j, \is)}{9}   \leq \alpha_j - \sum_{i\in N(j)\cap \is} (\alpha_j - c(j,i))  =   \alpha_j - \sum_{i\in \is} [\alpha_j - c(j,i)]^+\,.
  \end{align}
  The statement then follows by summing up over all clients and noting that
  any facility $i\in \is$ was temporarily opened and thus we have $\sum_{j\in
  \cDD} [\alpha_j - c(j,i)]^+ = \lambda$.  
  
  To prove~\eqref{eq:genkmeans}, we first note that $|\is \cap N(j)| \leq 1$.
  Indeed, consider $i \neq i'  \in N(j)$. Then $(j,i)$ and $(j,i')$ are edges in the client-facility graph $\Gcf$ 
  and as $\delta = \infty$, $i$ and $i'$ are adjacent in the conflict graph
  $\Gf$. Hence, the temporarily opened facilities in $N(j)$ form
  a clique in $H$ and at most one of them can be selected in the maximal independent
  set $\is$.  We complete the analysis by considering the two cases $|\is \cap
  N(j)| = 1$ and $|\is \cap N(j)| = 0$.
  \begin{description}
    \item[\textnormal{Case $|\is \cap N(j)| = 1$:}] Let $i^*$ be the unique facility in $\is \cap N(j)$.  Then
      \begin{align*}
        \frac{c(j,\is)}{9} \leq c(j,\is) \leq c(j, i^*) = \alpha_j  - (\alpha_j - c(j,i^*)) = \alpha_j - \sum_{i\in N(j) \cap \is}(\alpha_j - c(j,i)).
      \end{align*}
      Notice the amount of slack in the above analysis (specifically, the first
      inequality).  In the Euclidean case, we exploit this slack for a more aggressive opening and to improve the
      approximation guarantee. 
    \item[\textnormal{Case $|\is \cap N(j)| = 0$:}]  Let $i_1$ be $j$'s
      witness.  First, if $i_1 \in \is$ then by the same arguments as above we have
      the desired inequality; specifically, since $j$ has a tight edge to $i_1$ but $i_1 \not\in N(j)$ we must have $\alpha_j = c(j, i_1)$.  Now consider the more interesting case when $i_1
      \not \in \is$. As $\is$ is a maximal independent set in $\Gf$, there
      must be a facility $i_2 \in \is$ that is adjacent to $i_1$ in $\Gf$.
      By definition of $H$, there is a client $j_1$ such that $(j_1, i_1)$ and $(j_1, i_2)$ are edges in the client-facility graph $\Gcf$, \ie  $j_1 \in N(i_1)
      \cap N(i_2)$.  By the definition of witness and $N(\cdot)$, we have  
      \begin{align*}
        \alpha_{j} \geq c(j,i_1,)\,, \qquad \alpha_{j_1} > c(j_1,i_1)\,,  \qquad \alpha_{j_1} > c(j_1,i_2)\,,
      \end{align*}
      and by the description of the algorithm (see Claim~\ref{claim:twitness} in Section~\ref{PDalg}) we have $\alpha_{j} \geq t_{i_1} \geq \alpha_{j_1}$. 
      Hence, using the triangle inequality and that $(a+b+c)^2 \leq 3 (a^2 + b^2 + c^2)$,
      \begin{align*}
        c(j,\is) &\leq c(j, i_2)  = d(j,i_2)^2  \leq (d(j,i_1) + d(j_1,i_1) + d(j_1, i_2))^2    \\
        & \leq 3(d(j,i_1)^2 + d(j_1,i_1)^2 + d(j_1, i_2)^2) \\
        & = 3(c(j,i_1) + c(j_1, i_1) + c(j_1, i_2)) \leq 9 \alpha_j\,.
      \end{align*}
      As $\sum_{i\in N(j) \cap \is}(\alpha_j - c(j,i)) = 0$, this completes the proof of this case   and thus the theorem.
      
  \end{description}

\end{proof}
\subsubsection{$k$-Means objective in Euclidean metrics}
\label{sec:euckmeans}
\begin{figure}[t!hb]
\centering
    \begin{tikzpicture}[scale=0.7]
  \begin{scope}
  \draw[fill=black!5!white] (0,0) circle (1.2cm);
  \draw[fill=black] (0,0) circle (3pt) node[below right] {\small $j$};
  \draw[dashed] (0,0) -- node[above] {\small $1$} (-1.2,0);
  \draw[fill=black!5!white] (-2.4,0) circle (1.2cm);
  \draw[fill=black] (-2.4,0) circle (3pt) node[below right] {\small $j_1$};
    \draw[dashed] (-2.4,0) -- node[above] {\small $1$} (-1.2,0);
  \draw[dashed] (-2.4,0) -- node[above] {\small $1$} (-3.6,0);
  \draw[fill=white] (-1.3,-0.1) rectangle (-1.1, 0.1) node[below right = 0.1cm and -0.05cm] {\small $i_1$};
  \draw[fill=black] (-3.7,-0.1) rectangle (-3.5, 0.1) node[below right = 0.1cm and -0.05cm] {\small $i_2$};
  \node at (-1.3, 2.5) {\begin{minipage}{6.5cm}\small \vspace{-6mm} \begin{center}\textbf{Worst case configuration} \end{center} \vspace{-2mm} 
\footnotesize     The clients and the facilities are arranged on a line and we have $c(i_2,j) = d(i_2,j)^2 =  9 \alpha_j$.\end{minipage}};
  \draw (-6.5, -1.7) rectangle (3.9, 4);
\end{scope}
  \begin{scope}[xshift=11cm]
  \draw[fill=black!5!white] (0,0) circle (1.2cm);
  \draw[fill=black] (0,0) circle (3pt) node[below right] {\small $j$};
  \draw[dashed] (0,0) -- node[above] {\small $1$} (-1.2,0);
  \draw[fill=black!5!white] (-2.4,0) circle (1.2cm);
  \draw[fill=black] (-2.4,0) circle (3pt) node[below right] {\small $j_1$};
    \draw[dashed] (-2.4,0) -- node[above] {\small $1$} (-1.2,0);
  \draw[dashed,rotate around={25:(-2.4,0)}] (-2.4,0) -- node[above] {\small $1$} (-3.6,0);
  \draw[fill=white] (-1.3,-0.1) rectangle (-1.1, 0.1) node[below right = 0.1cm and -0.05cm] {\small $i_1$};
  \draw[fill=black, rotate around={25:(-2.4,0)}] (-3.7,-0.1) rectangle (-3.5, 0.1) node[below right = 0.00cm and 0.05cm] {\small $i_2$};
  \node at (-1.3, 2.5) {\begin{minipage}{6.5cm}\small \vspace{-6mm} \begin{center}\textbf{Better
      case in Euclidean space} \end{center} \vspace{-2mm} \footnotesize The distance $d(j, i_2)$ is
    better than that the triangle inequality gives yielding a better
    bound.\end{minipage}};
  \draw (-6.5, -1.7) rectangle (3.9, 4);
\end{scope}
\end{tikzpicture}
   \caption{The intuition how we improve the guarantee in the Euclidean case. In both cases, we have $\alpha_j = \alpha_{j_1} = 1$. Moreover, $i_1\not \in \is, i_2 \in \is$ and we are interested in bounding $c(j,i_2)$ as a function of $\alpha_j$.}
  \label{fig:euc}
\end{figure}
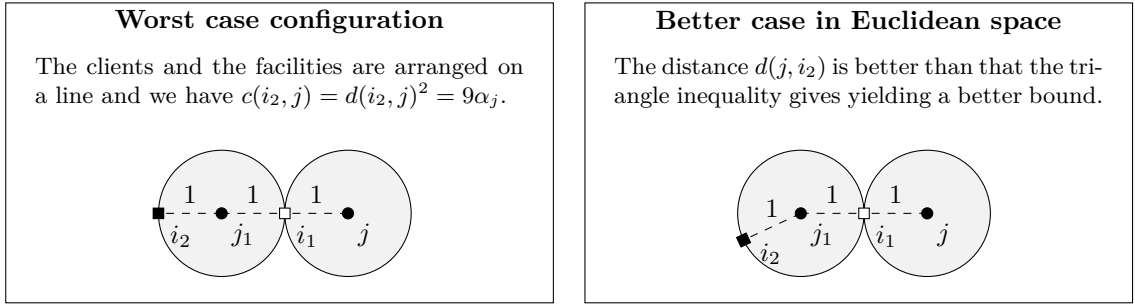
We start with some intuition that illustrates our approach.  From the standard  analysis of $\jv$ (and our analysis of $k$-means in general metrics), it is clear that the bottleneck for the approximation guarantee 
comes from the connection-cost analysis of clients that need to do a ``$3$-hop'' as illustrated in the left part of Figure~\ref{fig:euc}: client $j$ is connected to open facility $i_2$ and the squared-distance is bounded by the path $j-i_1-j_1-i_2$. 
Moreover, this analysis is tight when considering $\jv= \jvp{\infty}$.
Our strategy will now be as follows: Select $\delta$ to be a constant smaller than $4$. This means that in the configurations of Figure~\ref{fig:euc}, we will also open $i_2$ if the distance between $i_1$ and $i_2$ is close to $2$. Therefore, if we do not open $i_2$, the distance between $i_1$ and $i_2$ is less than $2$ (as in the right part of Figure~\ref{fig:euc}) which allows us to get an approximation guarantee better than $9$. However, this might result in a client contributing to the opening cost of many facilities in $\is$. Nonetheless, by using the properties of Euclidean metrics, we show that even in this case, we are able to achieve an $\lmp$ approximation guarantee with ratio better than $9$. 

Specifically, define $\dmean$ to be the constant larger than $2$ that minimizes  
  \begin{align*}
    \rmean (\delta) = \max\left\{ (1+\sqrt{\delta})^2, \frac{1}{\delta/2-1}\right\}\,, 
  \end{align*}
which will be our approximation guarantee. It can be verified that  $\dmean
\approx 2.3146$ and $\rmean \approx 6.3574$. Let also $c(j,i) = d(j,i)^2$ where $d$ is the underlying Euclidean metric. The proof uses the following basic facts about squared-distances in Euclidean metrics: given $x_1, x_2, \dots, x_s \in \mathbb{R}^\ell$, we have that $\min_{y \in \mathbb{R}^\ell} \sum_{i=1}^s \|x_i - y\|_2^2$ is attained by the \emph{centroid} $\mu = \tfrac{1}{s} \sum_{i=1}^s x_i$ and in addition we have the identity $\sum_{i=1}^s \|x_i - \mu\|_2^2 = \tfrac{1}{2s} \sum_{i=1}^s \sum_{j=1}^s \|x_i - x_j\|_2^2$.
\begin{theorem}
Let $d$ be a Euclidean metric on $\clients \cup \facilities$ and suppose that $c(j,i) = d(j,i)^2$  for every $i \in \facilities$ and $j \in \clients$.  Then, for any $\lambda \geq 0$, Algorithm \jvp{\dmean} constructs a solution $\alpha$ to \dualf and
  returns a set $\is$ of opened  facilities such that
  \begin{align*}
    \sum_{j\in \clients} c(j, \is)  \leq \rmean \cdot  (\sum_{j\in \cDD} \alpha_j - \lambda |\is|)\,.
  \end{align*}
  \label{thm:euckmeans}
\end{theorem}
\begin{proof}
  To simplify notation, we use $\delta$ instead of $\dmean$  throughout the proof.
    Consider any client $j\in \clients$. We shall prove that
  \begin{align*}
    \frac{c(j, \is)}{\rmean} \leq \alpha_j - \sum_{i\in N(j)\cap \is} (\alpha_j - c(j,i)) =  \alpha_j - \sum_{i\in \is} [\alpha_j - c(j,i)]^+ \,.
  \end{align*}
  Similarly to the proof of Theorem~\ref{thm:genkmeans}, the statement then follows by summing up over all clients.
  A difference compared to the standard analysis of \jv is that in
  our algorithm we may open several facilities in $N(j)$, i.e., client $j$
  may contribute to the opening of several facilities. 
  We divide our
  analysis into the three cases $|N(j) \cap \is| = 1$, $|N(j) \cap \is| >  1$, and $|N(j) \cap \is|
  = 0$. For brevity, let $S$ denote $N(j) \cap \is$ and $s = |S|$.
  \begin{description}
    \item[\textnormal{Case $s = 1$:}]  If we let $i^*$ be the unique facility
      in $S$,
    \begin{align*}
        \frac{c(j,\is)}{\rmean} \leq c(j,\is) \leq c(j, i^*) = \alpha_j  - (\alpha_j - c(j,i^*)) = \alpha_j - \sum_{i\in N(j) \cap \is}(\alpha_j - c(j,i))\,.
      \end{align*}
                \item[\textnormal{Case $s > 1$:}] In this case, there are multiple
      facilities in $\is$ that $j$ is contributing to.  We need to show that
      $\alpha_j - \sum_{i\in S} (\alpha_j-c(j,i))  \geq  \frac{1}{\rmean}
      c(j,\is)$.

      The sum $\sum_{i\in S} c(j,i)$  is the sum of square distances from $j$ to facilities in $S$ which is at least the sum of square distances of these facilities from their centroid $\mu$, i.e., $\sum_{i\in S} c(j,i) \geq \sum_{i\in S} c(i, \mu)$. Moreover, by the identity, $\sum_{i\in S} c(i,\mu) = \frac{1}{2s}\sum_{i \in S}\sum_{i' \in S} c(i,i')$, we get
\begin{equation*} 
        \sum_{i\in S} c(j,i) \geq \frac{1}{2s}\sum_{i \in S} \sum_{i' \in S}c(i,i')\,.
\end{equation*}
      As there is no edge between any pair of distinct facilities $i$ and $i'$ in $S\subseteq \is$, we must have \[c(i,i') > \delta\cdot\min(t_i,t_{i'}) \ge \delta\cdot\alpha_j,\]
where the last inequality follows because $j$ is contributing to both $i$ and $i'$ and hence  $\min(t_i, t_{i'}) \geq \alpha_j$.  By the above, 
\begin{align*}\sum_{i\in S} c(j,i) \geq \frac{\sum_{i \in S} \sum_{i' \in S}c(i,i')}{2s} \geq \frac{\sum_{i \in S}\sum_{i' \neq i \in S} \delta\cdot \alpha_j}{2s} = \delta \cdot \frac{s-1}{2} \cdot \alpha_j\,.
\end{align*}
Hence,
 \[ 
   \sum_{i\in S}(\alpha_j - c(j,i)) \leq \bigg(s - \delta \cdot \frac{s-1}{2}\bigg)\alpha_j = \Big(s \big(1-\tfrac{\delta}{2}\big) + \tfrac{\delta}{2}\Big)\alpha_j\,.
\] 
Now, since $\delta \geq 2$ the above upper bound is a non-increasing function of $s$. Therefore, since $s\geq 2$ we always have 
 \begin{equation}\label{okmeanslast1}
\sum_{i\in S} (\alpha_j - c(j,i)) \leq \big(2-\tfrac{\delta}{2}\big)\alpha_j\,.
\end{equation}
We also know that $\alpha_j > c(j,i)$ for any $i\in S$. Therefore, $\alpha_j > c(j,\is)$ and, since $\delta \ge 2$:
\begin{equation}\label{okmeanslast2}
\big(\tfrac{\delta}{2} - 1\big)c(j,\is) \leq \big(\tfrac{\delta}{2} - 1\big)\alpha_j\,.
\end{equation} 
Combining Inequalities~\eqref{okmeanslast1} and~\eqref{okmeanslast2},
\[
\sum_{i\in S} (\alpha_j - c(j,i)) + \big(\tfrac{\delta}{2} - 1\big)c(j,\is) \leq \big(2-\tfrac{\delta}{2}\big)\alpha_j  + \big(\tfrac{\delta}{2}-1\big)\alpha_j = \alpha_j\,.
\]
We conclude the analysis of this case by
rearranging the above inequality and recalling that $\rmean \geq
\frac{1}{\delta/2-1}$.

 \item[\textnormal{Case $s = 0$:}] Here, we claim that there exists a tight facility $i$ such that
   \begin{align}
     \label{eq:maineqkmeans}
    d(j,i) + \sqrt{\delta t_i}  \leq (1+\sqrt{\delta})\sqrt{\alpha_j} \,.
   \end{align}
To see that such a facility $i$ exists, consider the witness $w(j)$ of $j$.
   By Claim~\ref{claim:twitness}, we have $\alpha_j \geq t_{w(j)}$ and since $j$ has a tight edge to its witness $w(j)$,  $\alpha_j \geq c(j,w(j) = d(j,w(j))^2$; or, equivalently,  $\sqrt{\alpha_j} \geq \sqrt{t_{w(j)}}$ and $\sqrt{\alpha_j} \geq d(j,w(j))$ which implies that there is a tight facility, namely $w(j)$,  satisfying~\eqref{eq:maineqkmeans}. 

   Since $\is$ is a maximal independent set of $\Gf$, either $i\in \is$, in which case $d(j,\is) \leq d(j,i)$, or there is an $i'\in \is$ such that the edge $(i',i)$ is in $\Gf$, in which case
   \begin{align*}
     d(j,\is) \leq d(j,i) + d(i,i') \leq d(j,i) + \sqrt{\delta t_i}\,,
   \end{align*}
   where the second inequality follows from $d(i,i')^2 =c(i,i') \leq \delta \min(t_i,
   t_{i'})$ by the definition of $\Gf$. In any case, we have by~\eqref{eq:maineqkmeans}
   \begin{align*}
d(j, \is) \leq (1+\sqrt{\delta})\sqrt{\alpha_j} \,.
   \end{align*}
   Squaring both sides and recalling that $\rmean \geq (1+\sqrt{\delta})^2$ completes the last case and
   the proof of the theorem.
  \end{description}
\end{proof}

\subsubsection{$k$-Median objective in Euclidean metrics}
We use a very similar approach as the one for $k$-means (in Euclidean metrics) to address the $k$-median objective in Euclidean metrics. In this section, we have $c(j,i) = d(j,i)$, i.e., the distances are \emph{not} squared. Define,  
\begin{align*}
  \dmed = \sqrt{\tfrac{8}{3}} \qquad \mbox{and} \qquad  \rmed=  1+ \sqrt{\tfrac{8}{3}} = \max\left(1+\dmed,\, 1/(\dmed-1),\, 1/\big(\tfrac{3}{2}\dmed -2\big)\right)\,.
\end{align*}
We have $\dmed \approx 1.633$ and $\rmed \approx 2.633$.
\begin{theorem}
Let $d$ be a Euclidean metric on $\clients \cup \facilities$ and suppose that $c(j,i) = d(j,i)$  for every $i \in \facilities$ and $j \in \clients$.  Then, for any $\lambda \geq 0$, Algorithm \jvp{\dmed} constructs a solution $\alpha$ to \dualf and
  returns a set $\is$ of opened  facilities such that
  \begin{align*}
    \sum_{j\in \clients} c(j, \is)  \leq \rmed \cdot  (\sum_{j\in \cDD} \alpha_j - \lambda |\is|)\,.
  \end{align*}
  \label{thm:euckmed}
\end{theorem}
\begin{proof}
  To simplify notation, we use $\delta$ instead of $\dmed$  throughout the proof. Similar to the proof of the previous theorem, we proceed by considering a single client $j$ and prove 
  \begin{equation*}
\frac{c(j, \is)}{\rmed} \leq \alpha_j - \sum_{i\in N(j)\cap \is} (\alpha_j - c(j,i)) = 
\alpha_j - \sum_{i\in \is} [\alpha_j - c(j,i)]^+\,.
  \end{equation*} Let $S$ denote $N(j) \cap \is$ and $s = |S|$. We again proceed by case distinction on $s$. We first bound the number of cases.
  \begin{claim}
    We have $s\leq 3$. 
  \end{claim}
  \begin{proof}
      Using the centroid property of squared distances in Euclidean space, 
      \begin{equation*} 
      \sum_{i\in S}d(j,i)^2 \geq \frac{\sum_{i\in S}\sum_{i' \in S}d(i,i')^2}{2s}
=  \frac{\sum_{i\in S}\sum_{i' \neq i \in S}d(i,i')^2}{2s}
 > \frac{\delta^2 (s-1)\alpha_j^2 }{2}\,,
      \end{equation*}
      where the last inequality follows from the fact that each pair of facilities $i,i'\in S \subseteq \is$ are not adjacent in $\Gf$ so $d(i,i') > \delta \min(t_i, t_{i'})$ and  $\min(t_i, t_{i'}) \geq \alpha_j$ since $i, i' \in S \subseteq N(j)$. 
      Since the left-hand-side is upper bounded by $s \alpha_j^2$, we get $s > \frac{\delta^2 (s-1)}{2}$. Therefore $s < \frac{\delta^2}{\delta^2-2} = 4 $.
  \end{proof}
  We now proceed by considering the cases $s=0,1,2,3$.
  \begin{description}
    \item[\textnormal{Case $s = 0$:}] Consider the witness $i_1$ of $j$. We have $\alpha_j \geq t_{i_1}$ by Claim \ref{claim:twitness} and also $\alpha_j \geq c(j,i_1) = d(j,i_1)$. Since $\is$ is a maximal independent set of $\Gf$, either $i_1\in \is$, in which case $c(j,\is) = d(j,\is) \leq d(j,i_1) \leq \alpha_j$, or there is an $i_2\in \is$ such that the edge $(i_1,i_2)$ is in $\Gf$, in which case
   \begin{align*}
     c(j,\is) = d(j,\is) \leq d(j,i_1) + d(i_1,i_2) \leq d(j,i_1) + \delta t_{i_1} \leq (1+ \delta) \alpha_j\,.
   \end{align*}
   In any case, we have $c(j,\is)/\rmed \leq \alpha_j$ as required.
    \item[\textnormal{Case $s = 1$:}] If we let $i^*$ be the unique facility
      in $S$,
      \begin{align*}
        \frac{c(j,\is)}{\rmed} \leq c(j,\is) \leq c(j, i^*) = \alpha_j  - (\alpha_j - c(j,i^*)) = \alpha_j - \sum_{i\in N(j) \cap \is}(\alpha_j - c(j,i))\,.
      \end{align*}
    \item[\textnormal{Case $s = 2$:}]
      Let $S= \{i^*_1, i^*_2\}$. We have
      \begin{align*}
      2\alpha_j & =  c(j,i_1^*) + c(j,i^*_2) + (\alpha_j - c(j,i_1^*)) + (\alpha_j - c(j,i_2^*)) \\
      & \geq c(i_1^*,i_2^*) +  \sum_{i\in S} (\alpha_j - c(j,i))  \\ &\geq  \delta \alpha_j + \sum_{i\in S} (\alpha_j - c(j,i))\,,
      \end{align*}
      where we used the triangle inequality and that $c(i_1^*, i_2^*) > \delta
      \min(t_{i^*_1}, t_{i^*_2}) \geq \delta \alpha_j$ since both $i_1^*$ and $i_2^*$ are in $S$ and hence $i_1^*$ and $i_2^*$ are not adjacent in $\Gf$. 
           Rearranging the above inequality noting that $\alpha_j \geq c(j,\is)$, we have
\[
(\delta-1) c(j,\is) \leq \alpha_j - \sum_{i \in S} (\alpha_j - c(j,i)),
\]
and the case follows because $\rmed \geq 1/(\delta-1) $.

    \item[\textnormal{Case $s = 3$:}]
        Similar to the previous case,  
                \begin{align*}
          3\alpha_j &= \sum_{i\in S}c(j,i) + \sum_{i\in S} (\alpha_j - c(j,i)) \\
          & \geq\frac{1}{2}\sum_{\{i,i'\} \subseteq S}c(i,i') + \sum_{i\in S} (\alpha_j - c(j,i)) \\
          & \geq\frac{3\delta }{2}\cdot \alpha_j + \sum_{i\in S} (\alpha_j - c(j,i))\,,
        \end{align*}
        using the triangle inequality. Rearranging the above inequality noting that $\alpha_j \geq c(j,\is)$, we have
\[
\left(\tfrac{3\delta}{2}-2\right)c(j,\is) \leq \alpha_j - \sum_{i \in S} (\alpha_j - c(j,i))\,
\] and the lemma follows because $\rmed \geq 1/(\frac{3\delta}{2} - 2)$.
   \end{description}

\end{proof}

\section{Quasi-polynomial time algorithm}
\label{sec:quasi}
In this section, we present a quasi-polynomial time approach that turns the \lmp approximation algorithms presented in the previous section into approximation algorithms for the original problems ($k$-means and $k$-median), i.e., into algorithms that find solutions  satisfying  the strict constraint that at most $k$ facilities are opened. This is achieved  by only deteriorating the approximation guarantee by an arbitrarily small factor regulated by $\epsilon$.  We also
 introduce several of the ideas used in the polynomial time approach. Although
 the results obtained in this section are weaker (quasi-polynomial instead of
 polynomial), we believe that  the easier quasi-polynomial algorithm serves as
 a good starting point before reading the more complex 
 polynomial time algorithm.  From now on, we concentrate on the $k$-means problem and we let $\rho = \rmean$ denote the approximation guarantee and $\delta = \dmean$ denote the parameter to our algorithm, where $\rmean$ and
$\dmean$ are defined as in Section~\ref{sec:euckmeans} (it will be clear that the techniques presented here are easily applicable to the other considered objectives, as well).  Throughout this section we fix $\aeps> 0$ to be a small constant, and we assume for notational convenience and without loss of generality that $n\gg 1/\aeps$.  We shall also assume that the distances satisfy the following:
\begin{lemma}\label{lem:dis}
  By losing a factor $(1+100/n^2)$ in the approximation guarantee, we can assume that the squared-distance between any client $j$ and any facility $i$ satisfies: $1 \leq d(i,j)^2 \leq n^6$, where $n=|\cD|$.
\end{lemma}
\noindent The proof follows by standard discretization techniques and is presented in Appendix~\ref{appdistances}. 

Our algorithm will produce a $(\rho + O(\aeps))$-approximate solution.  In the algorithm, we consider separately the two phases of the primal-dual
algorithm from Section~\ref{sec:euckmeans}. Suppose that the first phase produces a set of values $\alpha
= \{\alpha_j\}_{j \in \cDD}$ satisfying the following definition:
\begin{definition}
\label{def:quasi-good-soln}
A feasible solution $\alpha$ of $\dualf[\lambda]$ is \emph{good} if for every
$j \in \cDD$ there exists a tight facility $i$ such that
$(1+\sqrt{\delta} + \aeps) \sqrt{\smash[b]{\alpha_j}} \ge d(j,i) + \sqrt{\smash[b]{\delta t_i}}$.
\end{definition}
\noindent Recall that for a dual solution $\alpha$, $t_i$ is defined to be the largest $\alpha$-value out of all clients that are contributing to a facility $i$: $t_i = \max_{j\in N(i)} \alpha_j$ where $N(i) = \{j\in \clients: \alpha_j - d(i,j)^2 > 0\}$.

As the condition of Definition~\ref{def:quasi-good-soln} relaxes~\eqref{eq:maineqkmeans} by a tiny amount (regulated
by $\epsilon$), our analysis in Section~\ref{PDalg} shows that as long as the
first stage of the primal-dual algorithm produces an $\alpha$ that is good, the
second stage will find a set of facilities $\is$ such that $\sum_{j\in \clients} d(j,\is)^2 = \sum_{j \in
\clients}c(j,\is) \le (\rho + O(\epsilon))\big(\sum_{j \in \cDD}\alpha_j
- \lambda|\is|\big)$.  If we could somehow find a value $\lambda$ such that the
second stage opened \emph{exactly} $k$ facilities, then we would obtain
a $(\rho + O(\epsilon))$-approximation algorithm.  In order to accomplish this, we first enumerate all potential values $\lambda = 0,1\cdot\zeps,2 \cdot \zeps,\ldots, L\cdot \zeps$, where $\zeps$ is a small step size and $L$ is large enough to guarantee that we eventually find a solution of size at most $k$ (for a precise definition of $L$ and $\zeps$, see~\eqref{eq:quasi-step} and~\eqref{eq:quasiL}).  Specifically, in Section~\ref{sec:quasi-polyn-time-algor}, we give an algorithm that in time $n^{O(\aeps^{-1} \log n)}$ generates a 
quasi-polynomial-length sequence of solutions
$\alphat{0},\alphat{1},\ldots,\alphat{L}$, where $\alphat{\ell}$ is a
good solution to $\dualf[\ell\cdot \zeps]$.  We shall ensure that each consecutive set of
values $\alphat{\ell}, \alphat{\ell+1}$ are \emph{close} in the following
sense:
\begin{definition}
\label{def:quasi-close-sequence}
Two solutions $\alpha$ and $\alpha'$ are \emph{close} if $|\alpha'_j
- \alpha_j| \le \frac{1}{n^2}$ for all $j \in \cD$.
\end{definition}

Unfortunately, it may be the case that for a good solution $\alphat{\ell}$ to $\dualf[\lambda]$, the second stage of our algorithm opens more than $k$ facilities, while for the next good solution $\alphat{\ell+1}$ to $\dualf[\lambda+\zeps]$, it opens fewer than $k$ facilities.  In order to obtain a solution that opens \emph{exactly} $k$ facilities, we must somehow interpolate between consecutive solutions in our sequence.  In Section~\ref{sec:quasi-rounding} we describe an algorithm that accomplishes this task.  Specifically, for each pair of consecutive solutions $\alphat{\ell},\alphat{\ell+1}$ we show that, since their $\alpha$-values are nearly the same, we can control the way in which a maximal independent set in the associated conflict graphs changes.  Formally, we show how to maintain a sequence of approximate integral solutions with cost bounded by $\alphat{\ell}$ and $\alphat{\ell+1}$, in which the number of open facilities decreases by at most one in each step.  This ensures that some solution indeed opens exactly $k$ facilities and it will be found in time $n^{O(\aeps^{-1} \log n)}$.  

\subsection{Generating a sequence of close, good solutions}
\label{sec:quasi-polyn-time-algor}
We first describe our procedure for generating a close sequence of good solutions. Select the following parameters
\begin{align}
  \zeps &= n^{-3-10 \log_{1+\aeps} n}  \label{eq:quasi-step}\,,\\
  L & = 4n^{7} \cdot \zeps^{-1} = n^{O(\aeps^{-1}  \log n)} \,. \label{eq:quasiL}
\end{align}
We also use the notion of \emph{buckets} that partition the real line:
\begin{definition}
  \label{def:quasi-buckets}
  For any value $v\in \mathbb{R}$, let  \ifshort
  $B(v) = 0$, if $v<1$, and $B(v)$ = $1+\lfloor \log_{1+\epsilon}(v)\rfloor$ if $v \geq 1$.
  \else
  \begin{align*}
    B(v) = \begin{cases}
      0 & \mbox{if $v<1$}, \\
      1+\lfloor \log_{1+\epsilon}(v)\rfloor  & \mbox{if $v\geq 1$}.
    \end{cases}
  \end{align*}
  \fi
We say that $B(v)$ is the index of the \emph{bucket containing $v$}. 
\end{definition}
\noindent
The buckets will be used to partition the $\alpha$-values of the clients. As, in every constructed solution $\alpha$, each client will have a tight edge to a facility,  Lemma \ref{lem:dis} implies that $\alpha_j$ will always be at least $1$. Therefore, the definition gives the property that the $\alpha$-values of any two clients $j$ and $j'$  in the same bucket differ by at most a factor of $1+\epsilon$. In other words, the buckets will be used to classify the clients according to similar $\alpha$-values.

We now describe a procedure \quasisweep that takes as input a good dual solution $\alphain$ of $\dualf$ and outputs a good dual solution $\alphaout$ of  $\dualf[\lambda + \zeps]$ such that $\alphain$ and $\alphaout$ are close.  
In order to generate the desired close sequence of solutions, 
we first define an initial solution for $\dualf[0]$ by $\alpha_j = \min_{i\in
\facilities} d(i,j)^2$ for $j\in \clients$.  Then, for $0 \le \ell < L$, we call \quasisweep with $\alphain = \alphat{\ell}$ to generate the next solution $\alphat{\ell+1}$ in our sequence.  We shall show that each $\alphat{\ell}$ is a feasible dual solution of $\dualf[\ell \cdot \zeps]$, and that the following invariant holds throughout the generation of our sequence:
\begin{invariant}
  \label{inv:quasi}
  In every solution $\alpha = \alphat{\ell}$, $(0 \le \ell \le L)$, every client $j\in \clients$ has a tight edge to  a tight facility $w(j) \in \facilities$ (its witness) such that $B(t_{w(j)}) \leq B(\alpha_j)$. 
\end{invariant}
Note that this implies that each solution in our sequence is good.  Indeed, consider a dual solution $\alpha$ that satisfies Invariant~\ref{inv:quasi}. Then, for any client $j$, we have some $i$ ($= w(j)$) such that $\sqrt{\alpha_j} \geq d(i,j)$ (since $j$ has a tight edge to $w(j)$) and
$\sqrt{(1+\epsilon)\delta\alpha_j} \geq \sqrt{\delta t_{i}}$ where we used that $B(\alpha_j) \geq B(t_i)$ implies $(1+\aeps)\alpha_j \geq t_i$.
Hence, 
\ifshort
$  (1+ \sqrt{\delta} + \aeps) \sqrt{\alpha_j} \geq \left(1 + \sqrt{(1+\aeps) \delta}\right)\sqrt{\alpha_j} \geq d(i,j) + \sqrt{\delta t_i}
  $,
\else
\begin{align*}
  (1+ \sqrt{\delta} + \aeps) \sqrt{\alpha_j} \geq \left(1 + \sqrt{(1+\aeps) \delta}\right)\sqrt{\alpha_j} \geq d(i,j) + \sqrt{\delta t_i}\,,
\end{align*}
\fi
and so $\alpha$ is good (here, for the first inequality we have used that $\sqrt{1+\aeps} \leq 1+\aeps/2$ and $\sqrt{\delta} \leq 2$).  We observe that our initial solution $\alphat{0}$ has $t_i = 0$ for all $i \in \facilities$, and so Invariant~\ref{inv:quasi} holds trivially.  In our following analysis, we will show that each call to \sweep preserves Invariant~\ref{inv:quasi}.

\subsubsection{Description of \quasisweep} We now formally describe the procedure \quasisweep that, given the last previously generated solution $\alphain$ in our sequences produces the solution $\alphaout$ returned next.

We initialize the algorithm by
setting $\alpha_j = \alphain_j$ for each $j\in \clients$ and by increasing the
opening prices of each facility from $\lambda$ to $\lambda+\zeps$. At this
point, no facility is tight and therefore the solution $\alpha$ is not a good
solution of $\dualf[\lambda+\zeps]$. We now describe how to modify $\alpha$ to obtain a solution $\alphaout$ satisfying Invariant~\ref{inv:quasi} (and hence into a good solution). The algorithm will maintain a current set $A$ of active clients  and a current threshold $\theta$. Initially, $A= \emptyset$, and $\theta = 0$. We slowly increase $\theta$ and whenever $\theta = \alpha_j$ for some client $j$, we add $j$ to $A$. While $j\in A$, we increase $\alpha_j$ at the same rate as $\theta$. We remove a client $j$ from $A$, whenever the following occurs:
\begin{itemize}
  \item[] $j$ has a tight edge to some tight facility $i$ with $B(\alpha_j) \geq B(t_i)$.  In this case, we say that $i$ is the \emph{witness} of $j$.
\end{itemize}
Note that if a client $j$ satisfies this condition when it is added to $A$, then we remove $j$ from $A$ immediately after it is added.  In this case, $\alpha_j$ is not increased.

Increasing the $\alpha$-values for clients in $A$, may cause the contributions to some facility $i$ to exceed the opening cost $\lambda+ \zeps$.  To prevent this from happening, we also decrease every value $\alpha_{j}$ with $B(\alpha_{j}) > B(\theta)$ at a rate of $|A|$ times the rate that $\theta$ is increasing.  Observe that while there exists any such $j \in N(i)$, the total contribution of the clients toward opening this $i$ cannot increase, and so $i$ cannot become tight.  It follows that once any facility $i$ becomes tight, $B(\alpha_{j}) \le B(\theta)$ for every $j \in N(i)$ and so $i$ is presently a witness for all clients $j \in N(i)\cap A$.  At this moment all such clients in $N(i) \cap A$ will be removed from $A$ and their $\alpha$-values will not subsequently be changed.  Thus, $i$ remains tight until the end of \quasisweep.  Moreover, observe any other client $j'$ that is added to $A$ later will immediately be removed from $A$ as soon as it has a tight edge to $i$.  Thus, neither $t_i$ nor the total contribution to $i$ change throughout the remainder of \quasisweep.  In particular, $i$ remains a witness for all such clients $j$ for the remainder of \quasisweep.

We stop increasing $\theta$ once every client $j$ has been added and removed from $A$. The procedure \quasisweep then terminates and outputs $\alphaout = \alpha$.  As we have just argued, the contributions to any tight facility can never increase, and every client that is removed from $j$ will have a witness through the rest of \quasisweep (in particular, in $\alphaout$).  Thus, $\alphaout$ is a feasible solution of $\dualf[\lambda+\zeps]$ in which every client $j$ has a witness $w(j)$, i.e., $j$ has a tight edge to the tight facility $w(j)$ and $B(t_{w(j)}) \leq B(\alpha_j)$. Hence,  the output of \sweep always satisfies Invariant~\ref{inv:quasi}.

This completes the description of \quasisweep. For a small example of its
execution see Figure~\ref{fig:quasiseq}. We now show that the produced sequence of solutions is close and to analyze the running time.

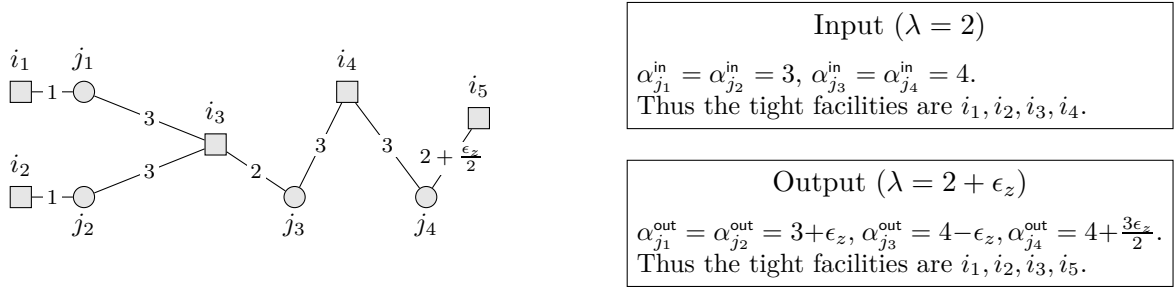
\begin{figure}[ht]
  \centering
  \begin{tikzpicture}
    \begin{scope}[scale=0.7]
      \node[draw=black] at (15, 1.5) {\begin{minipage}{7cm}\begin{center} Input $(\lambda = 2)$\end{center}\vspace{-2mm} \small $\alphain_{j_1} = \alphain_{j_2} = 3$, $\alphain_{j_3} = \alphain_{j_4} = 4$. \\Thus the tight facilities are $i_1, i_2, i_3, i_4$.\end{minipage}};
      \node[draw=black] at (15, -1.5) {\begin{minipage}{7cm}\begin{center} Output $(\lambda = 2+\zeps)$\end{center}\vspace{-2mm} \small $\alphaout_{j_1} = \alphaout_{j_2} = 3+\zeps$, $\alphaout_{j_3} = 4-\zeps, \alphaout_{j_4} = 4+\tfrac{3\zeps}{2}$. Thus the tight facilities are $i_1, i_2, i_3, i_5$.\end{minipage}};
      \node[ssquare, label={\small $i_1$}] (i1a) at (-1.7,1) {} ;
      \node[ssquare, label={\small $i_2$}] (i1b) at (-1.7,-1) {} ;
      \node[svertex, label={\small $j_1$}] (j1) at (-0.5,1) {};
      \node[svertex, label={[yshift=-0.8cm]\small $j_2$}] (j2) at (-0.5, -1) {};
      \node[ssquare, label={\small $i_3$}] (i2) at (2, 0) {};
      \node[svertex, label={[yshift=-0.8cm]\small $j_3$}] (j3) at (3.5, -1) {};
      \node[ssquare, label={\small $i_4$}] (i3) at (4.5, 1) {};
      \node[svertex, label={[yshift=-0.8cm]\small $j_4$}] (j4) at (6, -1) {};
      \node[ssquare, label={\small $i_5$}] (i4) at (7, 0.5) {};
      \draw (i1a) edge node[fill=white,inner sep=0.5mm] {\scriptsize $1$} (j1);
      \draw (i1b) edge node[fill=white,inner sep=0.5mm] {\scriptsize $1$} (j2);
      \draw (i2) edge node[fill=white,inner sep=0.5mm] {\scriptsize $3$} (j1) edge node[fill=white,inner sep=0.5mm] {\scriptsize $3$} (j2) edge node[fill=white,inner sep=0.5mm] {\scriptsize $2$} (j3);
      \draw (i3) edge node[fill=white,inner sep=0.5mm] {\scriptsize $3$} (j3) edge node[fill=white,inner sep=0.5mm] {\scriptsize $3$} (j4);
      \draw (i4) edge node[fill=white,inner sep=0.5mm] {\scriptsize $2+\tfrac{\zeps}{2}$} (j4);
    \end{scope}

  \end{tikzpicture}
  \caption{The instance has $4$ clients and $5$ facilities depicted by circles
  and squares, respectively. The number on an edge is the squared-distance of that edge and the squared-distances that are not depicted are all of value $5$. Given the input solution $\alphain$ with
$\lambda= 2$, \quasisweep proceeds as follows. First the opening prices
of facilities are increased to $2+\zeps$. Next the clients $j_1,j_2$ are added
to the set $A$ of active clients when the threshold $\theta = 3$.  Then, until $\theta = 3+\zeps$, $\alpha_{j_1}$ and $\alpha_{j_2}$ increase at a uniform rate while the (significantly) larger dual values $\alpha_{j_3}$ and $\alpha_{j_4}$ are decreasing $|A| =2$ times that rate.  
At the point $\theta = 3+\zeps$, both $i_1$ and $i_2$
become tight and the witnesses of $j_1$ and $j_2$ respectively. This
causes these clients to be removed from $A$ which stops their increase and the decrease of the larger values. When $\theta = 4-2\zeps$, $j_3$ and $j_4$ are added to $A$ and they start to increase at a uniform rate. Next, the facility $i_3$ becomes tight when $\theta = 4-\zeps$ and client $j_3$ is removed from $A$ with $i_3$ as its witness. Finally, $j_4$ is removed from $A$ when  $\theta = 4+3\zeps/2$ at which point $i_5$ becomes tight and its witness. }
  \label{fig:quasiseq}
\end{figure}

\subsubsection{Closeness and running time}
We begin by showing that \quasisweep produces a close sequence of solutions.

\begin{lemma}
  For each client $j \in \clients$, we have $|\alphain_j  - \alphaout_j| \leq 1/n^2$.
  \label{lem:quasi-close}
\end{lemma}
\begin{proof}
  We first note that the largest $\alpha$-value at any time is at most
  $(\lambda + \zeps)+  n^6 \leq L \zeps + n^6 = 4n^{7} + n^6 \leq 5 n^{7}$. This
  follows from the feasibility of $\alpha$ because, by Lemma~\ref{lem:dis}, no
  squared-distance is larger than $n^6$ and the opening cost of any facility is
  at most $\lambda+\zeps \leq L\zeps$.  Hence, $B(\alpha_j) \leq 1 + \lfloor
  \log_{1+\aeps}(5n^7) \rfloor \leq 10 \log_{1+\aeps}(n)$ for any client $j$
  and dual solution $\alpha$.

We now prove the following claim:
\begin{claim*}
Any $\alpha_j$ can increase by at most $\zeps n^{3b}$ while $B(\theta) \le b$.  
\end{claim*}
\begin{proof} The proof is by induction on $b= 0, 1,\dots, 10\log_{1+\aeps}(n)$.
\paragraph{Base case $b=0$:}   
This case is trivially true because there are no clients $j$ with $B(\alpha_j)  = 0$, and so no clients can have been added to $A$ while $B(\theta) = 0$. Indeed, any client $j$ had a tight edge to some facility in $\alphain$, which by Lemma~\ref{lem:dis} implies $\alphain_j\geq 1$, and a client's $\alpha$-value can decrease only while some smaller $\alpha$-value is increasing.

\paragraph{Inductive step (assume true for $0,1,\dots, b-1$ and prove for $b$):}       
Now, we suppose some $\alpha_j$ is increasing while $B(\theta) \le b$.  
Note that we then must have $\alpha_j = \theta$.  Let $i$ be the witness of $j$ in $\alphain$, and let $\Nin(i)$ be the set of clients contributing to $i$ in $\alphain$.  We  further suppose that $\alpha_j$ is increased by at least $\zeps$ while $B(\theta) \le b$; otherwise the claim follows immediately, since $\zeps \leq n^{3b} \zeps$ for all $b \geq 0$.

First, suppose that $\alpha_j < \alphain_j$ and so $\alpha_j$ was previously decreased by \quasisweep. Moreover, since $\alpha_j$ has increased by at least $\zeps$ while $B(\theta) \le b$, we must have previously decreased $\alpha_j$ while $B(\alpha_j) \le b$.  In particular, at the last moment $\alpha_j$ was decreased, we must have had $B(\alpha_j) \le b$, and since $\alpha_j$ was decreasing at this moment, we also had $B(\theta) < B(\alpha_j)$.  Then, $\alpha_j$ was decreased only while $B(\theta) < b$.  Moreover, during this time, $j$'s $\alpha$-value was decreased at a rate of $|A|$, and so was decreased (in total) at most $n$ times the largest amount that any other $\alpha_{j'}$ was increased.  By the induction hypothesis, any $\alpha_{j'}$ was increased at most $\zeps \cdot n^{3b-3}$ while $B(\theta) < b$, and so $\alpha_j$ was decreased  at most $\zeps \cdot n^{3b-2}$.  Thus, after $\alpha_j$ increases by at most $\zeps \cdot n^{3b-2}$ we will again have $\alpha_j = \alphain_j$.

Now, we consider how much $\alpha_j$ may increase while $\alpha_j \ge \alphain_j$ (and still $B(\theta) \leq b$).  For each $j' \in \Nin(i)$ we must have initially had $B(\alphain_{j'}) \le B(\alphain_j)$ since $i$ is a witness for $j$ in $\alphain$.  Additionally, by our assumptions in this case, $B(\alphain_j) \leq B(\alpha_j) \leq b$.  Thus, the $\alpha$-value of any $j' \in \Nin(i)$ was decreased by \quasisweep only while $B(\theta) \le b-1$ and so by the same argument as above the $\alpha$-value of any $j' \in \Nin(i)$ can decrease at most $\zeps n^{3b-2}$ throughout \quasisweep.  Thus, the total contribution to $i$ from all $j' \neq j$ can decrease at most $(n - 1) \cdot \zeps \cdot n^{3b - 2}$.  After increasing $\alpha_j$ at most $(n - 1) \cdot \zeps\cdot n^{3b-2} + \zeps$, $i$ will again be tight.  
Moreover, at this moment, any client $j'$ contributing to $i$ was either already added to $A$ (and potentially also removed) in which case $\alpha_{j'} \leq \theta = \alpha_j$ or it was not already added to $A$ in which case $\alpha_{j'} \leq \alphain_{j'}$ since $\alpha_{j'}$ has not been increased yet. In either case, $B(\alpha_{j'})  \leq B(\alpha_j)$ so $i$ is a witness for $j$, and $j$ will be removed from~$A$.

Altogether, the total amount $\alpha_j$ can increase while $B(\theta) \le b$ is then the sum of these two increases, which is $ \zeps \cdot n^{3b-2} + (n-1) \cdot \zeps \cdot n^{3b - 2} + \zeps \le \zeps \cdot n^{3b}$, as required.
\end{proof}

The claim immediately bounds the increase $\alphaout_j - \alphain_j$ by $\tfrac{1}{n^3} \leq \tfrac{1}{n^2}$ as 
required (recall that $\zeps = n^{-3-10 \log_{1+\aeps} n }$)).  Moreover, as shown in the proof of the claim above, the $\alpha$-value of every client decreases by no more than  $n$ times the maximum increase in the $\alpha$-value of any client.  Then, the desired bound $\tfrac{1}{n^2}$ on $\alphain_j - \alphaout_j$ follows as well.
\end{proof}

For the sake of clarity, we have presented the \quasisweep procedure in a continuous fashion.  We show in Appendix~\ref{sec:quasi-runtime} how to implement \quasisweep as a discrete algorithm running in polynomial time.  We conclude the analysis of this section by noting that, as  \sweep is repeated $L = n^{O(\aeps^{-1} \log n)}$
times, the total running time for producing the sequence $\alphat{0}, \alphat{1}, \dots, \alphat{L}$ is $n^{O(\aeps^{-1} \log n)}$.

\subsection{Finding a solution of size $k$}
\label{sec:quasi-rounding}

In this section we describe our algorithm for finding a solution of $k$ facilities
given a close sequence $\alphat{0}, \alphat{1}\dots, \alphat{L}$, where
$\alphat{\ell}$ is a good solution to $\dualf[\zeps \cdot \ell]$.

We associate with each dual solution $\alphat{\ell}$ a client-facility
graph and a conflict graph that are defined in  exactly the same way as in Section~\ref{sec:pddesc}: that is, the graph $\Gcft{\ell}$ is a bipartite graph with all of $\clients$ on one side and every tight facility in $\alphat{\ell}$ on the other and $\Gcft{\ell}$ contains the edge $(j,i)$ if and only if $\alphat{\ell}_j > c(j,i)$.  Given each $\Gcft{\ell}$, recall that $\Gft{\ell}$ is then a graph consisting of the facilities present in $\Gcft{\ell}$, which contains an edge $(i,i')$ if  $i$ and $i'$ are both adjacent to some client $j$ in $\Gcft{\ell}$ and $c(i,i') \leq \delta \min(t^{(\ell)}_i, t^{(\ell)}_{i'})$, where for each $i$, we have $t^{(\ell)}_i = \max \{ \alphat{\ell}_j : \alphat{\ell}_j > c(j,i)\}$ (and again we adopt the convention that $t^{(\ell)}_i = 0$ if $\alphat{\ell}_j \leq c(j,i)$ for all $j \in \clients$).
Thus, we have a sequence $\Gcft{0}, \dots, \Gcft{L}$  of client-facility graphs and
a sequence $\Gft{0}, \dots, \Gft{L}$ of conflict graphs obtained from our
sequence of dual solutions.  The main goal of this section is to give
a corresponding sequence  of maximal independent sets
of the conflict graphs so that the size of the solution (independent set) never
decreases by more than $1$ in this sequence.     Unfortunately, this is not quite possible.
Instead,  starting with a maximal independent set $\ist{\ell}$ of $\Gft{\ell}$,
we shall slowly transform it into a maximal independent set $\ist{\ell+1}$ of
$\Gft{\ell+1}$ by considering maximal independent sets in a sequence of
polynomially many intermediate conflict graphs $\Gft{\ell} = \Gft{\ell, 0}
, \Gft{\ell,1}, \dots, \Gft{\ell, p_\ell} = \Gft{\ell+1}$. We shall refer to
these independent sets as $\ist{\ell} = \ist{\ell, 0}, \ist{\ell,1}, \dots,
\ist{\ell, p_\ell} = \ist{\ell+1}$. This interpolation will allow us to ensure
that the size of our independent set decreases by at most $1$ throughout this
sequence.  It follows that at some point
we find a solution $\is$ of size exactly $k$: on the one hand,  since $\Gft{0}$ contains all facilities and no edges we have $\ist{0} = \cF$ , which by assumption is
strictly greater than $k$. On the other hand, we must have $|\ist{L}| \leq 1$. Indeed,  as $\alphat{L}$
is a good dual solution of $\dualf[L\zeps] = \dualf[4 n^7]$, we claim $\Gft{L}$ is a clique. First, note that any tight facility $i$ in $\alphat{L}$ has $t_i \geq \frac{L\zeps}{n}=4n^6$ which means that all clients have a tight edge to $i$ when $i$ becomes tight (since the maximum squared facility-client distance is $n^6$ by Lemma~\ref{lem:dis}). Second, any two facilities $i,i'$ have $d(i,i')^2 \leq 4n^6$ using the triangle inequality and facility-client distance bound. Combining these two insights, we can see that $\Gft{L}$ is a clique and so $|\ist{L}| \leq 1$.

It remains to describe and analyze the procedure $\quasigraphupdate$ that will
perform the interpolation between two conflict graphs $\Gft{\ell}$ and
$\Gft{\ell+1}$ when given a maximal independent set $\ist{\ell}$ of
$\Gft{\ell}$ so that $|\ist{\ell}| >k$.  We run this procedure at most $L$
times starting with $ \Gft{0}, \Gft{1}$, and $\ist{0} = \cF$  until we find a solution
of size $k$.

\subsubsection{Description of \quasigraphupdate}
\label{sec:descquasigraphupdate}
\begin{figure}[t]
  \centering
  \begin{tikzpicture}
    \begin{scope}[scale=0.7]
      \draw[rounded corners=15pt]  (-2.2,-2.5) rectangle (8.7,3);
      \node at (0.5, 2.5) {$\Gcft{\ell}$};
      \node at (6.5, 2.5) {$\Gft{\ell}$};
      \node[ssquare,fill=white, label=right:{\footnotesize $\timet{\ell}_1 = 3$}] (i1) at (1,1.5) {} ;
      \node[ssquare,fill=white, label=right:{\footnotesize $\timet{\ell}_2 = 3$}] (i2) at (1,0.5) {} ;
      \node[ssquare,fill=white, label=right:{\footnotesize $\timet{\ell}_3 = 4$}] (i3) at (1,-0.5) {} ;
      \node[ssquare,fill=white, label=right:{\footnotesize $\timet{\ell}_4 = 4$}] (i4) at (1,-1.5) {} ;
      \node[svertex, label=left:{\small $j_1$}] (j1) at (-1,1.5) {};
      \node[svertex, label=left:{\small $j_2$}] (j2) at (-1,0.5) {};
      \node[svertex, label=left:{\small $j_3$}] (j3) at (-1,-0.5) {};
      \node[svertex, label=left:{\small $j_4$}] (j4) at (-1,-1.5) {};
      \draw (j1) edge (i1);
      \draw (j2) edge (i2);
      \draw (j3) edge (i3) edge (i4);
      \draw (j4)  edge (i4);
      \node[ssquare,fill=white, pattern=north west lines, label=above:{\small $i_1$}] (hi1) at (5,0.5) {} ;
      \node[ssquare,fill=white, pattern=north west lines, label=above:{\small $i_2$}] (hi2) at (6,0.5) {} ;
      \node[ssquare,fill=white, pattern=north west lines, label=above:{\small $i_3$}] (hi3) at (7,0.5) {} ;
      \node[ssquare,fill=white, label=above:{\small $i_4$}] (hi4) at (8,0.5) {} ;
      \draw (hi3) edge (hi4);
    \end{scope}
    \begin{scope}[scale=0.7,yshift=-6cm]
      \draw[rounded corners=15pt]  (-2.2,-2.5) rectangle (8.7,3);
      \node at (0.5, 2.5) {$\Gcft{\ell+1}$};
      \node at (6.5, 2.5) {$\Gft{\ell+1}$};
      \node[ssquare,fill=gray!70, label=right:{\footnotesize $\timet{\ell+1}_1 = 3+\zeps$}] (i1) at (1,1.5) {} ;
      \node[ssquare,fill=gray!70, label=right:{\footnotesize $\timet{\ell+1}_2 = 3+\zeps$}] (i2) at (1,0.5) {} ;
      \node[ssquare,fill=gray!70, label=right:{\footnotesize $\timet{\ell+1}_3 = 4-\zeps$}] (i3) at (1,-0.5) {} ;
      \node[ssquare,fill=gray!70, label=right:{\footnotesize $\timet{\ell+1}_5 = 4+\tfrac{3\zeps}{2}$}] (i5) at (1,-1.5) {} ;
      \node[svertex, label=left:{\small $j_1$}] (j1) at (-1,1.5) {};
      \node[svertex, label=left:{\small $j_2$}] (j2) at (-1,0.5) {};
      \node[svertex, label=left:{\small $j_3$}] (j3) at (-1,-0.5) {};
      \node[svertex, label=left:{\small $j_4$}] (j4) at (-1,-1.5) {};
      \draw (j1) edge (i1) edge (i3);
      \draw (j2) edge (i2) edge (i3);
      \draw (j3) edge (i3);
      \draw (j4)  edge (i5);
      \node[ssquare,fill=gray!70, label=above:{\small $i_1$}] (hbi1) at (5,0.5) {} ;
      \node[ssquare,fill=gray!70, label=above:{\small $i_2$}] (hbi2) at (6,0.5) {} ;
      \node[ssquare,fill=gray!70, label=above:{\small $i_3$}] (hbi3) at (7,0.5) {} ;
      \node[ssquare,fill=gray!70, label=above:{\small $i_5$}] (hbi5) at (8,0.5) {} ;
      \draw (hbi3) edge (hbi2) edge[bend left] (hbi1);
    \end{scope}
    \begin{scope}[scale=0.7, xshift=12cm]
      \draw[rounded corners=15pt]  (-2.2,-6.5) rectangle (8.7,3);
      \node at (0.5, 2.5) {$\Gcf$};
      \node at (6.5, 2.5) {$\Gf$};
      \node[ssquare,fill=white, label=right:{\footnotesize $t_1 = 3$}] (i1) at (1,1.5) {} ;
      \node[ssquare,fill=white, label=right:{\footnotesize $t_2 = 3$}] (i2) at (1,0.5) {} ;
      \node[ssquare,fill=white, label=right:{\footnotesize $t_3 = 4$}] (i3) at (1,-0.5) {} ;
      \node[ssquare,fill=white, label=right:{\footnotesize $t_4 = 4$}] (i4) at (1,-1.5) {} ;
      \node[ssquare,fill=gray!70, label=right:{\footnotesize $t_{1'} = 3+\zeps$}] (bi1) at (1,-2.5) {} ;
      \node[ssquare,fill=gray!70, label=right:{\footnotesize $t_{2'} = 3+\zeps$}] (bi2) at (1,-3.5) {} ;
      \node[ssquare,fill=gray!70, label=right:{\footnotesize $t_{3'} = 4-\zeps$}] (bi3) at (1,-4.5) {} ;
      \node[ssquare,fill=gray!70, label=right:{\footnotesize $t_{5}= 4+\tfrac{3\zeps}{2}$}] (bi5) at (1,-5.5) {} ;
      \node[svertex, label=left:{\small $j_1$}] (j1) at (-1,-0.5) {};
      \node[svertex, label=left:{\small $j_2$}] (j2) at (-1,-1.5) {};
      \node[svertex, label=left:{\small $j_3$}] (j3) at (-1,-2.5) {};
      \node[svertex, label=left:{\small $j_4$}] (j4) at (-1,-3.5) {};
      \draw (j1) edge (i1);
      \draw (j2) edge (i2);
      \draw (j3) edge (i3) edge (i4);
      \draw (j4)  edge (i4);
      \draw (j1) edge (bi1) edge (bi3);
      \draw (j2) edge (bi2) edge (bi3);
      \draw (j3) edge (bi3);
      \draw (j4)  edge (bi5);
      \node[ssquare,fill=white, pattern=north west lines, label=above:{\small $i_1$}] (hi1) at (5,-0.5) {} ;
      \node[ssquare,fill=white, pattern=north west lines, label=above:{\small $i_2$}] (hi2) at (6,-0.5) {} ;
      \node[ssquare,fill=white, pattern=north west lines, label=above:{\small $i_3$}] (hi3) at (7,-0.5) {} ;
      \node[ssquare,fill=white, label=above:{\small $i_4$}] (hi4) at (8,-0.5) {} ;
      \draw (hi3) edge (hi4);
      \node[ssquare,fill=gray!70, label=below:{\small $i_1$}] (hbi1) at (5,-3) {} ;
      \node[ssquare,fill=gray!70, label=below:{\small $i_2$}] (hbi2) at (6,-3) {} ;
      \node[ssquare,fill=gray!70, label=below:{\small $i_3$}] (hbi3) at (7,-3) {} ;
      \node[ssquare,preaction ={fill, gray!70}, pattern=north west lines, label=below:{\small $i_5$}] (hbi5) at (8,-3) {} ;
      \draw (hbi3) edge (hbi2) edge[bend left] (hbi1);
      \draw (hbi1) edge (hi1);
      \draw (hbi2) edge (hi2);
      \draw (hbi3) edge (hi1) edge (hi2) edge (hi3) edge (hi4);
      \draw (hbi5) edge (hi4);
    \end{scope}

  \end{tikzpicture}
  \caption{An example of the ``hybrid'' client-facility graph and associated
    conflict graph used by \quasigraphupdate. $\Gcft{\ell}$ and $\Gcft{\ell+1}$
    are the client-facility graphs of $\alphain$ and $\alphaout$ of
    Figure~\ref{fig:quasiseq}. Next to the facilities, we have written the
    facility times ($t_i$'s) of those solutions. As the squared-distance
    between any two facilities is $5$ in the example of
    Figure~\ref{fig:quasiseq}, one can see that any two facilities with
    a common neighbor in the client-facility graph will be adjacent in the
    conflict graph.  $G$ is the ``hybrid''
    client-facility graph of $\Gcft{\ell}$ and $\Gcft{\ell+1}$. When $H$ is
    formed, we extend the given maximal independent set $\ist{\ell}$ of
    $\Gft{\ell}$ to form a maximal independent set of $H$. The facilities in
  the relevant independent sets are indicated with stripes.  }
  \label{fig:quasigraph}
\end{figure}
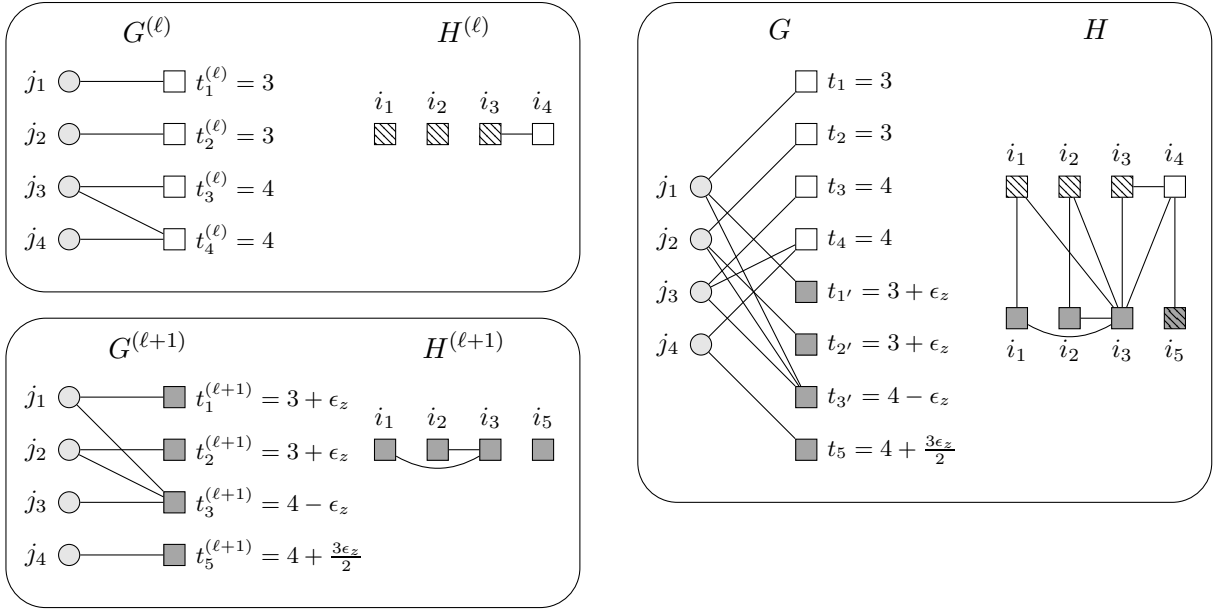
 Denote the input by $\Gft{\ell}, \Gft{\ell+1}$, and $\ist{\ell}$ (the maximal independent set of $\Gft{\ell}$ of size greater than $k$).
Although we are interested in producing a sequence of conflict graphs, it will be helpful to think of a process that
alters some ``hybrid'' client-facility graph $\Gcf$, then uses $\Gcf$ and the
corresponding opening times $t$ to construct a new conflict graph $\Gf$ after
each alteration.   To ease the description of this process, we duplicate each
facility that appears both in $\Gcft{\ell}$ and $\Gcft{\ell+1}$ so as to
ensure that these sets are disjoint. Let $\vfact{\ell}$ and $\vfact{\ell+1}$
denote the (now disjoint) sets of facilities in $\Gcft{\ell}$ and $\Gcft{\ell+1}$,
respectively. Note that the duplication of facilities
does not alter the solution space of the considered instance, as one may
assume that at most one facility is opened at
each location.  Note that our algorithm will also satisfy this property, since $d(i,i')^2 = 0$ for any pair of co-located facilities $i,i'$. 

Initially, we let $G$   be the client-facility graph with bipartition $\clients$ and $\vfac^{(\ell)}
\cup \vfac^{(\ell+1)}$ that has an edge from client $j$ to  facility $i\in
\vfac^{(\ell)}$ if $(j,i)$ is present in $G^{(\ell)}$ and to $i\in
\vfac^{(\ell+1)}$  if $(j,i)$ is present in $G^{(\ell+1)}$.
The opening time $t_i$ of facility $i$ is
now naturally set to $t^{(\ell)}_i$ if $i\in \vfac^{(\ell)}$ and to
$t^{(\ell+1)}_i$ if $i\in \vfac^{(\ell+1)}$.  Informally, $G$ is the union of
the two client-facility graphs $G^{(\ell)}$ and $G^{(\ell+1)}$ where the client
vertices are shared (see Figure~\ref{fig:quasigraph}).  We then generate\footnote{Recall that a conflict graph is defined in terms of a client-facility graph $\Gcf$ and $t$: the vertices are the facilities in $\Gcf$, and two facilities $i$ and $i'$ are adjacent if there is some client $j$ that is adjacent to both of them in $\Gcf$ and $d(i,i')^2 \leq \delta \min(t_i, t_{i'})$.} the conflict graph $\Gft{\ell,1}$ from $\Gcf$ and
$t$.   As the induced subgraph of $\Gft{\ell,1}$ on vertex set $\vfac^{\ell}$
equals $\Gft{\ell} = \Gft{\ell,0}$, we have that $\ist{\ell}$ is also an
independent set of $\Gft{\ell,1}$.  We obtain a maximal independent set
$\ist{\ell,1}$ of $\Gft{\ell,1}$ by greedily extending $\ist{\ell}$.  Clearly,
the independent set can only increase so we still have $|\ist{\ell,1}| > k$.

To produce the remaining sequence, we iteratively perform changes, but
construct and output a new conflict graph and maximal independent set after
\emph{each} such change.  Specifically,  we remove from $\Gcf$ each facility $i
\in \vfact{\ell}$, one by one.  At the end of the procedure (after
$|\vfact{\ell}|$ many steps), we have $\Gcf = \Gcft{\ell+1}$ and so $\Gft{\ell,
p_\ell}  = \Gft{\ell+1}$.  Note that at each step, our modification to $\Gcf$
results in removing a single facility $i$ from the associated conflict graph.
Thus, if $\ist{\ell,s}$ is an independent set in $\Gft{\ell,s}$ before
a modification, then $\ist{\ell,s} \setminus \{i\}$ is an independent set  in $\Gft{\ell, s+1}$.  We obtain a maximal
independent set $\ist{\ell, s+1}$ of $\Gft{\ell, s+1}$ by greedily extending
$\ist{\ell,s} \setminus \{i\}$.  Then, for each step $s$, we have $|\ist{\ell,
s+1}| \geq |\ist{\ell, s}| -1$, as required.

\subsubsection{Analysis}
\label{sec:quasigraphupdateanalysis}
The total running time is $n^{O(\aeps^{-1} \log n)}$ since the number of steps $L$
(and the number of dual solutions in our sequence) is  $n^{O(\aeps^{-1} \log n)}$  and
each step runs in polynomial time since it involves the construction of at most
$O(|\cF|)$ conflict graphs and maximal independent sets. 

We proceed to analyze the approximation guarantee. 
Consider the first time that we produce some  maximal independent set $\is$ of size
exactly $k$.   Suppose that when this happened,
we were moving between two solutions $\alphat{\ell}$ and $\alphat{\ell+1}$,
i.e., $\is = \ist{\ell, s}$ is a maximal independent set of $\Gft{\ell,s}$ for
some $1 \leq s \leq p_\ell$.  That we may assume that $s \geq 1$ follows from
$|\ist{0}| > k$ and $\ist{\ell-1, p_\ell} = \ist{\ell} = \ist{\ell,0}$
(recall that $\is$ was selected to be the \emph{first} independent set of size
$k$). 

To ease notation, we let $\Gf =  \Gft{\ell, s}$ and denote by $\Gcf$ the ``hybrid'' client-facility graph that generated $\Gf$. In order to analyze the cost of $\is$, let us form a hybrid solution $\alpha$ 
by setting $\alpha_j = \min(\alphat{\ell}_j,\alphat{\ell+1}_j)$ for each client
$j \in \clients$.  Note that $\alpha \leq \alphat{\ell}$ is a feasible solution of $\dualf[\lambda]$ where $\lambda = \ell \cdot \zeps$ and, since $\alphat{\ell}$ and $\alphat{\ell+1}$ are close,
$\alpha_j \ge \alphat{\ell}_j - \frac{1}{n^2}$ and $\alpha_j \ge
\alphat{\ell+1}_j - \frac{1}{n^2}$ for all $j$.  For each client $j$, we define a set of
facilities $S_j \subseteq \is$ to which $j$ contributes, as follows.  For all
$i \in \is$, we have $i \in S_j$ if $\alpha_j > d(j,i)^2$.  Note that $S_j$ is a subset of $j$'s neighborhood in $G$ and therefore 
\begin{align*}
  \alpha_j  = \min(\alphat{\ell}_j, \alphat{\ell+1}_j) \leq t_i  = 
  \begin{cases}
    \timet{\ell}_i & \mbox{if $i\in \vfact{\ell}$} \\
    \timet{\ell+1}_i & \mbox{if $i\in \vfact{\ell+1}$} 
  \end{cases}
    \qquad  \mbox{for all $i\in S_j$.}
\end{align*}

Using the fact that $\alphat{\ell+1}$ is a good dual solution, we can bound the total
service cost of all clients in the integral solution $\is$.  Let us first
proceed separately for those clients with $|S_j| >  0$.  Let $\clients_0 = \{j
\in \clients : |S_j| = 0\}$, and $\clients_{> 0} = \clients \setminus
\clients_0$. We remark that the  analysis is now very similar to the proof of Theorem~\ref{thm:euckmeans}. 
We define $\beta_{ij} = [\alpha_j - d(i,j)^2]^+$ and similarly  
 $\betat{\ell}_{ij} = [\alphat{\ell}_j - d(i,j)^2]^+$ and $\betat{\ell+1}_{ij}
 = [\alphat{\ell+1}_j - d(i,j)^2]^+$.
\begin{lemma}\label{lem:quasi-connected-client-costs}
  For any $j\in \clients_{>0}$, $d(j,\is)^2 \leq \rho \cdot \left( \alpha_j - \sum_{i \in S_j}\beta_{ij}\right)$.
\end{lemma}
\begin{proof}
  Consider some $j \in \clients_{>0}$ and first suppose that $|S_j| = 1$.  Then, if we let $S_j = \{i\}$,
  $\alpha_j = \beta_{ij} + d(j,i)^2 \ge \beta_{ij} + d(j,\is)^2$ just
  as in ``Case $s=1$'' of Theorem~\ref{thm:euckmeans}.   Next, suppose that
  $|S_j| = s > 1$.  In other words, $j$ is contributing to multiple facilities
  in $\is$.  By construction we have $\alpha_j \leq \min(t_i, t_{i'})$ for any two facilities $i, i' \in S_j$.  Thus, $\alpha_j -  \sum_{i\in S_j}\beta_{ij}
  \geq \frac{1}{\rho}d(j,\is)^2$ by the exact same arguments as in ``Case
  $s>1$'' of Theorem~\ref{thm:euckmeans}.
\end{proof}
Next, we bound the total service cost of all those clients that do not
contribute to any facility in $\is$. The proof is very similar to ``Case
$s=0$'' in the proof of Theorem~\ref{thm:euckmeans}.
\begin{lemma}\label{lem:quasi-disconnected-client-costs}
  For every $j\in \clients_0$, $d(j,\is)^2 \leq (1+5\aeps)\rho \cdot \alpha_j$.
\end{lemma}
\begin{proof}
Consider some client $j \in \clients_0$, and let $i \in \vfact{\ell+1}$ be
a tight facility  so that 
\begin{align*}
  (1+\sqrt{\delta} + \aeps) \sqrt{\alphat{\ell+1}_j} \geq d(j, i) + \sqrt{\delta \timet{\ell+1}_i}\,. 
\end{align*}
Such a facility $i$ is guaranteed to exist because $\alphat{\ell+1}$ is a good dual solution. 
Furthermore, note that  $i$ is present in $\Gf$ since $\Gf$ contains all facilities in $\vfact{\ell+1}$.
By definition $t_i = \timet{\ell+1}_i$ and, as  all $\alpha$-values are at least $1$ (by the preprocessing of Lemma~\ref{lem:dis}), $(1+\frac{1}{n^2}){\alpha_j} \geq  {\alpha_j} + 1/n^2 \geq {\alphat{\ell+1}_j}$. Hence, the above inequality implies
\begin{align*}
(1+\tfrac{1}{n^2})^{1/2}(1+\sqrt{\delta} + \aeps) \sqrt{\alpha_j} \geq d(j, i) + \sqrt{\delta t_i}\,. 
\end{align*}
Note the similarity of this inequality with that of~\eqref{eq:maineqkmeans} and the proof is now identical to ``Case $s=0$'' of Theorem~\ref{thm:euckmeans}.

 Indeed, since $\is$ is a maximal independent set of $\Gf$, either $i\in \is$, in which case $d(j,\is) \leq d(j,i)$, or there is a $i'\in \is$ such that the edge $(i',i)$ is in $\Gf$, in which case
 \begin{align*}
   d(j,\is) \leq d(j,i) + d(i,i') \leq d(j,i) + \sqrt{\delta t_i}\,,
 \end{align*}
 where the inequality follows from $d(i,i')^2  \leq \delta \min(t_i,
 t_{i'})$ by the definition of $\Gf$. In any case, we have  (using $n \gg 1/\aeps$)
 \begin{align*}
 d(j, \is) \leq (1+\tfrac{1}{n^2})^{1/2}(1+\sqrt{\delta} + \aeps)\sqrt{\alpha_j} \leq   (1+2\aeps) (1+\sqrt{\delta})\sqrt{\alpha_j}\,.
 \end{align*} Squaring both sides and recalling that $\rho \geq (1+\sqrt{\delta})^2$ and that $\aeps$ is a small constant so $(1+2\aeps)^2 \leq (1+5\aeps)$ completes the 
 proof of the lemma.
\end{proof}
One difference compared to the analysis in Section~\ref{sec:euckmeans} is that not all opened facilities are fully paid for. However, they are almost paid for: 
\begin{lemma}
  For any $i\in \is$, $\sum_{j\in \clients} \beta_{ij}  \geq \lambda - \tfrac{1}{n}$.
  \label{lem:quasifacilitycost}
\end{lemma}
\begin{proof}
  If $i\in \vfact{\ell+1}$, then it is a tight facility with respect to $\alphat{\ell+1}$, i.e.,
  $\sum_{j\in \clients} \betat{\ell+1}_{ij} = \lambda + \zeps$. Similarly, if $i \in \vfact{\ell}$ then $\sum_{j\in \clients} \betat{\ell}_{ij} = \lambda$. Now since
  $\alpha_j \geq \max(\alphat{\ell+1}_j, \alphat{\ell}_j) - \tfrac{1}{n^2}$ for every client $j$, 
  \begin{align*}
    \sum_{j\in \clients} \beta_{ij} \geq  \sum_{j\in \clients} \left(\max(\betat{\ell+1}_{ij}, \betat{\ell}_{ij})- \tfrac{1}{n^2} \right)  \geq \lambda - \tfrac{1}{n}\,. & \qedhere
  \end{align*}
\end{proof}
We now combine the above lemmas to bound the approximation guarantee of the found solution. Recall that $\opt_k$ denotes the optimum value of the standard LP-relaxation (see Section~\ref{sec:prelim}).
\begin{theorem}
  \label{thm:quasiapproxbound}
$\sum_{j \in \cD} d(j,\is)^2 \le (1+6\aeps)\rho\cdot \opt_k$.
\end{theorem}
\begin{proof}
From Lemmas \ref{lem:quasi-connected-client-costs}  and \ref{lem:quasi-disconnected-client-costs} we have:
\begin{align*}
  \sum_{j \in \clients}d(j,\is)^2 \le (1+5\aeps)\rho \sum_{j \in \clients}\left(\alpha_j - \sum_{i \in S_j}\beta_{ij}\right)\,. 
\end{align*}
By Lemma~\ref{lem:quasifacilitycost} (note that by definition, $\sum_{i \in \is} \beta_{ij} = \sum_{i\in S_j} \beta_{ij}$),
\begin{align*}
  \sum_{j \in \clients}\left(\alpha_j - \sum_{i \in S_j}\beta_{ij}\right) \leq \sum_{j\in \clients} \alpha_j - |\is| \left(\lambda - \tfrac{1}{n}\right) = \sum_{j\in \clients} \alpha_j - k \cdot \lambda + \tfrac{k}{n}\leq \opt_k + 1\,,
\end{align*}
where the last inequality follows from $k\leq n$ and, as $\alpha$ is a feasible solution to $\dualf$, $\sum_{j\in \clients} \alpha_j - k \cdot \lambda \leq \opt_k$. 
The statement now follows from $\opt_k \geq \sum_{j\in
\clients} \min_{i\in \facilities} d(i,j)^2 \geq n$ and $n \gg 1/\aeps$, which imply that $\opt_k + 1 \leq (1 + \aeps)\opt_k$.

\end{proof}

We have thus proved that our quasi-polynomial algorithm produces a  $(\rho + O(\epsilon))$-approximate solution which implies Theorem~\ref{thm:maineuckmeans}. The quasi-polynomial algorithms for the other considered problems are the same except for the selection of $\delta$ and $\rho$, and that in the $k$-median problem the connection costs are the (non-squared) distances.

\section{Polynomial time algorithm}
\label{sec:polyn-time-appr}
We now show how to obtain a polynomial-time algorithm, building on the ideas  presented in the previous section.  As in Section \ref{sec:quasi},  we focus exclusively on the $k$-means problem, and let $\delta = \dmean \approx 2.3146$ and $\rho = \rmean = (1 + \sqrt{\delta})^2 \approx 6.3574$, and assume that the squared-distances between clients and facilities are in $[1, n^6]$ by Lemma~\ref{lem:dis}.  Additionally, we choose $\aeps$ and $\gamma$ to be suitably small constants with $0<\gamma \ll \aeps \ll 1$, and for notational convenience we assume without loss of generality that $n\gg 1/\gamma$.

Similarly to Section~\ref{sec:quasi-polyn-time-algor}, we give an algorithm for generating a close sequence of feasible solutions to $\dualf$, and then show how to use this sequence to generate a sequence of integral solutions that must contain some solution of size exactly $k$.  Here, however, we ensure that our sequence of feasible solutions is of polynomial length.  In order to accomplish this, we must relax some of the requirements in our definition of a good solution (Definition \ref{def:quasi-good-soln}). 

First, rather than requiring that every facility has opening cost $\lambda$, we instead allow each facility $i$ to have its own price in $z_i \in \{\lambda, \lambda+\tfrac{1}{n}\}$ (Condition~\ref{item:roundable-1} of Definition~\ref{def:roundable}).  For each $\alphat{\ell}$, our algorithm will produce an associated set of facility prices $\zt{\ell} = \{\zt{\ell}_i\}_{i \in \cF}$.  For any such $(\alphat{\ell},\zt{\ell})$, we define $\betat{\ell}_{ij} = [\alphat{\ell}_j - d(j,i)^2]^+$ and $\Nt{\ell}(i) = \{j\ :\ \betat{\ell}_{ij} > 0 \}$, as before.  However, we now say that a facility $i$ is \emph{tight} in $\solt{\ell}$ if $\sum_{j \in \clients}\betat{\ell}_{ij}= \zt{\ell}_i$.  That is, we consider a facility $i$ tight once its (possibly unique) price $z_i$ is paid in the dual.  Intuitively, if all the facility prices $z_i$ are \emph{almost} the same, we can still carry out our analysis, and obtain a $(\rho + O(\epsilon))$-approximation.

Second, we shall designate a set of \emph{special} facilities $\sfac \subseteq \facilities$ that we shall open, \emph{even if they are not tight}.  To each special facility $i \in \sfac$ we assign a set of special clients $\sclient(i) \subseteq \clients$ that are allowed to pay for $i$.  Then, for each $i \in \sfac$, we define the time $\stime_i = \max_{j \in N(i) \cap \sclient(i)}\alpha_j$, while for each $i \in \facilities \setminus \sfac$ we set $\stime_i = t_i = \max_{j \in N(i)}\alpha_j$.  Again, we adopt the convention that $\stime_i = 0$ if $N(i) \cap \sclient(i) = \emptyset$ for $i \in \sfac$ or $N(i) = \emptyset$ for $i \in \facilities \setminus \sfac$.  Although a facility in $\sfac$ is not necessarily tight, we shall require that the total of all payments to such facilities by special clients is \emph{almost} equal to $\lambda |\sfac|$ (Condition~\ref{item:roundable-3} of Definition~\ref{def:roundable}).  That is, on average, each facility of $\sfac$ is almost tight.
  
Finally, given the times $\stime_i$, we shall not require that \emph{every} client $j$ has some tight or special facility $i$ such that $(1 + \sqrt{\delta} + 10\epsilon)\sqrt{\smash[b]{\alpha_j}} \ge d(j,i) + \sqrt{\smash[b]{\delta \stime_i}}$.  Specifically, we shall allow some small set of \emph{bad} clients $\bclient$ to instead satisfy a weaker inequality $6\sqrt{\smash[b]{\alpha_j}} \ge d(j,i) + \sqrt{\smash[b]{\delta \stime_i}}$ for some tight or special facility $i$.  Such clients will have a higher service cost, so we require that their total contribution to the cost of an optimal solution is small (Condition~\ref{item:roundable-2} of Definition~\ref{def:roundable}).

Combining the above, we have the following definition.  
\begin{definition}
\label{def:roundable}
Consider a tuple $(\alpha, z, \sfac, \sclient)$ where $\alpha \in \mathbb{R}^\clients, z \in \mathbb{R}^\facilities$, $\sfac \subseteq \facilities$ is a set of special facilities, and $\sclient : \sfac \to 2^\clients$ is a function assigning each special facility $i$ a set of special clients $\sclient(i)$.  We say that this tuple is roundable for $\lambda$ (or $\lambda$-roundable) 
if $\alpha$ is a feasible solution of $\dualf[\lambda + \frac{1}{n}]$, and:
\begin{enumerate}
\item For all $i \in \facilities$, $\lambda \le z_i \le \lambda + \frac{1}{n}$. \label{item:roundable-1}
\item There exists a subset $\bclient$ of clients so that for all $j \in \clients$ there is a facility $w(j)$ that is either tight or in $\sfac$ and:
\begin{enumerate}
\item $(1+\sqrt{\delta} + 10\epsilon)^2\alpha_j \ge \left(d(j,w(j)) +  \sqrt{\smash[b]{\delta \cdot \stime_{w(j)}}}\right)^2$ for all $j \in \clients \setminus \bclient$.
\label{item:good-client-roundable}
\item $36\gamma \cdot \opt_k \ge \sum_{j \in \bclient}\left(d(j,w(j)) + \sqrt{\smash[b]{\delta\cdot \stime_{w(j)}}}\right)^2$,
\label{item:bad-client-roundable}
\end{enumerate} \label{item:roundable-2}
\item $\sum_{i \in \sfac}\sum_{j \in \sclient(i)} \beta_{ij} \ge \lambda |\sfac| - \gamma \cdot \opt_k$ and $|\sfac| \leq n$. \label{item:roundable-3}
\end{enumerate}
\end{definition}

Observe that any $\lambda$-roundable solution with $\sfac=\emptyset$, and $\bclient =\emptyset$ is essentially a good solution for $\dualf[\lambda+\frac{1}{n}]$ (as defined for the quasi-polynomial algorithm in Section~\ref{sec:quasi}) except that the opening costs of the facilities are allowed to vary slightly.  We shall also say that $(\alpha,z)$ is roundable if $(\alpha, z, \emptyset, \sclient)$ is roundable.

An overview of our polynomial time algorithm is shown in Algorithm
\ref{alg:1}.  The algorithm maintains a current base price $\lambda$ and a current roundable solution $\cst{0} = (\alphat{0}, \zt{0}, \sfac^{(0)}, \sclient^{(0)})$ for $\lambda$, as well as a corresponding integral solution $\ist{0}$.  As in the quasi-polynomial algorithm, we shall enumerate a sequence $0,1 \cdot \zeps, 2\cdot\zeps,\ldots,L\cdot \zeps$ of base prices $\lambda$, where now $\zeps = n^{-O(1)}$ and, as before we define $L = 4n^7\cdot \zeps^{-1}$.  Here, however we increase facility prices from $\lambda$ to $\lambda+\zeps$ \emph{one-by-one} using an auxiliary procedure \raiseprice, which takes as input a fractional dual solution $\alphat{0}$, a set of prices $\zt{0}$, a current integral solution $\ist{0}$, and a facility $i$.  \raiseprice increases the price of facility $i$, then outputs a close sequence of
roundable solutions $\cst{1}=\rolt{1},\ldots,\cst{q}=\rolt{q}$, each having $\zt{\ell}_i = \zt{0}_i + \zeps$ and $\zt{\ell}_{i'} = \zt{0}_{i'}$ for all $i' \neq i$.  Note that in addition to increasing the facility prices one-by-one, we now generate a \emph{sequence} of solutions for each individual price increase.

Initially, we set $\lambda \gets 0$ and then initialize $\cst{0}$ by setting $\zt{0}_i \gets 0$ for all $i \in \facilities$ and $\sfac=\emptyset$ (observe that $\sclient$ is then an empty function), and constructing $\alphat{0}$ as follows.  We set $\alpha_j = 0$ for all $j \in \clients$ and then increase all $\alpha_j$ at a uniform rate.  We stop increasing a value $\alpha_j$ whenever $j$ gains a tight edge to some facility $i \in \facilities$ or $2\sqrt{\alpha_j} \ge d(j,j') + 6\sqrt{\alpha_{j'}}$ for some $j' \in \clients$ (the rationale behind this choice will be made clear in Section \ref{sec:polyn-time-algor}).  Finally, we initialize our current integral solution $\ist{0} = \facilities$.

As long as an integral solution of size $k$ has
not yet been produced, Algorithm \ref{alg:1} iterates through each facility
$i \in \facilities$, calling \raiseprice to raise $z_i$
by $\zeps < 1/n$.  The sequences that are produced are used to obtain a sequence of integral solutions in which the number of open facilities decreases by at most $1$.  This is done by using a second procedure, \graphupdate, which is very similar to the procedure \quasigraphupdate described in the previous section.  Note that raise price always increases a single facility $i$'s price by $\zeps < 1/n$, and does not increase $z_i$ further until all other facility prices have also been increased by $\zeps$.  Thus, each in every pair of consecutive solutions $\cst{\ell}, \cst{\ell+1}$ considered by \graphupdate in line \ref{line:graphupdate}, every price $z_i \in \{\lambda, \lambda + \zeps\}$ and so both solutions are $\lambda$-roundable (for the same value $\lambda$).  We describe our auxiliary procedures \graphupdate and \raiseprice in the next sections.  Note that initially  $|\ist{0}| = |\facilities|$ and, by the same reasoning as in Section \ref{sec:quasi-rounding}, once $\lambda = L\cdot \zeps = 4n^7$ we must have $|\ist{0}| = 1$.  Thus, at some intermediate point, we will indeed find some solution $\is$ of size $k$.

\begin{algorithm}[h]
\DontPrintSemicolon
\SetNoFillComment
\caption{Polynomial time $(\rmean + O(\epsilon))$-approximation algorithm for $k$-means}
\label{alg:1}
\small
Initialize $\cst{0}=\rolt{0}$ as described in our discussion above\;
$\lambda \gets 0, \ist{0} \gets \facilities$\;
\For{$\lambda = 0,\,1\cdot\zeps,\,2\cdot\zeps,\,\ldots,\,L\cdot\zeps$}{
  \tcc{Raise the price of each facility $i$ to $\zt{0}_i + \zeps = \lambda + \zeps$}
  \ForEach{$i \in \facilities$}{
    Call $\raiseprice(\alphat{0}, \zt{0}, \ist{0},i)$ to produce a sequence $\cst{1},\ldots,\cst{q}$ of $\lambda$-roundable solutions\;
   \tcc{Move through this sequence, constructing integral solutions}
    \For{$\ell = 0$ \KwTo $q-1$}{
      Call $\graphupdate(\cst{\ell}, \cst{\ell+1},\ist{\ell})$ to produce a sequence $\ist{\ell,0},\ldots,\ist{\ell,p_\ell}$\; \label{line:graphupdate}
      \lIf{\normalfont $|\ist{\ell,r}| = k$ for some $\ist{\ell,r}$ in this sequence}{\Return $\ist{\ell,r}$}
      \lElse{$\ist{\ell+1} \gets \ist{\ell,p_\ell}$}
    }
   \tcc{After each price increase, update current solutions}
    $\cst{0} \gets \cst{q}$, $\ist{0} \gets \ist{q}$\;    \label{line:update}
  }
  \tcc{All prices have been increased. Continue to the next base price $\lambda$}
}
\end{algorithm}

Algorithm~\ref{alg:1} executes $L = 4n^7\cdot \zeps^{-1}$ base price
increases, each of which performs $|\facilities|$ calls to
\raiseprice.  In order to show that Algorithm~\ref{alg:1} runs it
polynomial time, it is sufficient to show that each call to
\raiseprice and \graphupdate produces a polynomial length sequence in
polynomial time.  In the next sections, we describe these procedures
in more detail and show that they run in polynomial time. In addition,
we show that \raiseprice produces a sequence of roundable solutions
(Proposition~\ref{prop:roundable}) that are close (Proposition~\ref{prop:close-values}).
In Section~\ref{sec:rounding}, we show that given these solutions,
\graphupdate finds a $(\rho+1000\epsilon)$-approximate solution
(Theorem~\ref{thm:rounding-approx-ratio}).\footnote{We remark that we have chosen to first describe \graphupdate as that procedure is very similar to
  \quasigraphupdate in the quasi-polynomial algorithm whereas
  \raiseprices is more complex.}  This implies our main
theorem:
\begin{theorem}
  For any $\aeps >0$, there is a $(\rho+\epsilon)$-approximation algorithm for $k$-means.
  \label{thm:mainsec}
\end{theorem}

\section{Opening a set of exactly $k$ facilities in a close, roundable sequence}
\label{sec:rounding}
In this section, we describe our algorithm  \graphupdate for interpolating
between two close roundable solutions  $\cst{\ell}$ and $\cst{\ell+1}$ starting with a maximal independent set $\ist{\ell}$ of the conflict graph\footnote{Below, we slightly generalize the definition of client-facility and conflict graphs in Section~\ref{sec:pddesc} to that of roundable solutions.} $\Gft{\ell}$ of $\cst{\ell}$.
The goal of this procedure is the same as that of \quasigraphupdate explained in Section~\ref{sec:quasi-rounding}: we maintain a sequence of maximal independent sets in appropriately constructed conflict graphs so that the size of the independent set never decreases by more than $1$, and the last solution is a maximal independent set of the conflict graph $\Gft{\ell+1}$ of $\cst{\ell+1}$.
Similar to Section~\ref{sec:quasi-rounding}, we use a ``hybrid'' client-facility
graph to generate our conflict graph in each step of our procedure.  The only
difference is that we need to slightly generalize the definition of a client-facility
graph  to incorporate the concept of roundable solutions. 

\paragraph{Client-facility and conflict graphs of roundable solutions.} We define the
\emph{client-facility} graph $\Gcf$ of a roundable solution $\cS
= (\alpha, z, \sfac, \sclient)$ as in Section~\ref{sec:pddesc} with the following
two changes: First, recall that we now consider a facility $i$ \emph{tight} if and only if
$\sum_{j \in N(i)}\beta_{ij} = z_i$.  Second, we shall additionally add every
facility $i \in \sfac$ to $\Gcf$, but place an edge between each $i \in \sfac$
and $j \in \clients$ only if $j \in N(i) \cap \sclient(i)$.  Intuitively, we treat
special facilities $i \in \sfac$ essentially the same as tight facilities,
except only those clients in $N(i) \cap \sclient(i)$ are considered to be paying for $i$.

Formally, let $\vfac$ denote the set of all tight facilities or special
facilities with respect to $\cS$.  Then, $\Gcf$ is a bipartite graph on
$\clients$ and $\vfac$ that contains an edge $(i,j)$ if and only if $i \in
\vfac \setminus \sfac$ and $j \in N(i)$ or $i \in \sfac$ and $j \in N(i)
\cap \sclient(i)$.  As before, we assign an
\emph{opening time} $\stime_i$ to each $i \in \vfac$.  For $i \in \vfac
\setminus \sfac$, $\stime_i = t_i =   \max_{j \in N(i)}\alpha_j$, and for $i \in \sfac$, $\stime_i = \max_{j \in N(i) \cap \sclient(i)}\alpha_j$.  In other words, $\stime_i$ equals the maximum $\alpha_j$ over all clients $j$ such that $(j,i)$ is an edge in $\Gcf$ (in the case that there is no such edge, we adopt the convention that $\stime_i = 0$).  Note that $\tau_i \leq t_i$ for any facility $i$.

Given a client facility graph $\Gcf$, and a set of opening times $\stime$, we
construct the corresponding \emph{conflict graph} $\Gf$ in the same way as in
Section~\ref{PDalg}:  the vertex set of $\Gf$ is the set of all facilities
appearing in $\Gcf$ and we place an edge between two facilities $i$ and $i'$ in
$\Gf$ if and only if there is some $j \in \clients$ such that both $(j,i)$ and
$(j,i')$ are present in $\Gcf$ and $d(i,i')^2 \le \delta \cdot  \min(\stime_i,
\stime_{i'})$. Notice that this coincides with the definition in
Section~\ref{PDalg} when the set of special facilities is empty. In particular, the initial independent set $\cF$ is a maximal independent set of the conflict graph associated to the initial solution (which has all facilities and no edges). Then, as in each iteration the last constructed  independent set by \graphupdate is given as input in the next call (see Algorithm~\ref{alg:1}), we maintain the property that  the input independent set $\ist{\ell}$ is a maximal independent set of the conflict graph of $\cst{\ell}$.

\paragraph{Description of \graphupdate.} Our algorithm now proceeds in the exact same
way as \quasigraphupdate in Section~\ref{sec:quasi-rounding}. A short description is repeated here for convenience.  Let $\Gcft{\ell}, \stimet{\ell}$ and
$\Gcft{\ell+1}, \stimet{\ell+1}$ be the client-facility graphs and times
associated with $\cst{\ell}$ and $\cst{\ell+1}$, respectively. Furthermore, let $\Gft{\ell}$ and $\Gft{\ell+1}$ be the conflict graphs generated by  $\Gcft{\ell}, \stimet{\ell}$ and
$\Gcft{\ell+1}, \stimet{\ell+1}$. Recall that the input to \graphupdate is $\cst{\ell}, \cst{\ell+1}$ and a maximal independent set $\ist{\ell}$ of $\Gft{\ell}$.

Define the
``hybrid'' client-facility graph $G$ as  the union of $\Gcft{\ell}$ and
$\Gcft{\ell+1}$ where the client vertices are shared and the facilities are
duplicated if necessary so as to make sure that the facilities of $\Gcft{\ell}$ and $\Gcft{\ell+1}$ are disjoint.  The opening times are defined by
\begin{align*}
  \stime_i = \begin{cases}
    \stimet{\ell}_i  & \mbox{ if $i\in \vfact{\ell}$} \\
    \stimet{\ell+1}_i & \mbox{ if $i\in \vfact{\ell+1}$} 
  \end{cases}\,,
\end{align*}
where $\vfact{\ell}$ and $\vfact{\ell+1}$ denote the (disjoint) sets of
facilities in $\Gcft{\ell}$ and $\Gcft{\ell+1}$, respectively.
We then generate the conflict graph $\Gft{\ell,1}$ from $\Gcf$ and
$\stime$.   As the induced subgraph of $\Gft{\ell,1}$ on vertex set $\vfac^{\ell}$
equals $\Gft{\ell} = \Gft{\ell,0}$, we have that the given maximal independent set $\ist{\ell}$ of $\Gft{\ell}$ is also an
independent set of $\Gft{\ell,1}$.  We obtain a maximal independent set
$\ist{\ell,1}$ of $\Gft{\ell,1}$ by greedily extending $\ist{\ell}$.   
We then obtain the remaining conflict graphs and independent sets by removing
from $\Gcf$ each facility $i\in \vfact{\ell}$, one by one. After each step
we generate the associated conflict graph and we greedily extend the previous independent
set (with $i$ potentially removed) so as to obtain a maximal independent set in the
updated conflict graph. This results, as in Section~\ref{sec:quasi-rounding}, in the sequence $\Gft{\ell} = \Gft{\ell,0}, \Gft{\ell,1}, \dots, \Gft{\ell, p_\ell} = \Gft{\ell+1}$ of $|\vfact{\ell}|+2$ many conflict graphs 
and  a sequence $\ist{\ell} = \ist{\ell,0}, \ist{\ell,1}, \dots, \ist{\ell, p_\ell} = \ist{\ell+1}$ of associated maximal independent sets so that $|\ist{\ell, s}| \geq |\ist{\ell, s-1}| -1$ for any $s=1,\dots, p_\ell$. The output of \graphupdate is this sequence of independent sets.

\subsection{Analysis}
\label{sec:obtaining}
\graphupdate clearly runs in polynomial time since the number of steps is  
polynomial and each step requires only the construction of a conflict graph and greedily maintaining a maximal independent set. 

We proceed to analyze the approximation guarantee.
In comparison to Section~\ref{sec:quasigraphupdateanalysis}, our analysis here
is slightly more involved because it is with respect to roundable solutions
instead of good solutions. In addition, we prove that \emph{all} independent
sets constructed in Algorithm~\ref{alg:1} (by calls to \graphupdate) of size at least $k$ have small
connection cost.   Specifically, we show that any constructed independent set $\is$ with $|\is| \geq k$
has $\sum_{j\in\clients} d(j,\is)^2 \leq (\rho + O(\aeps)) \opt_k$. 

First note that the initial independent set $\ist{0}$ of
Algorithm~\ref{alg:1} contains all facilities and hence $\sum_{j\in \clients}
d(j,\ist{0})^2 \leq \opt_k$.  All other independent sets are
constructed by calls to \graphupdate. Consider one such independent set $\is$
with $|\is| \geq k$ and  consider the first time this independent set was constructed. Suppose that when this happened,
we were moving between two solutions $\cst{\ell}$ and $\cst{\ell+1}$ that are
roundable for  the same $\lambda$.  Then, $\is=\ist{\ell,s}$ for some step $s\geq 1$ of
\graphupdate.  We may assume $s \geq 1$ because $\ist{\ell,0}
= \ist{\ell}$ was constructed in the previous call to \graphupdate (or it equals the initial independent set).
Let $\Gcf$ and $\stime$ be the client-facility graph and the opening times that generated the conflict graph $H = \Gft{\ell, s}$ in which $\is = \ist{\ell, s}$ is a maximal independent set. Also note that we may assume, without loss of generality, that $|\is| \leq n$.  Otherwise, we can reduce the size of $\is$ 
since the connection cost of $\is$ equals that of $\bigcup_{j\in \clients}
\{\argmin_{i\in \is} d(j,i)\}$.

Similar to Section~\ref{sec:quasigraphupdateanalysis}, we analyze the cost of $\is$ with respect to a hybrid solution $\alpha$ obtained 
by setting $\alpha_j = \min(\alphat{\ell}_j,\alphat{\ell+1}_j)$ for each client
$j \in \clients$.  The following observations and concepts are also very similar to the ones in that section. We remark that $\alpha$ is a feasible solution of $\dualf[\lambda
+ \tfrac{1}{n}]$ and, since $\alphat{\ell}$ and $\alphat{\ell+1}$ are close,
$\alpha_j \ge \alphat{\ell}_j - \frac{1}{n^2}$ and $\alpha_j \ge
\alphat{\ell+1}_j - \frac{1}{n^2}$.  For each client $j$, we define a set of
facilities $S_j \subseteq \is$ to which $j$ contributes, as follows.  For all
$i \in \is$, we have $i \in S_j$ if $\alpha_j > d(j,i)^2$ and $(j,i)$ is an edge
in $G$.  Note that $S_j$ is a subset of $j$'s neighborhood in $G$ and therefore 
\begin{align}
  \label{eq:times}
  \alpha_j  = \min(\alphat{\ell}_j, \alphat{\ell+1}_j) \leq \stime_i \qquad  \mbox{for all $i\in S_j$.}
\end{align}

Using the fact that $\cst{\ell+1}$ is roundable, we can bound the total
service cost of all clients in the integral solution $\is$.  Let us first
proceed separately for those clients with $|S_j| >  0$.  Let $\clients_0 = \{j
\in \clients : |S_j| = 0\}$, and $\clients_{> 0} = \clients \setminus
\clients_0$.  The following lemma is identical to Lemma~\ref{lem:quasi-connected-client-costs} and its proof  is therefore omitted.

\begin{lemma}
  \label{lem:connected-client-costs}
  For any $j\in \clients_{>0}$, $d(j,\is)^2 \leq \rho \cdot \left( \alpha_j - \sum_{i \in S_j}\beta_{ij}\right)$.
\end{lemma}

                 Next, we bound the total service cost of all those clients that do not
contribute to any facility in $\is$. The proof is very similar to that of Lemma~\ref{lem:quasi-disconnected-client-costs}
except that we also need to
handle the bad clients in \bclient. 
\begin{lemma}\label{lem:disconnected-client-costs}
$\sum_{j \in \clients_0}d(j,\is)^2 \le (\rho+ 200\epsilon)\sum_{j \in \clients_0}\alpha_j + 36\gamma \cdot \opt_k$.
\end{lemma}
\begin{proof}
Consider some client $j \in \clients_0$, and let $w(j) \in \vfact{\ell+1}$ be
the tight or special facility for $j$ corresponding to the roundable solution
$\cst{\ell+1}$. Note that  $w(j)$ is present in $\Gf$ (since $\Gf= \Gft{\ell,s}$ with $s\geq 1$ contains all facilities in $\vfact{\ell+1}$) and $\stime_{w(j)}
= \stimet{\ell+1}_{w(j)}$ by definition.   Thus, since $\is$ is a maximal independent set of $\Gf$, either $w(j)\in \is$, in which case $d(j,\is) \leq d(j,w(j))$, or  there must be some
other facility $i \in \is$ such that $\Gf$ contains the edge $(i,w(j))$, in which case
 \begin{align*}
   d(j,\is) \leq d(j,w(j)) + d(w(j),i) \leq d(j,w(j)) + \sqrt{\delta \stime_i}\,,
 \end{align*}
 where the last inequality follows from the fact that $w(j)$ and $i$ are adjacent in $H$ and thus $d(w(j),i)^2  \leq \delta \min(\stime_i,\stime_{i'})$ by the definition of $\Gf$. In any case, we have   $d(j,\is) \le d(j,w(j)) + \sqrt{\delta
 \stime_i}$ with $\stime_i = \stimet{\ell+1}_i$,  and so:
\begin{align*}
\sum_{j \in \clients_0}d(j,\is)^2 &= \sum_{j \in \clients_0 \setminus \bclient}d(j,\is)^2 + \sum_{j \in \clients_0 \cap \bclient}d(j,\is)^2 \\
&\le \sum_{j \in \clients_0 \setminus \bclient} \left(d(j,w(j)) + \sqrt{\delta \cdot \stimet{\ell+1}_{w(j)}}\right)^2 + 
\sum_{j \in \bclient} \left(d(j,w(j)) + \sqrt{\delta \cdot \stimet{\ell+1}_{w(j)}}\right)^2 \\
&\le \sum_{j \in \clients_0 \setminus \bclient} (1+\sqrt{\delta} + 10\aeps)^2\cdot \alphat{\ell+1}_j + 36\gamma \cdot \opt_k\,,
\end{align*}
where the final inequality follows from the fact that $\cst{\ell+1}$ is roundable.  The statement now follows since\footnote{Here we assume without loss of generality that $\alpha_j \geq 1$ for every client $j$. That this is without loss of generality follows from the fact that the distance from any client to a facility is at least $1$ (Lemma~\ref{lem:dis}). In particular, any solution produced by Algorithm~\ref{alg:1} satisfies this (see Invariant~\ref{inv:feasibility}).} $\alphat{\ell+1}_j \leq \alpha_j + 1/n^2 \leq (1 + 1/n^2)\alpha_j$ for all $j \in \clients$, $\epsilon\leq 1, \sqrt{\delta} \leq 2$, and $\rho \ge (1 + \sqrt{\delta})^2$.
\end{proof}
We now bound the contributions to the opened facilities as in
Lemma~\ref{lem:quasifacilitycost} except that we also need to handle the
special facilities. 
\begin{lemma}For any $i\in \is \setminus (\sfact{\ell} \cup \sfact{\ell+1})$, we have $\sum_{j\in \clients} \beta_{ij}  \geq  \lambda - \tfrac{1}{n}$ and for any $i \in \sfact{x}$ for some $x\in \{\ell, \ell+1\}$, we have $\sum_{j\in \sclientt{x}(i)} \beta_{ij}  \geq\left(\sum_{j\in \sclientt{x}(i)}  \betat{x}_{ij} \right) - \tfrac{1}{n}$.
  \label{lem:facilitycost}
\end{lemma}
\begin{proof}
  For the first bound, consider a facility $i\in \is \setminus (\sfact{\ell} \cup \sfact{\ell+1})$ and let $x\in \{\ell, \ell+1\}$ be such that $i \in \vfact{x}$. Then $i$ is a tight facility with respect to $\solt{x}$, i.e.,
  $\sum_{j\in \clients} \betat{x}_{ij} = \zt{x}_i$. As $\cst{x}$ is roundable for $\lambda$, we have $\zt{x}_i \geq \lambda$. Moreover,  $\alpha_j \geq \alphat{x}_j - \tfrac{1}{n^2}$ for every client $j$, and so
  \begin{align*}
    \sum_{j\in \clients} \beta_{ij} \geq  \sum_{j\in \clients} \left(\betat{x}_{ij}- \tfrac{1}{n^2}\right)   \geq \lambda - \tfrac{1}{n}\,. 
  \end{align*}
  Now consider a special facility $i\in \sfact{x}$ for some $x\in \{\ell, \ell+1\}$.  Then, by again using that $\alpha_j \geq \alphat{x}_j - \tfrac{1}{n^2}$ for every client $j$, 
  \begin{align*}
    \sum_{j\in \sclientt{x}(i)} \beta_{ij} \geq  \sum_{j\in \sclientt{x}(i)} \left(\betat{x}_{ij}- \tfrac{1}{n^2}\right)\,,
  \end{align*}
  and the lemma follows since $|\sclientt{x}(i)| \leq |\clients| = n$.
\end{proof}
We are now ready to prove our main result, which bounds the connection cost of $\is$ in terms of $\opt_k$ as desired. The proof is very similar to the proof of Theorem~\ref{thm:quasiapproxbound}.
\begin{theorem}
\label{thm:rounding-approx-ratio}
For any $\is$ produced by \graphupdate with $|\is| \geq k$,
\[\sum_{j \in \cD} d(j,\is)^2 \le (\rho + 1000\epsilon)\cdot \opt_k.\]
\end{theorem}
\begin{proof}
From Lemmas \ref{lem:connected-client-costs}  and \ref{lem:disconnected-client-costs} we have:
\begin{align}
\sum_{j \in \clients}d(j,\is)^2 \le (\rho + 200\epsilon)\left(\sum_{j \in \clients}\alpha_j - \sum_{i \in S_j}\beta_{ij}\right) + 36\gamma \cdot \opt_k\,.
\label{eq:cost-main}
\end{align}
Note that by definition, if 
$i \not\in  \sfact{\ell} \cup \sfact{\ell+1}$ then 
$\sum_{j\in \clients} \beta_{ij} = \sum_{j : i \in S_j} \beta_{ij}$ and if 
$i \in \sfact{x}$ then $\sum_{j\in \sclientt{x}(i)} \beta_{ij} = \sum_{j : i \in S_j} \beta_{ij}$.  
Also, recall that by our construction of $H$, $\sfact{\ell}$ and $\sfact{\ell+1}$ are distinct.  Thus, by Lemma~\ref{lem:facilitycost},
\begin{align*}
  \sum_{j \in \clients}\left(\alpha_j - \sum_{i \in S_j}\beta_{ij}\right) & \leq \sum_{j\in \clients} \alpha_j - |\is\setminus (\sfact{\ell} \cup \sfact{\ell+1})| \left(\lambda - \tfrac{1}{n}\right) - \sum_{x\in \{\ell, \ell+1\}} \sum_{i\in \sfact{x} \cap \is} \left(\sum_{j\in \sclientt{x}(i)} \betat{x}_{ij} - \tfrac{1}{n}\right) \\
  & \leq  \sum_{j\in \clients} \alpha_j - |\is\setminus (\sfact{\ell} \cup \sfact{\ell+1})| \lambda  - \sum_{x\in \{\ell, \ell+1\}} \sum_{i\in \sfact{x} \cap \is} \sum_{j\in \sclientt{x}(i)} \betat{x}_{ij} + \frac{|\is|}{n}\,.
\end{align*}
Since $\cst{x}$ is roundable for $x\in \{\ell, \ell+1\}$, we have $\sum_{i\in
  \sfact{x}} \sum_{j\in \sclientt{x}(i)} \betat{x}_{ij} \geq \lambda |\sfact{x}
  | - \gamma \cdot \opt_k$. Moreover, as $\alphat{x}$ is a feasible solution of
  $\dualf[ \lambda+ \tfrac{1}{n}]$, we have that $\sum_{j\in \sclientt{x}(i)}
  \betat{x}_{ij} \leq \lambda + \tfrac{1}{n}$ for any $i\in \sfact{x}$. Therefore,
  \begin{align*}
    \sum_{i\in \sfact{x}\cap \is} \sum_{j\in \sclientt{x}(i)} \betat{x}_{ij} \geq \lambda | \sfact{x} \cap \is| - \tfrac{|\sfact{x}  \setminus \is|}{n}  - \gamma \cdot \opt_k
 \geq \lambda | \sfact{x} \cap \is| - 2\gamma \cdot \opt_k\,.
  \end{align*} 
  where for the final inequality we use
  that $|\sfact{x}|\leq n \leq \opt_k$, which follows from Definition~\ref{def:roundable}, 
  the fact that  any client has distance at least $1$ to its closest facility, and $1/n \ll \gamma$.  Combining this with the above inequalities yields
\begin{align*}
  \sum_{j \in \clients}\left(\alpha_j - \sum_{i \in S_j}\beta_{ij}\right) & \leq  \sum_{j\in \clients} \alpha_j - |\is| \lambda  + 4\gamma \cdot \opt_k + \frac{|\is|}{n}\\
  & = \sum_{j\in \clients} \alpha_j - |\is|(\lambda + \tfrac{1}{n})  + 4\gamma \cdot \opt_k + \frac{2|\is|}{n} \\
  & \leq  \opt_k +  4\gamma \cdot \opt_k + \frac{2|\is|}{n} \leq (1+5\gamma ) \opt_k \,,
\end{align*}
where we in the penultimate inequality used that $\alpha$ is a feasible solution to \dualf[\lambda + \tfrac{1}{n}] and $|\is|\geq k$, therefore $\sum_{j\in \clients} \alpha_j - |\is| (\lambda + \tfrac{1}{n}) \leq \sum_{j\in \clients} \alpha_j - k (\lambda + \tfrac{1}{n}) \leq \opt_k$; and, in the last inequality, we used that  $\gamma\cdot \opt_k \geq \gamma n \geq 2$ and the assumption that $|\is|\leq n$.

We conclude the proof by substituting this bound in~\eqref{eq:cost-main}: 
\begin{align*}
  \sum_{j \in \clients}d(j,\is)^2 \le (\rho + 200\epsilon) (1+5 \gamma) \opt_k  + 36 \gamma \cdot \opt_k  \leq (\rho + 1000\aeps) \opt_k\,. & \qedhere
\end{align*}

\end{proof}

\section{The algorithm \raiseprice}
\label{sec:polyn-time-algor}
In this section, we give the details of the algorithm \raiseprice, which
is responsible for raising facility prices and generating sequences of roundable
solutions in Algorithm \ref{alg:1}.  It is based on similar insights as used in the
quasi-polynomial algorithm described in Section~\ref{sec:quasi}.  Let us first provide
a high-level overview of our approach.  Recall that in our analysis of that
procedure, changing the values $\alpha_j$ in some bucket $b$ by $\zeps$
roughly required changing the values in bucket $b+1$ by up to $n \zeps$.  Because
there were $\Omega(\log(n))$ buckets, the total change in the last bucket was
potentially $\zeps n^{\Omega(\log n)}$, and so to obtain a close sequence
of $\alpha$-values, we required $\zeps= n^{-\Omega(\log n)}$ in that section.  
Here, we reduce the dependence on $n$ by changing the way in which we increase the opening price $z$.  As in the quasi-polynomial procedure, our algorithm
repeatedly increases the opening cost of every facility from $\lambda$ to
$\lambda + \zeps$, for some appropriate small increment $\zeps = n^{-O(1)} < \aeps$.  However,
instead of performing each such increase for every facility at once, we instead
increase only a single facility's price at a time.  Each such increase will
still cause some clients to become unsatisfied (or undecided as we shall call them), and so we must repair the
solution.  In contrast to the quasi-polynomial procedure, \raiseprice repairs
the solution over a series of \emph{stages}.  We show this will result in a \emph{polynomial length} sequence of close, roundable solutions.

\textbf{Notation: }Throughout this section, we let $z_i$ denote the current price for a facility $i \in \facilities$, where always $z_i \in \{\lambda,  \lambda + \zeps\}$.  We shall now say that $i$ is \emph{tight} if $\sum_{j \in \clients}\beta_{ij} = z_i$, where as before for a solution $\alpha$, we use $\beta_{ij}$ as a shorthand for $[\alpha_{j} - d(j,i)^2]^+$.  It will also be convenient to denote $\sqrt{\alpha_j}$ by $\sqrtalpha_j$.  Note that $\beta_{ij} > 0$, if and only if $\sqrtalpha_j > d(j,i)$.  As in the quasi-polynomial procedure, we shall divide the range of possible values for $\alpha_j$ into buckets: we define $B(v) = 1 + \lfloor \log_{1+\aeps} v \rfloor$ for any $v \ge 1$ and $B(v) = 0$, for all $v \le 1$.

To control the number of undecided (unsatisfied) clients, it will be important
to control the way clients may be increased and decreased throughout our
algorithm.  To accomplish this, we shall not insist that every client has some
tight witness in every vector $\alpha$ that we produce (in contrast to
Invariant~\ref{inv:quasi} in the quasi-polynomial algorithm).  Rather, we shall
consider several different types of clients:
\begin{itemize}
  \item \emph{witnessed clients} $j$ have a tight edge to some tight facility $i$ with $B(\alpha_j) \ge B(t_i)$.  In this case, we say that $i$ is a \emph{witness} for $j$.  Note that if $i$ is a witness for $j$ we necessarily have $(1 + \aeps)\alpha_j \ge t_i$.\footnote{Here, we use that all $\alpha$-values will be at least one and two values in the same bucket differs thus by at most a factor $1+\aeps$. We also remark that this is the same concept as in Invariant~\ref{inv:quasi} of the quasi-polynomial algorithm.}
\item \emph{stopped clients} $j$ have
  \begin{align}
    2\sqrtalpha_j \ge d(j,j') + 6\sqrtalpha_{j'}
    \label{eq:stopped}
  \end{align}
  for some other client $j'$.  In this case, we say that \emph{$j'$ stops $j$}.  Note that if $j'$ stops $j$, we necessarily have $\sqrtalpha_{j} \ge 3\sqrtalpha_{j'}$ and so $\alpha_{j} \ge 9\alpha_{j'}$.
\item \emph{undecided} clients $j$ are neither witnessed nor stopped.
\end{itemize}
Let us additionally call any client that is witnessed or stopped
\emph{decided}. Note that the sets of witnessed and stopped clients are not necessarily disjoint.  
However, we have the following lemma, which follows directly from the triangle inequality and our definitions:
\begin{lemma}
 \label{lem:stopped-normal}
 Suppose that $j$ is stopped.  Then $j$ must be stopped by some $j'$ that is \emph{not} stopped.
 \end{lemma}
 \begin{proof}
 We proceed by induction over clients $j$ in non-decreasing order of $\alpha_j$.  First, note that the client $j$ with smallest value $\alpha_j$ cannot be stopped.  For the general case, suppose that $j$ is stopped by some $j_1$.  Then, $\alpha_{j_1} < \alpha_j$.  If $j_1$ is stopped, then by the induction hypothesis it must be stopped by some $j_2$ that is not stopped.  Then, we have $2\sqrtalpha_j \ge d(j,j_1) + 6\sqrtalpha_{j_1}$, and $2\sqrtalpha_{j_1} \ge d(j_1,j_2) + 6\sqrtalpha_{j_2}$.  It follows that
 \begin{equation*}
 2\sqrtalpha_j \geq d(j,j_1) + (6-2)\sqrtalpha_{j_1} + d(j_1,j_2) + 6\sqrtalpha_{j_2}  \ge d(j,j_2) + 6\sqrtalpha_{j_2}.
 \end{equation*}
 Thus $j$ is stopped by $j_2$, as well.
 \end{proof}
 Intuitively, the stopping criterion will ensure that no $\alpha_j$ grows too large compared to the $\alpha$-values of nearby clients. At the same time it is designed so that all decided clients will have a good approximation guarantee.

Finally, we shall require that the following invariants hold throughout the execution of Algorithm~\ref{alg:1}.
\begin{invariant}[Feasibility]
\label{inv:feasibility}
For all $j \in \clients$, $\alpha_j \ge 1$ and for all $i \in \facilities$, $\sum_{j \in \clients}\beta_{ij} \le z_i$.
\end{invariant}
We remark that for dual feasibility $\alpha_j \geq 0$ is sufficient but the stronger assumption $\alpha_j \geq 1$ which is implied by Lemma~\ref{lem:dis} will be convenient. 

\begin{invariant}[No strict containment]
  \label{inv:containment}
  For any two clients $j,j' \in \clients$, $\sqrtalpha_j \leq d(j,j')  + \sqrtalpha_{j'}$
\end{invariant}
Note that the above invariant says that the ball centered at $j$ of radius $\sqrtalpha_j$ does \emph{not} strictly contain the ball centered at $j'$ of radius $\sqrtalpha_{j'}$.  For future reference, we  refer to the ball centered at  client $j$ of radius $\sqrtalpha_j$ as the $\alpha$-ball of that client. 
\begin{invariant}[$\solt{0}$ Completely Decided]
\label{inv:good-solution}
Every client is decided in $\solt{0}$.
\end{invariant}

Invariant~\ref{inv:good-solution} will be maintained as follows (as we show formally in Lemma~\ref{lem:invfeasgood}): The initial solution satisfies the invariant.  Then, given an initial solution $\solt{0}$ in which all clients are decided, \raiseprice will output a close, roundable sequence
$\cS^{(1)},\ldots,\cS^{(q)}$, where $\cS^{(q)} = (\alphat{q}, \zt{q}, \emptyset, \sclientt{q})$ is a roundable solution in which
all clients are decided. As the next call to \raiseprice will use $\solt{q}$ as
the initial solution, the invariant is maintained. 

\subsection{The main \raiseprice procedure.}
\label{sec:rais-sweep-proc}
\raiseprice is described in detail in Algorithm~\ref{alg:2}.
Initially, we suppose that we are given a $\lambda$-roundable and completely decided dual solution
$\solt{0}$ (i.e., satisfying Invariant~\ref{inv:good-solution}) where $z_i \in \{\lambda, \lambda+\zeps\}$ for all
$i \in \facilities$.  Additionally, let $\ist{0}$ be the independent set  of the conflict graph $\Gft{0}$ associated to the roundable solution $\solt{0}$, produced at the end of the previous call to \graphupdate{} as described in
Algorithm~\ref{alg:1}.  We shall assume that $|\ist{0}| \ge k$, as otherwise, Algorithm~\ref{alg:1} would have already terminated.  For a specified facility $i^+$, \raiseprice
sets $z_{i^+} \gets z_{i^+} + \zeps$.  This may result in some
clients using $i^+$ as a witness becoming undecided; specifically,
those clients that are not stopped and have no witness except $i^+$ in
$\solt{0}$.  We let $U^{(0)}$ to be the set of all these initially
undecided clients.  Throughout \raiseprice, we maintain a set $U$ of
currently undecided clients, and repair the solution over a series of
multiple stages, by calling an auxiliary procedure, \sweep.  
Each repair stage $s$ will be associated with a threshold $\theta_s$, and will make
multiple calls to the procedure \sweep, each producing a new solution $\alpha$.
The algorithm \raiseprice constructs a roundable solution $\cS = (\alpha,z,\sfac,\sclient)$ from each 
such $\alpha$, and returns the sequence $\cS^{(1)},\ldots,\cS^{(q)}$ of all such roundable solutions, 
in the order they were constructed.  \raiseprice terminates once it constructs some solution
in which all clients are decided.  In Section~\ref{sec:overview}, we shall show that this must happen after
at most $O(\log n)$ stages, and that each stage requires only a polynomial number of calls to \sweep. In addition, we show that the produced sequence is close and roundable.

\begin{algorithm}[h]
\DontPrintSemicolon
\SetNoFillComment
\caption{$\raiseprice(\alphat{0},\zt{0},\ist{0},i^+)$}
\label{alg:2}
\KwIn{
\begin{itemize}[nosep]
  \item $\solt{0}$ : a  $\lambda$-roundable solution satisfying Invariants~\ref{inv:feasibility}-\ref{inv:good-solution} with each~$z_i \in \{\lambda, \lambda + \zeps\}$.
\item $\ist{0}$ : an independent set of conflict graph $\Gft{0}$ of $\solt{0}$, produced by \graphupdate.
\item $i^+$ : a facility whose price $z_{i^+}$ is being increased from $\lambda$ to $\lambda + \zeps$
\end{itemize}
}
\KwOut{Sequence $\cS^{(1)}=\rolt{1},\ldots,\cS^{(q)}=\rolt{q}$ 
of close $\lambda$-roundable solutions,  where all clients are decided in $\cS^{(q)}$.}
$(\alpha, z) \gets \solt{0}$\;
  $z_{i^+} \gets z_{i^+} + \zeps$\;
  Let $U^{(0)}$ be the set of clients now undecided.\;
  Set $K = \Theta(\epsilon^{-1}\gamma^{-4})$ and select a shift parameter $0 \leq \sigma < K/2$.\;
  Set $\theta_1 = (\max_{j \in U^{(0)}}\alphat{0}_j + 2\zeps)(1 + \aeps)^\sigma$ and $\theta_s = (1 + \aeps)^K\theta_{s-1}$ for all $s > 1$.\;
  $U \gets U^{(0)}$\;
  $\ell \gets 1, s \gets 1$\;
  \While{$U \neq \emptyset$}
  {
    \tcc{Execute repair stage $s$}
    \While{there is some $j \in U$ with $\alpha_j < \theta_s$ }
    {
      $\alpha \gets \sweep(\theta_s,\alpha)$ \qquad (this procedure is described in Section~\ref{sec:sweep-procedure})\;
      $U \gets $ set of clients now undecided.\;
      Form $\sfac$ and $\sclient$ using $\alpha$, $z$, $\alphat{0}$, and $\ist{0}$.\; \label{line:dense-consr}
      $\cS^{(\ell)} \gets (\alpha, z, \sfac, \sclient)$.\; 
      $\ell \gets \ell + 1$\;
    }
    $s \gets s + 1$\;
  }
\end{algorithm}

Before describing \sweep in detail, let us first provide some intuition for the selection of the thresholds $\theta_s$ and describe the construction of each roundable solution $\cS^{(\ell)} = (\alphat{\ell}, \zt{\ell}, \sfact{\ell}, \sclientt{\ell})$ in \raiseprice.  Our procedure \sweep will adjust client values $\alpha_j$ similarly to the procedure \quasisweep described in Section~\ref{sec:quasi-polyn-time-algor}.  However, in each stage $s$, we  ensure that \sweep never increases any $\alpha_j$ above the threshold $\theta_s$ beyond its initial value $\alphat{0}_j$, i.e., we ensure that $\alpha_j \leq \alphat{0}_j$ for any $\alpha_j \ge \theta_s$.  We set
\begin{equation*}
  \theta_1 = (\max_{j \in U^{(0)}}\alphat{0}_j + 2\zeps)(1 + \aeps)^\sigma\quad \mbox{and}\quad \theta_s = (1+\aeps)^{K}\theta_{s-1}\,,
\end{equation*}
where $K = \Theta(\epsilon^{-1}\gamma^{-4})$ is an integer parameter and $\sigma$ is a integer ``shift'' parameter chosen uniformly at random\footnote{We shall show that it is in fact easy to select an appropriate $\sigma$ deterministically (see Remark~\ref{rem:derandomize-shift}).} from $[0, K/2)$.  Our selection of thresholds ensures that each stage updates only those $\alpha_j$ in a constant $K$ number of buckets.  Thus, the total change in any $\alpha$-value will be at most $n^{O(K)}$, which will allow us to obtain a polynomial running time.  This comes at the price of some clients remaining undecided after each stage, and some such clients $j$ may have service cost much higher than $\rho \cdot \alpha_j$.  We let $\cB$ denote the set of all such ``bad'' clients.  Using  that the $\alpha$-values are relatively well-behaved throughout \raiseprice, we show that only those clients $j$ with $\alphat{0}_j$ relatively near to the threshold $\theta_s$ can be added to $\cB$ in stage $s$.    Then, the random shift $\sigma$ in choosing our definition of thresholds will allow us to show that only an $O(K^{-1})$ fraction of clients can be bad throughout \raiseprice.  Moreover, we can bound the cost of each client $j \in \cB$ by $36 \alphat{0}_j$.  Intuitively, then, if at least a constant fraction of each $\alphat{0}_j$ is contributing to the service cost $c(j,\ist{0})$, then we can bound the effect of these bad clients by setting $K$ to be a sufficiently large constant, then using Theorem \ref{thm:rounding-approx-ratio} to conclude that:
  \[\sum_{j \in \cB}36\alphat{0}_j \leq \epsilon \cdot \sum_{j \in \clients}c(j,\ist{0}) \leq O(\epsilon) \cdot \opt_k\,.\]
Unfortunately, it may happen that many clients $j \in \cB$ have $\alphat{0}_j - c(j,\ist{0}) \approx \alphat{0}_j$.  That is, some clients may be using almost all of their $\alphat{0}$-values to pay for the opening costs of facilities.  In this case, we could have $\sum_{j \in \cB}\alphat{0}_j$ arbitrarily larger than $\sum_{j \in \clients}c(j, \ist{0})$.  In order to cope with this situation, we introduce a notion of \emph{dense} clients and facilities in Section~\ref{sec:handl-dense-clients}.  These troublesome clients and facilities are handled by carefully constructing the remaining components $\sfac$ and $\sclient$ of the roundable solution in line \ref{line:dense-consr}.  We defer the formal details to Section~\ref{sec:handl-dense-clients}, but the intuition is if enough bad clients are paying mostly for the opening cost of a facility, then we can afford to open this facility \emph{even if it is not tight}.  This is precisely the role of special facilities in Definition~\ref{def:roundable}.

\subsection{The \sweep Procedure}
\label{sec:sweep-procedure}
It remains to describe our last procedure, \sweep in more detail.  \sweep
operates in some stage $s$, with corresponding threshold value $\theta_s$,
takes as input the previous $\alpha$ produced by the algorithm, and produces
a new $\alpha$.  Note that in every call to \sweep, we let $\solt{0}$ denote
the roundable solution passed to \raiseprice, and $U$ is the set of undecided
clients immediately before \sweep was called.  Just like \quasisweep, the procedure \sweep, maintains a current set of active clients $A$ and a current threshold $\theta$, where initially, $A = \emptyset$, and $\theta = 0$.  We slowly increase $\theta$ and whenever $\theta = \alpha_j$ for some client $j$, we add $j$ to $A$.  While $j \in A$, we increase $\alpha_j$ at the same rate as $\theta$.  However, in contrast to \quasisweep, \sweep removes a client $j$ from $A$, whenever one of the following five events occurs:
\begin{enumerate}[label={Rule \arabic*.},ref={\arabic*},leftmargin=*,labelindent=\parindent]
\item $j$ has some witness $i$.
\item $j$ is stopped by some client $j'$.
\item $j \in U$ and $\alpha_j$ is $\zeps$ larger than its value at the start of \sweep.
\item $\alpha_j \ge \theta_s$ and $\alpha_j \ge \alphat{0}_j$.
\item  There is a client $j'$ that has already been removed from $A$ such that $\sqrtalpha_j \geq d(j,j') + \sqrtalpha_{j'}$.
\end{enumerate}
We remark that Rule $5$ says that $j$ is removed from $A$ as soon as its $\alpha$-ball contains the $\alpha$-ball of another client $j'$ that is not currently in $A$. This rule is designed so that the algorithm maintains Invariant~\ref{inv:containment}.
Also note that if a client $j$ satisfies one of these conditions when it is added to
$A$, then we remove $j$ from $A$ immediately after it is added.  In this case,
$\alpha_j$ is not increased.

As in \quasisweep, increasing the values $\alpha_j$ for clients in $A$ may
cause $\sum_{j \in \clients}\beta_{ij}$ to exceed $z_i$ for some facility $i$.
 We again handle this by decreasing some other values
$\alpha_{j'}$.  However, here we are more careful in our choice of clients to
decrease.  Let us call a facility $i$ \emph{potentially tight} if one of the following conditions hold:
\begin{itemize}
  \item There is some $j \in N(i)$ with $\alpha_j > \alphat{0}_j$.
  \item For all $j\in \Nt{0}(i)$, $\alpha_j \geq \alphat{0}_j$.
\end{itemize}
We now decrease 
$\alpha_{j'}$ if and only if $B(\alpha_{j'}) > B(\theta)$ and additionally: for
some potentially tight facility $i$ with $j' \in N(i)$ and $|N(i) \cap A| \ge 1$, we have
$\alpha_{j'} = t_i$.  We decrease each such $\alpha_{j'}$ at a rate of $|A|$
times the rate that $\theta$ is increasing.  To see that this maintains
feasibility we observe that at any time there are $|A \cap
N(i)|$ clients whose contribution to facility $i$ is increasing, and these
contributions are increasing at the same rate as $\theta$.  Suppose that $i$ is
tight at some moment with some $j \in N(i) \cap A$.  Then, since $z_i \ge
\zt{0}_i$, there must be either at least one client $j' \in N(i)$ with $\alpha_{j'} >
\alphat{0}_{j'}$ or we have $\alpha_{j'} = \alphat{0}_{j'}$ for all $j' \in \Nt{0}(i)$.  In either case, $i$ must be potentially tight.  Consider some client
$j_0 \in N(i)$ with $\alpha_{j_0} = t_i$ and note that $B(\alpha_{j_0}) = B(t_i)
> B(\alpha_{j}) = B(\theta)$, since otherwise we would remove $j$ from $A$ by
Rule 1.  The value of $\alpha_{j_0}$ is currently decreasing at a rate of $|A|
\ge |N(i) \cap A|$ times the rate that $\theta$ is increasing.  Thus, the total
contribution to any tight facility $i$ is never increased.  

As in \quasisweep, we stop increasing $\theta$ once every client $j$ has been added and removed from $A$, and then output the resulting $\alpha$.  Note that \sweep never changes any $\alpha_j < \theta$.  In particular, once some $j$ has been removed from $A$ it is not subsequently changed.  Additionally, observe that once $B(\theta) \ge B(\alpha_j)$, \sweep will not decrease $\alpha_j$.

\section{Analysis of the polynomial-time algorithm}
\label{sec:overview}

In contrast to the quasi-polynomial time procedure, here our analysis is quite involved.  Let us first provide a high-level overview of our overall approach.  Note that any solution that does not contain any undecided clients is roundable with $\sfac = \emptyset$, and $\bclient = \emptyset$.  Indeed, for any witnessed client $j$ there is a tight facility $i$ with $(1+\aeps)\sqrtalpha_j \ge \sqrt{t_i}$  and $\sqrtalpha_j \ge d(j,i)$ and so 
\begin{equation*}
(1 + (1+\aeps)\sqrt{\delta})\sqrtalpha_j \ge d(j,i) + \sqrt{\delta t_i}\,.
\end{equation*}
Similarly, any stopped client $j$ in such a solution must be stopped by some witnessed $j'$ (using Lemma~\ref{lem:stopped-normal} and the assumption that all clients are decided).  Let $i$ be the witness of $j'$.  Then,
\begin{equation*}
d(j,i) + \sqrt{\delta t_i} \le d(j,j') + d(j',i) + \sqrt{\delta t_i} \le
2\sqrtalpha_j - 6\sqrtalpha_{j'} + (1 + (1+\aeps)\sqrt{\delta})\sqrtalpha_{j'} < (1 + \sqrt{\delta})\sqrtalpha_j,
\end{equation*}
since $1 \leq \sqrt{\delta} \leq 2$.  As $\tau_i \leq t_i$ for any facility $i\in \facilities$, the required inequalities from Definition \ref{def:roundable} hold for any decided client $j$. In the following our main goal will then be to bound the cost in the general case in which some clients in a solution are undecided.

Our first task is to characterize which clients may \emph{currently} be undecided. To this end, we first prove some basic properties about the way \sweep alters $\alpha$-values together with Invariants~\ref{inv:feasibility},~\ref{inv:containment}, and~\ref{inv:good-solution} (Section~\ref{sec:basic-prop-sweep}).  Then, we show that only clients above threshold $\theta_s$ in each stage $s$ can become undecided (Section~\ref{sec:char-undec-clients-1}). In Section~\ref{sec:bound-cost-clients}, we bound the cost of all decided and undecided clients, showing that we can indeed obtain a $(\rho+O(\epsilon))$-approximation for all decided clients and a slightly worse guarantee for undecided ones.   Next, we would like to argue that most clients have good connection cost. Specifically, we would like to choose a set of thresholds that ensure that only a constant fraction of clients become undecided throughout the \emph{entirety} of \raiseprice.  In order to accomplish this, we show that our $\alpha$-values remain relatively stable throughout \raiseprice (Section~\ref{sec:prel-bound-chang}).  This allows us to prove that \raiseprice outputs close solutions and to  characterize those clients that may become undecided in \raiseprice by their values $\alphat{0}_j$ at the \emph{beginning} of \raiseprice.   This, together with our selection of thresholds, ensures that only an arbitrarily small, constant
fraction of clients do not have the desired guarantee.  However, we must also
show that these clients do not contribute more than a constant to $\opt_k$.
As discussed above, this will follow immediately from our analysis for those clients whose service
cost is at least a constant fraction of $\alphat{0}_j$.  For other (i.e.\ \emph{dense}) clients, we must use a different argument, involving the sets of special facilities and clients $\sfac$ and $\sclient(i)$ (Section~\ref{sec:handl-dense-clients}).  Finally, we put all of these pieces together and show that \raiseprice produces a close sequence of polynomially many roundable solutions and runs in polynomial
time (Section~\ref{sec:select-thresh-thet}).

\subsection{Basic properties of \sweep{} and Invariants~\ref{inv:feasibility},~\ref{inv:containment}, and~\ref{inv:good-solution}}
\label{sec:basic-prop-sweep}
We start by showing that Invariants~\ref{inv:feasibility},~\ref{inv:containment}, and \ref{inv:good-solution} hold.
\begin{lemma}
  \label{lem:invfeasgood}
  Invariants~\ref{inv:feasibility},~\ref{inv:containment}, and \ref{inv:good-solution} hold throughout Algorithm~\ref{alg:1}.
\end{lemma}
\begin{proof}
  We begin by proving Invariant~\ref{inv:feasibility}, \ie that the algorithm maintains a feasible dual solution $\alpha$ with the additional property that $\alpha_j \geq 1$ for all $j\in \clients$.  Recall our construction of the initial solution $\alphat{0}$ for Algorithm \ref{alg:1}:  we set $\alpha_j = 0$ for all $j \in \clients$ and then increase all $\alpha_j$ at a uniform rate.  We stop increasing a value $\alpha_j$ whenever $j$ gains a tight edge to some facility $i \in \facilities$ or $2\sqrtalpha_j \ge d(j,j') + 6\sqrtalpha_{j'}$ for some $j' \in \clients$.  Note that  no $\alpha_j$ is increased after $\alpha_j = d(j,i)^2$ for some facility $i$.  Thus, we have $\betat{0}_{ij} = 0$ for all $j \in\clients$ and $i \in \facilities$, and so $\alphat{0}$ is feasible.  Now, we show that $\min_{j \in \clients}\alphat{0}_j \ge 1$.  Consider the client $j_0$ that first stops increasing in our greedy initialization process.  At the time $\alpha_{j_0}$ stops increasing, we have $\alpha_{j} = \alpha_{j_0}$ for all $j \in \clients$ and so $2\sqrtalpha_j \geq d(j,j') + 6\sqrtalpha_{j'}$ cannot hold for any pair $j,j'$ of clients.  Thus, $j_0$ must have stopped increasing because $\alpha_{j_0} = d(j_0,i)^2$ for some facility $i$.  By our preprocessing (Lemma \ref{lem:dis}) we have $d(j_0,i)^2 \ge 1$, and so $\alphat{0}_{j_0} \ge 1$.  Moreover, $\alphat{0}_{j_0} = \min_{j \in \clients}\alphat{0}_j$, and so indeed $\alphat{0}_j \ge 1$ for all $j \in \clients$.  Now, we show that Algorithm~\ref{alg:1} preserves Invariant~\ref{inv:feasibility}.   Note that $\alpha$ is altered only by subroutine \sweep, and by construction, \sweep ensures that always $\sum_{j}\beta_{ij} \le z_i$.  Moreover, \sweep decreases any $\alpha_{j}$ only while there is some $j' \in N(i) \cap A$ for some facility $i$.  By our preprocessing (Lemma~\ref{lem:dis}) $\alpha_{j'} \geq d(j',i)^2 \ge 1$ for any such $j'$.  Thus, no $\alpha_{j}$ is ever decreased below 1.

Next, we prove Invariant~\ref{inv:containment}, \ie that no client's $\alpha$-ball is strictly contained in the $\alpha$-ball of another client.  First, let us show that the initially constructed solution $(\alphat{0},\zt{0})$ satisfies Invariant~\ref{inv:containment}.  Note that $\alphat{0}_j$ is equal to the value of $\alpha_j$ at the time that our initialization procedure stopped increasing $\alpha_j$.
Consider any pair of clients $j$ and $j'$.  If $\alphat{0}_{j} \le \alphat{0}_{j'}$ then clearly $\sqrtalphat{0}_j \le d(j,j')
+ \sqrtalphat{0}_{j'}$.  Thus, suppose that $\alphat{0}_{j} > \alphat{0}_{j'}$, so $\alpha_{j'}$ stopped increasing before $\alpha_j$ in our initialization procedure.  If $\alpha_{j'}$ stopped increasing because $j'$ gained a tight edge to a facility $i$, then once $\sqrtalpha_{j} = d(j,j') + \sqrtalpha_{j'}$, $j$ will have a tight edge to $i$ and stop increasing.  If $\alpha_{j'}$ stopped increasing because $2\sqrtalpha_{j'} = d(j',j'') + 6\sqrtalpha_{j''}$ for some client $j''$, then when $\sqrtalpha_{j} = d(j,j') + \sqrtalpha_{j'}$ we will have  
\[
2\sqrtalpha_{j} = 2d(j,j') + 2\sqrtalpha_{j'} = 2d(j,j') + d(j',j'') + 6\sqrtalpha_{j''} \ge d(j,j'') + 6\sqrtalpha_{j''}
\]
and so $\alpha_j$ must stop increasing.  In any case, we must have $\sqrtalphat{0}_j \le d(j,j') + \sqrtalphat{0}_{j'}$.  Having shown that the invariant is true for the first $\alphat{0}$ constructed in Algorithm~\ref{alg:1}, let us now prove that it is maintained.  First, we show that the inequality $\sqrtalpha_{j} \leq d(j,j') + \sqrtalpha_{j'}$ will not be violated by increasing $\sqrtalpha_{j}$.  Suppose that $j \in A$ and so $\alpha_j$ is increasing.  As long as $j' \in A$, as well, we have $\alpha_j = \alpha_{j'} = \theta$, and so $\sqrtalpha_j \leq d(j,j') + \sqrtalpha_{j'}$.  On the other hand, if $j' \not\in A$, then as soon as $\sqrtalpha_{j} = d(j,j') + \sqrtalpha_{j'}$, $j$ will be removed from $A$ by Rule~5 and $\sqrtalpha_j$ will no longer increase.  Now we show that also $\sqrtalpha_{j} \leq d(j,j') + \sqrtalpha_{j'}$ will not be violated by decreasing $\alpha_{j'}$.  Suppose that $\alpha_{j'}$ is decreasing.  Then, there must be some potentially tight facility $i$ with $j' \in N(i)$ and $t_i = \alpha_{j'}$.  Let $i$ be \emph{any} such facility.  If at some point we have $\sqrtalpha_{j} = d(j,j') + \sqrtalpha_{j'}$, then we must also have $j \in N(i)$ at this moment and $\alpha_{j}\geq \alpha_{j'} = t_i$.  Thus, $\alpha_j$ is also decreasing and in fact $\alpha_j = \alpha_{j'}$ (since also $\alpha_j \leq t_i$).  Then, $\sqrtalpha_{j}$ and $\sqrtalpha_{j'}$ are decreasing at same rate and so $\sqrtalpha_j = d(j,j') + \sqrtalpha_{j'}$ as long as $\sqrtalpha_{j'}$ continues to decrease.

Finally, we prove Invariant~\ref{inv:good-solution}, \ie that the input solution $\solt{0}$ to \raiseprice is always completely decided. 
Every client $j$ is either stopped by some client $j'$ or has a tight edge to some facility $i$ in our initially constructed solution $\solt{0}$.  Moreover, the initialization process ensures that $N(i) = \emptyset$ for all $i$ (since $\beta_{ij} = 0$ for all $i \in \facilities$ and $j \in \clients$).  Thus, in the latter case $t_i = 0$ and so $i$ is in fact a witness for $j$, and so every client $j$ is indeed either stopped or witnessed in this initial solution $\solt{0}$. To show that Invariant~\ref{inv:good-solution} holds throughout the rest of the Algorithm~\ref{alg:1}, we note that $\solt{0}$ is always updated (in line~\ref{line:update} of Algorithm~\ref{alg:1} where  $\cS^{(0)} \gets \cS^{(q)}$) with the $\alpha$-values corresponding to the last solution produced in a call to \raiseprice.  Due to the condition in the main loop of \raiseprice, every client is decided in this solution.
\end{proof}

The next lemma makes some basic observations about the way in which \sweep alters the $\alpha$-values.
\begin{lemma} The procedure \sweep satisfies the following properties:
\label{lem:sweep-basic}
\begin{enumerate}[label={\normalfont Property \arabic*.},ref={\arabic*},leftmargin=*,labelindent=\parindent]
\item \label{item:sb-decided} Any client $j$ that becomes decided after being added to $A$ remains decided until the end of the same call to \sweep.
\item \label{item:sb-contained} If the $\alpha$-ball of a client $j$ contains the $\alpha$-ball of a decided client, then $j$ is decided.
  \item\label{item:sb-mu} Consider the solution $\alpha$ at the beginning of \sweep, and let $\mu = \min_{j' \in U}\alpha_{j'}$.  Then, no  $\alpha_j < \mu$ is increased by \sweep, and no $\alpha_j$ with $B(\alpha_j) \le B(\mu)$ is decreased by \sweep.
\end{enumerate}
\end{lemma}
\begin{proof}
For Property~\ref{item:sb-decided}, suppose first that $j$ had a witness $i$ at some point after being added to $A$.  Consider any $j' \in N(i)$ at this moment.  At this moment, we must have $B(\alpha_{j'}) \le B(\alpha_{j}) \le B(\theta)$ and so $\alpha_{j'}$ cannot be decreased for the remainder of  \sweep.  In particular, $j$ retains a tight edge to $i$ until the end of \sweep and $i$ remains tight until the end of \sweep.  Additionally, any client $j'$  with $\alpha_{j'} > t_i$ will be removed from $A$ as soon as it gains a tight edge to $i$ (by Rule~1 since $i$ would then be a witness for $j'$).  Thus, $t_i$ cannot increase and so $i$ remains a witness for $j$ until the end of \sweep.  Next, suppose that $j$ was stopped by some $j'$ after being added to $A$.  Then, at this moment, $\alpha_{j'} < \alpha_j \le \theta$.  Hence, for the remainder of \sweep, neither $\alpha_{j'}$ or $\alpha_j$ are changed and so $j$ remains stopped by $j'$.

For Property~\ref{item:sb-contained} suppose that  the $\alpha$-ball of client $j$ contains the $\alpha$-ball of a decided client $j'$. Then if $j'$ has a witness $i$, then $i$ is also a witness for $j$, since $\sqrtalpha_j \geq d(j,j') + \sqrtalpha_{j'} \geq d(j,j') + d(j',i) \geq d(j,i)$ and $B(\alpha_j) \geq B(\alpha_{j'}) \geq B(t_i)$. Similarly if $j'$ is stopped by some client $j''$ then 
   \begin{align*}
     2\sqrtalpha_{j} \geq 2(d(j,j') + \sqrtalpha_{j'}) \geq 2d(j,j') + 6\sqrtalpha_{j''} + d(j',j'') \geq 6\sqrtalpha_{j''} + d(j, j'')\,,
   \end{align*}
 and so $j$ is also stopped by $j''$.

Finally, for Property~\ref{item:sb-mu}, consider the first client $j$ whose value $\alpha_j$ is increased by \sweep.  Note that $j$ must not be decided before calling \sweep: otherwise, since no other $\alpha$-value has yet been changed, this would hold at the moment $j$ was added to $A$, as well, and so $j$ would immediately be removed by Rule~1 or 2.  Thus, the first $\alpha_j$ that is increased by \sweep must correspond to some $j \in U$, and at the moment this occurs, $\theta = \alpha_j \ge \mu$.  Furthermore, by the definition of \sweep,  no $\alpha_j$ can then be decreased unless $B(\alpha_j) \geq B(\mu) + 1$.
\end{proof}

\subsection{Characterizing currently undecided clients}
\label{sec:char-undec-clients-1}
The next observations follow rather directly from the properties given in Lemma~\ref{lem:sweep-basic} and the invariants.  These facts will help us bound the number of clients that can become bad throughout the algorithm, and also the total number of calls to \sweep that must be executed in each call to \raiseprice.  Throughout this section, we consider a single call to \raiseprice and let $(\alphat{0},\zt{0},\ist{0},i^+)$ be its input. 

\begin{lemma}
\label{lem:no-changes-stage-1}
In stage 1, \sweep is executed only a single time.  After this call, for every $j \in U^{(0)}$, we have $\alpha_j \le \alphat{0}_j + \zeps < \theta_1$ and $j$ is decided.
\end{lemma}
\begin{proof}
Consider any client $j_0 \in U^{(0)}$.  Then $i^+$ was $j_0$'s witness in $\solt{0}$, and $j_0$ must not have been stopped or have had any other witness $i \neq i^+$.  
Observe that our choice of $\theta_1$ ensures that $\alphat{0}_{j_0} + \zeps < \theta_1$, so any $j_0 \in U^{(0)}$ will be removed from $A$ by Rule~3 once $\alpha_j = \alphat{0}_j + \zeps$.  Thus we must have $\alpha_j \le \alphat{0}_{j_0} + \zeps < \theta_1$ at the end of \sweep for every $j_0 \in U^{(0)}$.  

This also implies that no such $j_0$ is removed from $A$ by Rule~4.   We now show that when $j_0$ is removed from $A$ by any other rule, it must be decided.  By Property~\ref{item:sb-decided}, $j_0$ is then decided at the end of \sweep, as well.   First, we observe that if $j_0$ is removed from $A$ by Rules~1 or 2, then it is decided by definition.  Next, suppose that $j_0$ was removed by Rule~3, and let $\mu = \min_{j \in U^{(0)}}\alphat{0}_j$.  Since $i^+$ was a witness for every $j \in U^{(0)}$, we must have $B(\alphat{0}_{j})\leq B(\mu)$ for all $j \in \Nt{0}(i^+)$.   Thus, by Property~\ref{item:sb-mu} of \sweep, $\alpha_{j} \ge \alphat{0}_{j}$ for every $j \in \Nt{0}(i^+)$.  Then, since $\alpha_{j_0} = \alphat{0}_{j_0} + \zeps$, at the time $j_0$ was removed from $A$, $i^+$ must have been tight and also a witness for $j_0$.  By Property~\ref{item:sb-decided}, $j_0$ then remains decided until the end of \sweep.  Finally, we consider the case in which $j_0$ was removed by Rule~5.  We show the following:
\begin{claim*}
Suppose that some client $j$ is removed from $A$ by Rule~5 and that $j$ is undecided at this time.  Then, $\alpha_j \ge \theta_1$.
\end{claim*}
\begin{proof}
Consider the first time that any client $j$ that is undecided is removed from $A$ by Rule~5.  By Property~\ref{item:sb-contained}, the $\alpha$-ball of this client $j$ must contain the $\alpha$-ball of some undecided client $j'$ that was previously removed from $A$.  By Property~\ref{item:sb-decided} and our choice of time, $j'$ must have been removed from $A$ by Rule~3 or 4.  However, if $j'$ was removed by Rule~3, we must have $j' \in U^{(0)}$ and so, as we have previously shown, $j'$ must be decided.  Thus, $j'$ was removed by Rule~4, and so presently $\alpha_j = \theta \ge \alpha_{j'} \ge \theta_1$.  To complete the proof, we observe that any client that is removed from $A$ after $j$ must have an $\alpha$-value at least $\alpha_j$.
\end{proof}
\noindent
It follows by the above Claim that no $j_0 \in U^{(0)}$ can be undecided when it is removed by Rule~5, since, as we have shown, $\alpha_{j_0} < \theta_1$ for all such $j_0$.  By the above cases, every client $j_0 \in U^{(0)}$ is decided with $\alpha_j \le \alphat{0}_{j_0} + \zeps < \theta_1$ at the end of \sweep.  

It remains to show that \raiseprice continues to stage 2 after one call to \sweep.  Consider some client $j$ that is undecided at the end of \sweep.  By Property~\ref{item:sb-decided} $j$ must not have been removed from $A$ by Rule~1 or Rule~2.  Moreover, we must have $j \not\in U^{(0)}$ and so $j$ was not removed from $A$ by Rule~3.  Thus, $j$ was removed from $A$ by Rule~4 or 5.   In either case (by the definition of Rule~4 or the above Claim), we have $\alpha_j \ge \theta_1$ at this moment (and so also at the end of \sweep, since no $\alpha_j$ is changed after $j$ is removed from $A$).  Thus, after the first call to \sweep in stage 1, every undecided client $j$ has $\alpha_j \ge \theta_1$ and so \raiseprice immediately continues to stage~2.
\end{proof}

\begin{lemma}
\label{lem:sweep-preservation}
Consider any solution $(\alpha,z)$ produced by \raiseprice.  If $j$ is undecided in $(\alpha,z)$, then $\alpha_j \ge \alphat{0}_j$.
\end{lemma}
\begin{proof}
  Suppose toward contradiction that the statement is false. Consider the first call to \sweep that produces a solution violating it and for this call let $j$ be the first client (in the order of removal from $A$) such that $\alpha_j < \alphat{0}_j$ when $j$ is removed from $A$ but $j$ is undecided\footnote{By Property~\ref{item:sb-decided}, any client $j$ violating the statement must be undecided when removed from $A$ and have $\alpha_j < \alphat{0}_j$ at the time of its removal from $A$ since \sweep does not change $j$'s $\alpha$-value thereafter.}.   Then since, $j$ is undecided it was removed by Rule 3, 4, or 5. If $j$ was removed by Rule 4, then at this moment $\alpha_j \ge \alphat{0}_j$.
Suppose then that $j$ was removed by Rule 3.  Then, $j \in U$.  By Lemma~\ref{lem:no-changes-stage-1}, no client $j \in U$ before the first call to \sweep is undecided after this call, so $j$ must have been undecided at the end of some preceding call to \sweep.  By assumption, we must have had $\alpha_j \ge \alphat{0}_j$ at the moment $j$ was removed from $A$ in this preceding call (and so also immediately before the present call).  But, $\alpha_j$ has increased by $\zeps$, so still $\alpha_j \ge \alphat{0}_j$.  Finally, suppose $j$ was removed by Rule 5.  Then, the $\alpha$-ball of $j$ must contain the $\alpha$-ball of a client $j'$ that has already been removed $A$.  If $j'$ is decided, then by Property~\ref{item:sb-contained} $j$ is decided as well.  Suppose that $j'$ is undecided.  Then, since we picked the first client that violated the condition of the lemma, and $j'$ was already removed from $A$, we have that $\alpha_{j'} \geq \alphat{0}_{j'}$.   But then, if $\alpha_j < \alphat{0}_j$, we have
$\sqrtalphat{0}_j > \sqrtalpha_{j} \ge d(j,j') + \sqrtalpha_{j'} \ge d(j,j') + \sqrtalphat{0}_{j'}$ and $\alphat{0}$ violates Invariant~\ref{inv:containment}.  
In all cases we showed that we must have $\alpha_j \geq \alphat{0}_j$ at the moment that $j$ was removed from $A$, and so also at the end of \sweep, contradicting our assumption that $\alpha_j < \alphat{0}_j$ for some undecided client $j$.
\end{proof}

\begin{lemma}
\label{lem:no-changes-stage-later}
In every stage $s > 1$, no $\alpha_j$ is changed by \sweep until $\theta \ge \theta_{s-1}$.  In particular, every client $j$ with $\alpha_j < \theta_{s-1}$ is decided for every solution produced by \raiseprice in stage $s > 1$.
\end{lemma}
\begin{proof}
By Property~\ref{item:sb-mu} of \sweep, no $\alpha_j$ is changed until $\theta
= \min_{j \in U}\alpha_j$.  Thus to prove the first part of the claim, it suffices to show that in every stage $s > 1$, if $U \neq \emptyset$ then $\min_{j \in U}\alpha_j \ge \theta_{s-1}$.  Note that the second part of the claim then follows as well for every solution except the one produced by the final call to \sweep in \raiseprice, and this last solution has no undecided clients by Invariant~\ref{inv:good-solution}.

Let us now prove that $\min_{j \in U}\alpha_j \ge \theta_{s-1}$ in every stage $s > 1$.  We proceed by induction on the number of calls to \sweep made in stage $s$.  Before the first call to \sweep in stage $s$, we must have $\alpha_j \ge \theta_{s-1}$ for every $j \in U$, since otherwise stage $s-1$ would have continued.  So, consider some later call to \sweep in stage $s$, and consider any $j \in U$ before this call.  Then, we must have had $j$ undecided after the preceding call to \sweep in stage $s$.  Moreover, by Property~\ref{item:sb-decided}, $j$ must have been undecided when it was removed from $A$ in this preceding call.  Consider the \emph{first} client $j$ that was undecided upon removal from $A$ in this preceding call.  Then, $j$ cannot have been removed by Rules~1 or 2.  Moreover, since every client that has been removed from $A$ before $j$ is decided, Property~\ref{item:sb-contained} implies that $j$ must not have been removed by Rule~5.  If $j$ was removed by Rule~3, then we must have had $j \in U$ already in this preceding call to \sweep, and so by the induction hypothesis, $\alpha_j \ge \theta_{s-1}$.  Then, since $j$ was removed from $A$ by Rule~3, we had $\alpha_j \ge \theta_{s-1} + \zeps$.  Finally, if $j$ was removed by Rule~4, then we must have $\alpha_j \ge \theta_s > \theta_{s-1}$ by definition.  Thus, throughout every stage $s > 1$, if $U \neq\emptyset$, then $\min_{j \in U}\alpha_j \ge \theta_{s-1}$, as desired.
\end{proof}

\begin{corollary}
\label{cor:good-neighbor}
Suppose that in $\solt{0}$, $j$ is not stopped and has only $i^+$ as a witness, \ie $j\in U^{(0)}$.  Then, we have that $j$ is decided with $\alpha_j \le \alphat{0}_j + \zeps$ in every solution $(\alpha,z)$ produced by $\raiseprice(\alphat{0}, \zt{0}, \ist{0}, i^+)$.
\end{corollary}
\begin{proof}
We have $j \in U^{(0)}$ and so by Lemma~\ref{lem:no-changes-stage-1}, $j$ is decided with $\alpha_j \leq \alphat{0}_j + \zeps < \theta_1$ in the first solution produced by \raiseprice.  Moreover,  by Lemma~\ref{lem:no-changes-stage-later}, $\alpha_j$ remains unchanged and $j$ remains decided in all later stages.
\end{proof}

\subsection{Bounding the cost of clients}
\label{sec:bound-cost-clients}

In this section we derive inequalities that are used to bound the service cost
of each $(\alpha,z)$ produced during the algorithm.  Consider some solution
$\alpha$ produced by the algorithm, and define 
\begin{align}
  \cB = \{j\in \clients: \mbox{ $j$ is undecided and $2\sqrtalpha_j < d(j,j') +  6\sqrtalphat{0}_{j'}$ for all clients $j'$}\}\,.
  \label{eq:bdef}
\end{align}
The set $\cB$ is defined to contain those clients that are (potentially) bad, \ie have worse connection cost than our target guarantee. 
Specifically, we now show that all clients $j \in \clients \setminus
\cB$, satisfy the first inequality of Property 2 in Definition~\ref{def:roundable} (with $\tau_i$ replaced by $t_i$), while all clients  (in particular those in $\cB$)  satisfy a slightly weaker inequality.

\begin{lemma}
  \label{lem:nearby-witnessed-client}
  \label{lem:roundable-alpha-lower-bound}
 Consider any $(\alpha,z)$ produced by \raiseprice. For every client $j$  the following holds:
\begin{itemize}
\item If $j \in \clients \setminus \cB$, then there exists a tight facility $i$ such that $(1+\sqrt{\delta}+\aeps)\sqrtalpha_j \ge d(j,i) + \sqrt{\delta t_i}$.
\item There exists a tight facility $i$ such that $6\sqrtalphat{0}_j \ge d(j,i) + \sqrt{\delta t_i}$.
\end{itemize}

\end{lemma}
\begin{proof} 

  The proof is by induction on the well-ordered set (with respect to the natural order $\leq$) 
  \begin{align*}
    R = \{0\} \cup \{\alpha_j\}_{j\in \clients \setminus \cB} \cup \{(1+\aeps)\alphat{0}_j\}_{j\in \clients}\,. 
  \end{align*}
  Specifically,  we prove the following induction hypothesis: for
  $r\in R$,
  \begin{enumerate}[label={(\alph*)}]
    \item \label{near1} each client $j \in \clients\setminus \cB$  with $\alpha_j \leq r$ has a tight facility $i$ such that $(1+\sqrt{\delta}+\aeps)\sqrtalpha_j \ge d(j,i) + \sqrt{\delta t_i}$; 
    \item \label{near2} each client $j\in \clients$ with
   $(1+\aeps)\alphat{0}_j\leq r$ has a tight facility $i$ such that $6\sqrtalphat{0}_j \ge d(j,i) + \sqrt{\delta t_i}$.
  \end{enumerate}
  The statement then follows from the above with  $r = \arg\max_{r\in R} r$. 
  
  For the base case (when $r=0$), the claim is vacuous since there is no client $j$ such
  that $\alpha_j \le 0$  or $(1+\aeps) \alphat{0}_j \leq 0$ (because every
  $\alpha$-value is at least $1$ by Invariant~\ref{inv:feasibility}).  
    For the induction step, we assume that each client $j\in \clients\setminus \cB$ with
  $\alpha_j < r$ satisfies~\ref{near1} and each client $j\in
  \clients$ with $(1+\aeps)\alphat{0}_j < r$ satisfies~\ref{near2}.
  We need to prove that any client $j_0 \in \clients \setminus \cB$ with  $\alpha_{j_0}
  = r$  (respectively, $j_0 \in \clients$ with $(1+\aeps)\alphat{0}_{j_0} = r$) satisfies~\ref{near1} (respectively,~\ref{near2}).    We divide the proof into two cases.
                      \paragraph{Case 1: $j_0 \in \clients\setminus \cB$ with $\alpha_{j_0} = r$.} We prove that in this case $j_0$ satisfies~\ref{near1}.  Since $j_0 \not\in \cB$, either $j_0$ has a witness, $j_0$ is currently stopped, or there is another client $j$ such that $2\sqrtalpha_{j_0} \geq d(j_0, j) + 6 \sqrtalphat{0}_{j}$.

Suppose first that $j_0$ has a witness $i$.  Then, $i$ is a tight facility and, since $j_0$ has a tight edge to $i$, $d(j_0, i) \leq \sqrtalpha_{j_0}$. Moreover,
$B(\alpha_{j_0}) \geq B(t_i)$ which implies that $(1+\tfrac{\aeps}{2}) \sqrtalpha_{j_0}\geq \sqrt{(1+\aeps)} \sqrtalpha_{j_0} \geq \sqrt{t_i}$. Therefore, using that $\sqrt{\delta} \leq 2$,
  \begin{align*}
    d(j_0, i) + \sqrt{\delta t_i} \leq (1+ \sqrt{\delta} + \aeps) \sqrtalpha_{j_0}\,.
  \end{align*}

Now suppose that $j_0$ is stopped by another client $j$. Then $\alpha_j \leq
  \alpha_{j_0}/3^2 = r/9$.   On the one hand, if $j\in \clients \setminus \cB$, we have    $d(j, i) + \sqrt{\delta t_i} \leq (1+ \sqrt{\delta}
  + \aeps) \sqrtalpha_j \leq 6 \sqrtalpha_j$ for some tight facility $i$ by the induction
  hypothesis~\ref{near1}.  On the other hand, if $j \in \cB$ then $j$ is undecided so by Lemma~\ref{lem:sweep-preservation},
  $\alphat{0}_j \leq \alpha_j$. This in turn implies that $\alphat{0}_j \leq
  \alpha_j \leq r/9 < r/(1+\aeps)$.  We can thus apply the induction hypothesis~\ref{near2} to $j$, to conclude that there is a tight facility $i$ such
that $d(j,i) + \sqrt{\delta t_i} \leq 6 \sqrtalphat{0}_j \leq  6 \sqrtalpha_j$.  From above we have that,  whether $j$ is in $\cB$  or not,  there is a tight facility $i$ such that
  \begin{align*}
    d(j_0,i) + \sqrt{\delta t_i} &\leq d(j_0, j) + d(j,i) + \sqrt{\delta t_i} \\
    &\leq d(j_0, j) + 6\sqrtalpha_j \\
    &\leq 2\sqrtalpha_{j_0} \leq (1+\sqrt{\delta} + \aeps) \sqrtalpha_{j_0}\,,
  \end{align*}
  where the penultimate inequality uses the fact that $j_0$ is stopped by $j$ and thus $2\sqrtalpha_{j_0} \geq d(j, j_{0}) + 6 \sqrtalpha_{j}$.

Finally, suppose that $j_0$ is not stopped or witnessed. Then, $j_0$ is currently undecided and, as $j_0 \not \in \cB$,  there is a client $j$ such that $2\sqrtalpha_{j_0} \geq d(j_0, j) + 6 \sqrtalphat{0}_{j}$. This implies that 
$\alphat{0}_j  \leq \alpha_{j_0}/9 = r/9 < r/(1+\aeps)$.  We can thus apply the induction hypothesis~\ref{near2} to $j$ to conclude, that there is a tight facility $i$ such that $d(j,i) + \sqrt{\delta t_i} \leq 6\sqrtalphat{0}_j$.  Now, we have:
\begin{align*}
d(j_0, i) + \sqrt{\delta t_i} &\leq d(j_0,j) + d(j,i) + \sqrt{\delta t_i} \\
&\leq d(j_0,j) + 6\sqrtalphat{0}_j \\
&\leq 2\sqrtalpha_{j_0} \leq (1 + \sqrt{\delta} + \aeps)\sqrtalpha_{j_0}\,.
\end{align*}

\paragraph{Case 2: $j_0 \in \clients$ with $(1+\aeps)\alphat{0}_{j_0} = r$.}We prove that in this case $j_0$ satisfies~\ref{near2}.
Suppose first that $\alpha_{j_0} < \alphat{0}_{j_0}$. Then  $j_0$  is decided by
Lemma~\ref{lem:sweep-preservation}. Therefore $j_0 \in \clients \setminus \cB$ with $\alpha_{j_0} < r$ and so  by the induction hypothesis~\ref{near1} there is a tight facility $i$ satisfying $d(j_0, i) + \sqrt{\delta t_i} \leq (1+ \sqrt{\delta} + \aeps) \sqrtalpha_{j_0} < 6 \sqrtalphat{0}_{j_0}$, as required.  Similarly, if $j_0 \in U^{(0)}$ then by
Corollary~\ref{cor:good-neighbor}, $j_0$ is decided and 
\[\alpha_{j_0} \leq  \alphat{0}_{j_0} + \zeps < (1+\aeps)\alphat{0}_{j_0} = r\,,\]
where the second inequality follows from $\zeps < \aeps$ and $\alphat{0}_{j_0} \geq 1$ by Invariant~\ref{inv:feasibility}.  We can thus again apply the induction hypothesis~\ref{near1} to conclude that there is a tight facility $i$ satisfying $d(j_0, i) + \sqrt{\delta t_i} \leq (1+\sqrt{\delta} + \aeps) \sqrtalpha_{j_0} \leq 6 \sqrtalphat{0}_{j_0}$.  Thus, from now on, we assume that $\alpha_{j_0} \geq \alphat{0}_{j_0}$ and that
$j_0 \not \in U^{(0)}$. We divide the remaining part of the analysis into two sub-cases depending on whether $j_0$ was stopped in $\alphat{0}$.   

First, suppose that $j_{0}$ was stopped in $\alphat{0}$ by another client $j$.
  Then $\alphat{0}_j \leq \alphat{0}_{j_0}/9 < r/(1+\aeps)$ and so by the induction
  hypothesis~\ref{near2}, there is a tight facility $i$ satisfying $d(j, i) + \sqrt{\delta
  t_i} \leq 6 \sqrtalphat{0}_j$.  Hence,
  \begin{align*}
    d(j_0, i) + \sqrt{\delta t_i} & \leq d(j_0, j) + d(j,i) + \sqrt{\delta t_i} \\
     & \leq d(j_0, j) + 6\sqrtalphat{0}_j \\
     & \leq 2 \sqrtalphat{0}_{j_0}  < 6 \sqrtalphat{0}_{j_0}\,.
  \end{align*}
  
Finally, suppose that $j_0$ was not stopped in $\alphat{0}$. Then since every client is decided in $\alphat{0}$ (Invariant~\ref{inv:good-solution}) $j_0$
  had a witness $i$ in $\alphat{0}$. Moreover, as $j_0 \not \in U^{(0)}$, we may assume that $i  \neq i^+$ and so $z_i = \zt{0}_i$. By the definition of a witness,
  $\alphat{0}_{j_1} \leq (1+\aeps)\alphat{0}_{j_0}$ for all $j_1 \in
  \Nt{0}(i)$.  If $\alpha_{j_1} \geq \alphat{0}_{j_1}$ for all $j_1 \in \Nt{0}(i)$, then, since $z_i = \zt{0}_i$, our feasibility invariant (Invariant~\ref{inv:feasibility}) implies that in
fact $\alpha_{j_1} = \alphat{0}_{j_1}$ for all $j_1 \in \Nt{0}(i)$ and so $N(i) = \Nt{0}(i)$.  Therefore, in this case $i$ is still a witness for $j_0$ and
  $d(j_0, i) + \sqrt{\delta t_i} \leq (1+\sqrt{\delta} + \aeps) \sqrtalphat{0}_{j_0} \leq 6 \sqrtalphat{0}_{j_0}$.  It remains to consider the case when  $\alpha_{j_1}
  < \alphat{0}_{j_1}$  for some $j_1 \in \Nt{0}(i)$ (note that $j_1 \neq j_0$, since $\alpha_{j_0} \geq \alphat{0}_{j_0}$ by assumption).  Since $\alpha_{j_1} < \alphat{0}_{j_1}$, $j_1$ must be decided (by Lemma~\ref{lem:sweep-preservation}) and so $j_1 \in \clients\setminus \cB$.  Moreover, $\alpha_{j_1} < \alphat{0}_{j_1} \leq (1 + \aeps)\alphat{0}_{j_0} = r$, and so we can apply the induction hypothesis~\ref{near1} to conclude that there is a tight facility $i_1$ satisfying $d(j_1, i_1) + \sqrt{\delta t_{i_1}} \leq (1+\sqrt{\delta} + \aeps) \sqrtalpha_{j_1} < (1+\sqrt{\delta} + \aeps)\sqrtalphat{0}_{j_1}$. Then,
  \begin{align*}
    d(j_0,i_1) + \sqrt{\delta t_{i_1}} & \leq d(j_0,i) + d(i, j_1) + d(j_1,i_1) + \sqrt{\delta t_{i_1}} \\
    &< \sqrtalphat{0}_{j_0} + \sqrtalphat{0}_{j_1} + (1+\sqrt{\delta} + \aeps) \sqrtalphat{0}_{j_1}\\
    &\leq \sqrtalphat{0}_{j_0} + (1+\aeps)^{1/2}\sqrtalphat{0}_{j_0} + (1+\aeps)^{1/2}(1+\sqrt{\delta} + \aeps) \sqrtalphat{0}_{j_0}\\
    &\leq 6 \sqrtalphat{0}_{j_0}\,, & 
  \end{align*}
as required.
\end{proof}
Lemma~\ref{lem:nearby-witnessed-client} shows that the clients  in $ \clients \setminus \cB$ satisfy the first inequality of Property 2 in Definition~\ref{def:roundable}
while the potentially bad clients $j \in\cB$  satisfy a slightly weaker inequality.  It remains to prove that the potentially bad clients will have a small contribution towards the total cost of our solution.

\subsection{Showing that $\alpha$-values are stable}
\label{sec:prel-bound-chang}

The key to our remaining analysis is showing that the $\alpha$-values are relatively well-behaved throughout the algorithm.  The following lemma implies that \sweep decreases an $\alpha_{j'}$ only because it is increasing an $\alpha_{j}$ which is at most  a constant factor smaller. This will imply the required stability properties. 
\begin{lemma}
  At any time during Algorithm~\ref{alg:1}: if a client $j$ has a tight edge to some facility, then $\alpha_{j'} \leq 19^2 \alpha_{j}$ for every other client $j'$ with a tight edge to this facility.
 \label{lem:tightboundedratio}
 \end{lemma}
 \begin{proof}
   We prove the following stronger statement: at any time during Algorithm~\ref{alg:1}, we have
   \begin{align}
     \label{eq:stableinv}
      2\sqrtalpha_{j'} \leq d(j',j) + 18 \sqrtalpha_{j} 
   \end{align}
 for any pair $j,j'$ of clients. To see that this implies the lemma consider two clients $j$ and $j'$ that both have tight edges to $i^*$. Then
\[
2\sqrtalpha_{j'} \leq d(j',j) + 18\sqrtalpha_{j} \leq d(j',i^*) + d(i^*,j) + 18\sqrtalpha_j \leq \sqrtalpha_{j'} + \sqrtalpha_{j} + 18\sqrtalpha_{j},
\]
which implies that $\alpha_{j'} \leq 19^2 \alpha_{j}$.
   
Inequality~\eqref{eq:stableinv} is clearly satisfied by the initial solution $\alphat{0}$ constructed at the beginning of Algorithm~\ref{alg:1}, since we stop increasing any $\alpha_j$ as soon as $2\sqrtalpha_j \ge d(j',j) + 6\sqrtalpha_{j'}$ for any client $j'$, and neither $\alpha_j$ nor $\alpha_{j'}$ are later changed.  We now show that \eqref{eq:stableinv} continues to hold throughout the execution of
Algorithm~\ref{alg:1}. The only procedure that updates the dual solution is \sweep, so let us analyze its behavior.

First note that the inequality cannot become violated by increasing $\alpha_{j'}$, because as soon as $2\sqrtalpha_{j'} \ge d(j',j) + 6\sqrtalpha_{j}$, $j'$ will be removed from $A$ by Rule~$2$ of \sweep.  It remains to prove that the inequality does not become violated because $\alpha_j$ is decreasing.  To this end, consider a time when $\alpha_j$ is decreasing. Then, by the definition of \sweep, there must be some potentially tight facility $i$, such that $j \in N(i)$ with $\alpha_{j} = t_{i}$.  Since $j$ has the largest $\alpha$-value in $N(i)$ and $i$ is potentially tight, there is some client $j_1 \in N(i)$ (note that possibly $j_1 = j$) such that $\alphat{0}_{j_1} \leq \alpha_{j_1} \leq \alpha_j$.  We show the following:
\begin{claim*}
There exists some facility $i^\star$ such that $i^\star$ was tight in $\solt{0}$ and also:
 \begin{align*}
   d(j_1, i^\star) \leq 2\sqrtalphat{0}_{j_1} \leq 2 \sqrtalpha_{j} \qquad \mbox{and} \qquad \alphat{0}_{j''} \leq (1+\aeps)  \alphat{0}_{j_1} \leq (1+\aeps) \alpha_{j} \mbox{ for all }  j'' \in \Nt{0}(i^\star)\,.
 \end{align*}
\end{claim*}
\begin{proof}
By Invariant~\ref{inv:good-solution}, every client  must be decided in $\solt{0}$.  Consider client $j_1$.
If $j_1$ was witnessed in $\solt{0}$, then there was a tight facility  $i^\star$ such that $d(j_1, i^\star) \leq
\sqrtalphat{0}_{j_1} \leq \sqrtalpha_{j}$ and $\alphat{0}_{j''} \leq (1+\aeps) \alphat{0}_{j_1}$ for every $j'' \in \Nt{0}(i^\star)$.   If $j_1$ was stopped by a client $j_2$ in $\solt{0}$ (\ie $2\sqrtalphat{0}_{j_1} \geq d(j_1, j_2) + 6 \sqrtalphat{0}_{j_2}$),  then we may assume that $j_2$ is witnessed by Lemma~\ref{lem:stopped-normal}.  In this case, let $i^\star$  be the witness of $j_2$.  Then, 
\[
d(j_1, i^\star) \leq d(j_1,j_2) + d(j_2,i^\star) \le d(j_1,j_2) + \sqrtalphat{0}_{j_2} \le 2\sqrtalphat{0}_{j_1} \le 2\sqrtalpha_j\,,
\]
and also
\[
\alphat{0}_{j''} \leq (1+\aeps) \alphat{0}_{j_2} \leq \alphat{0}_{j_1} \leq \alpha_{j},
\]
for all $j'' \in \Nt{0}(i^\star)$.  In either case, the claim holds.
\end{proof}

Now, let $i^\star$ be the facility guaranteed to exist by the Claim.  Consider the dual solution $\alphat{p}$ at the last time that $j'$ was previously increased.  Then, we must have $\alphat{p}_{j'} \geq \alpha_{j'}$.  Additionally, since Algorithm~\ref{alg:1} never decreases any facility's price, and the current call to \raiseprice has increased any facility's price by at most $\zeps$, we have $\zt{p}_{i^\star} \le z_{i^\star} \le \zt{0}_{i^\star} + \zeps$.  Let $j^\star = \arg\min_{j'' \in \Nt{0}(i^\star)}\alphat{p}_{j''}$.  We claim that:
 \begin{equation}
\label{eq:stableinv-aux-1}
   \alphat{p}_{j^\star} = \min_{j'' \in \Nt{0}(i^\star)} \alphat{p}_{j''} \leq (1+\aeps) \alpha_{j} + \zeps\,.
 \end{equation}
Indeed, otherwise by the Claim, we would have $\alphat{p}_{j''} > (1 + \aeps)\alpha_j + \zeps \ge \alphat{0}_{j''} + \zeps$ for every $j'' \in \Nt{0}(i^\star)$.  Then, since $i^\star$ is tight in $\solt{0}$ we would have:
\begin{align*}
\sum_{j'' \in \clients} [\alphat{p}_{j''} - d(j'',i^\star)]^+ 
&\ge \sum_{j'' \in \Nt{0}(i^\star)} [\alphat{p}_{j''} - d(j'',i^\star)]^+ \\
&> \sum_{j'' \in \Nt{0}(i^\star)} [\alphat{0}_{j''} + \zeps - d(j'',i^\star)]^+ 
\ge \zt{0}_{i^\star} + \zeps \ge \zt{p}\,,
\end{align*}
contradicting Invariant~\ref{inv:feasibility}.
 
We shall now show that \eqref{eq:stableinv-aux-1} and the claim imply \eqref{eq:stableinv}.  Since $j'$ was increasing when $\alphat{p}$ was maintained, Rule~2 of \sweep implies that:
 \begin{align*}
   2\sqrtalpha_{j'} \leq   2\sqrtalphat{p}_{j'} & \leq d(j', j^\star) + 6 \sqrtalphat{p}_{j^\star} & \text{({\small $j'$ was last increased in $\alphat{p}$})} \\
   &\leq  d(j', j) + d(j, j_1) + d(j_1, i^\star) + d(i^\star, j^\star)  +  6 \sqrtalphat{p}_{j^\star}  & \text{({\small triangle inequality})} \\
   &\leq d(j', j) + d(j, j_1) + d(j_1,i^\star) + \sqrtalphat{0}_{j^\star} + 6 \sqrtalphat{p}_{j^\star}  & \text{({\small $j^\star \in \Nt{0}(i^\star)$})} \\
   &\leq d(j', j) + d(j, j_1) + d(j_1,i^\star) + \sqrtalphat{0}_{j^\star} + 12 \sqrtalpha_{j}  & \text{({\small inequality \eqref{eq:stableinv-aux-1}})} \\
   &\leq d(j', j) + d(j, j_1) + 2\sqrtalpha_{j} + (1+\aeps)^{1/2} \sqrtalpha_{j} + 12 \sqrtalpha_{j}  & \text{({\small $j^\star \in \Nt{0}(i^\star)$ and Claim above})} \\
   &\leq d(j', j) + d(j, j_1) + 16\sqrtalpha_{j} & \text{({\small arithmetic})} \\
   & \leq d(j', j) + d(j, i) + d(i, j_1) + 16 \sqrtalpha_{j} & \text{({\small triangle inequality})}\\
   & \leq d(j', j) + \sqrtalpha_{j} + \sqrtalpha_{j_1} + 16 \sqrtalpha_{j} & \text{({\small $j,j_1 \in N(i)$})} \\
   & \leq d(j',j) + 18 \sqrtalpha_{j}\,. & \text{({\small $\alpha_{j} \ge \alpha_{j_1}$ since $j$ decreasing})}
 \end{align*}
 and thus~\eqref{eq:stableinv} remains satisfied when $j$ is decreasing. 
\end{proof}

Using Lemma~\ref{lem:tightboundedratio}, we can now prove that \raiseprice produces a close sequence of solutions, and also bound the total number of clients in $\cB$ for any solution produced by \raiseprice.  For both of these tasks, we make use of the following auxiliary lemma, which is a consequence of Lemma~\ref{lem:tightboundedratio}.

\begin{lemma}
\label{lem:alpha-unchanged}
Throughout stage $s$, for all $j$ with $\alpha_j > \theta_s$, we have $\alpha_j \le \alphat{0}_j$, and for all $j$ with $\alphat{0}_j \ge 20^2\theta_s$ or $\alpha_j \ge 20^2\theta_s$, we have $\alpha_j = \alphat{0}_j$.
\end{lemma}
\begin{proof}
For the first claim, we show that any client $\alpha_j$ with $\alpha_j \ge \theta_s$ can continue to increase in stage $s$ only while  $\alpha_j < \alphat{0}_j$.  Indeed, if $\alpha_j \ge \theta_s$ then once $\alpha_j = \alphat{0}_j$, $j$ will immediately be removed from $A$ by Rule $4$. 

For the remaining claim, suppose first that $\alphat{0}_j \ge 20^2\theta_s$.  Then, as we have just shown, $\alpha_j \leq \alphat{0}_j$ throughout stage $s$.  Suppose towards contradiction that in fact $\alpha_j < \alphat{0}_j$ at some moment in stage $s$ or earlier, and let $\alphat{-}$ be the value of $\alpha$ at this time.  Then, at some moment in stage $s$ or earlier, we must have had $\alphat{-}_j < \alpha_j < \alphat{0}_j$, and $\alpha_j > 19^2\theta_s$ but $j$ decreasing.  Since $j$ is being decreased by \sweep at this moment, we must have $j\in N(i)$ for a potentially tight facility $i$.  Since $\alphat{0}_j > \alpha_j$ we must also have $j \in \Nt{0}(i)$.  However, Lemma~\ref{lem:tightboundedratio} implies that for every other $j' \in N(i)$ at this moment we have $\alpha_{j'} \geq 19^{-2}\alpha_j > \theta_s.$  Thus, by the first claim, $\alpha_{j'} \le \alphat{0}_{j'}$ for all $j' \in N(i)$.  This contradicts the fact that $i$ is potentially tight, since $j \in \Nt{0}(i)$ with $\alpha_j < \alphat{0}_j$. 

Finally, suppose that $\alpha_j \ge 20^2\theta_s$.  Then, by the first claim, we must have $\alpha_j \le \alphat{0}_j$ and so also $\alphat{0}_j \ge 20^2\theta_s$.  Then, as we have just shown, $\alpha_j = \alphat{0}_j$.
\end{proof}

\subsubsection{\raiseprice produces a close sequence of $\alpha$-values in polynomial time}
\label{sec:show-that-rais}
In the preceding section, we showed that all of the $\alpha$-values are relatively stable throughout the algorithm.  Using those observations, we can now prove that \raiseprice indeed produces a close sequence of $\alpha$-values.  To that end, let us select the remaining parameters $K$, $\sigma$, and $\zeps$ used in \raiseprice.

Recall that the thresholds used by \raiseprice are defined by:
\begin{equation*}
\theta_1 = (\max_{j \in U^{(0)}}\alphat{0}_j + 2\zeps)(1 + \aeps)^\sigma
\quad \text{and} \quad \theta_s = (1
  + \aeps)^{K}\theta_{s-1} \text{ for all $s > 1$.}
\end{equation*}
Therefore, the ratio of two consecutive thresholds  is
$\theta_s/\theta_{s-1} = (1+\aeps)^K$. We select $K$ to be the smallest integer satisfying
\begin{align*}
  (1+\aeps)^K \geq C_0^{2/\gamma^4},\quad \mbox{ where $\const0 = 81 \cdot 25 \cdot 20^{8}$.}
\end{align*}
Note that $K = \Theta(\aeps^{-1}\gamma^{-4})$.  Given $K$, we select an integer ``shift'' $\sigma$ uniformly at random from the interval $(0,K/2]$.  

Finally, we set the price increment $\zeps$ to:
\begin{equation}
  \zeps = n^{-6(K + \const1+2) - 3}  \qquad \text{where } \const1 = \lceil \log_{1+\aeps} (20^{4}) \rceil + 1 = O(\aeps^{-1})\,.
\end{equation}
Using these parameters, we can show that the sequence of solutions $(\alpha,z)$ produced by \raiseprice is indeed close.  Because each successive $\alpha$-value in this sequence is produced by calling \sweep on the previous value, it suffices to show the following.
\begin{proposition}
Each call to \sweep changes every $\alpha_j$ by at most $n^{-2}$.
\label{prop:close-values}
\end{proposition}
\begin{proof}
Consider a call to \sweep performed in stage $s$.  By the definition of \sweep, it suffices to bound how much $\alpha_j$ has changed at the moment it is removed from $A$, since it is not subsequently changed.  Let us begin by bounding how much any $\alpha_j$ may be increased.  As in our analysis of \quasisweep, it will then be possible to bound how much any $\alpha$-value is decreased.  Let $\alphat{1}$ and $U^{(1)}$ be the values of $\alpha$ and $U$ before this call to \sweep, and let $\mu = \min_{j \in U^{(1)}}\alphat{1}_j$.  We first show the following:
\begin{claim*}
Any $\alpha_j$ can increase by at most $\zeps n^{6(b+1)}$ while $B(\theta) \le B(\mu) + b$.
\end{claim*}
\begin{proof} The proof is by induction on $b= -1, 0, 1,\dots$.  

\paragraph{Base case $b=-1$:}   Initially we have $\theta = 0$ and, by Invariant~\ref{inv:feasibility}, $\mu \ge 1$.  Thus, at the start of any call to \sweep, we must have $B(\theta) = 0$ and $B(\mu) \ge 1$.  Now, note that while $B(\theta) \le B(\mu) - 1$ we must have $\theta < \mu$.  Then, by Property~\ref{item:sb-mu} of \sweep no $\alpha$-value has yet been altered, and so the claim holds trivially.

\paragraph{Inductive step (assume true for $-1,0,1,\dots, b-1$ and prove for $b$):}        
Now suppose that some $\alpha_j$ is increased by at least $\zeps$ while $B(\theta) \le B(\mu) + b$.  Otherwise, the claim is immediate since $\zeps < \zeps n^{6(b+1)}$.  Note that while this $\alpha_j$ is increasing we must also have $\alpha_j = \theta$ and so $B(\alpha_j) \le B(\mu) + b$.

First, suppose that $\alpha_j < \alphat{1}_j$.  Then, $\alpha_j$ was previously decreased.  Moreover, since $\alpha_j$ was increased by at least $\zeps$ while $B(\theta) \leq B(\mu) + b$, we must have previously decreased $\alpha_j$ while $B(\alpha_j) \leq B(\mu) + b$.  In particular, at the last moment $\alpha_j$ was decreased, we must have had $B(\alpha_j) \leq B(\mu) + b$, and since $\alpha_j$ was decreasing at this moment, we also had $B(\theta) < B(\alpha_j)$.  Therefore, $\alpha_j$ was decreased only while $B(\theta) < B(\mu) + b$.  Moreover, during this time, $j$'s $\alpha$-value was decreased at most $n$ times the maximum amount that any other client's $\alpha$-value was increased.  By the induction  hypothesis, any client's $\alpha$-value can increase at most $\zeps n^{6b}$ while $B(\theta) < B(\mu) + b$.  Thus, $\alpha_j$ has decreased at most $\zeps\cdot n^{6b+1}$, and after increasing $\alpha_j$ by at most this amount, we will again have $\alpha_j = \alphat{1}_j$.

Next, let us bound how much $j$'s $\alpha$-value may increase while $\alpha_j \ge \alphat{1}_j$ (and still $B(\alpha_j) \le B(\mu) + b$).  We now consider three cases, based on the initial status of $j$ in $\alphat{1}$.  

If $j$ is undecided initially, then $j \in U$ and $\alpha_j$ can increase by at most $\zeps \le \zeps n^{6b}$ (since $b \ge 0$) before it is removed by Rule 3.

Next, suppose that $j$ had some witness $i$ in $\alphat{1}$, and let $\Nt{1}(i)$ be the set of clients paying for $i$ in $\alphat{1}$.  For each $j' \in \Nt{1}(i)$ we must have $B(\alphat{1}_{j'}) \le B(\alphat{1}_j) \le B(\mu) + b$, and so $\alpha_{j'}$ is decreased by \sweep only while $B(\theta) \le B(\mu) +  b-1$.  By the same argument given above (when considering the case that $\alpha_j < \alphat{1}_j$), the $\alpha$-value of any such $j' \in \Nt{1}(i)$ can decrease at most $\zeps n^{6b+1}$ during \sweep.  Thus, the total contribution to $i$ can decrease at most $n \cdot \zeps n^{6b + 1} = \zeps n^{6b + 2}$ during \sweep.  After increasing $\alpha_j$ by at most this amount, $i$ will again be tight.  Moreover, at this moment any client $j'$ contributing to $i$ was either already added to $A$ (and potentially also removed), in which case $B(\alpha_{j'}) \le B(\theta) = B(\alpha_j)$, or it was not already added to $A$, in which case $B(\alpha_{j'}) \le B(\alphat{1}_{j'}) \le B(\alphat{1}_j) \le B(\alpha_j)$.  Thus, at this moment $i$ is a witness for $j$, and so $j$ will be removed from $A$ by Rule 1.

Finally, suppose that $j$ was initially stopped by some client $j'$.
Then, by Lemma \ref{lem:stopped-normal}, we may assume that $j'$  was
\emph{not} stopped.  Let $\Delta = \sqrtalpha_{j'} - \sqrtalphat{1}_{j'}$ be the amount that $\sqrtalpha_{j'}$ has been increased by \sweep.  Then, once $\sqrtalpha_j - \sqrtalphat{1}_j \ge 3\Delta$, we will have:
\begin{equation*}
2\sqrtalpha_j \ge 2\sqrtalphat{1}_j +6(\sqrtalpha_{j'} - \sqrtalphat{1}_{j'})
\ge d(j,j') + 6\sqrtalpha_{j'},
\end{equation*}
where in the last inequality we have used the fact that $j'$ stopped $j$ in $\alphat{1}$.  Thus, $\sqrtalpha_j$ can increase by at most $3\Delta$, before $j$ will again be stopped by $j'$ and removed from $A$ by Rule 2.  It remains to bound the corresponding increases in $\alpha_j$ and $\alpha_{j'}$.  We have:
\begin{equation*}
\alpha_j - \alphat{1}_j \le \Bigl(\sqrtalphat{1}_j + 3\Delta\Bigr)^2 - \alphat{1}_j
= 6\Delta \cdot \sqrtalphat{1}_j + 9\Delta^2\,.
\end{equation*}
Now, let us bound the right hand side.  Since $j'$ is not stopped, the previous cases show that $\alpha_{j'}-\alphat{1}_{j'} \le \zeps n^{6b + 2}$.  Then, we have:
\begin{equation*}
\textstyle
\Delta^2 = \left(\sqrt{\vphantom{\alphat{1}_j}\alpha_{j'}} - \sqrt{\alphat{1}_{j'}}\right)^2 \le
\left(\sqrt{\alphat{1}_{j'} + \zeps n^{6b+2}} - \sqrt{\alphat{1}_{j'}}\right)^2 \le
\zeps n^{6b+2},
\end{equation*}
where the last inequality follows from $\sqrt{a + b} \le \sqrt{a} + \sqrt{b}$ for all $a,b \in \bRn$.  On the other hand, since $g(x) = \sqrt{x}$ is a concave, increasing function of $x$ for all $x > 0$, we have:
\begin{equation*}
\textstyle
\Delta = \sqrt{\vphantom{\alphat{1}_j}\alpha_{j'}} - \sqrt{\alphat{1}_{j'}} \le
g\Bigl(\alphat{1}_{j'} + \zeps n^{6b+2}\Bigr) - g\Bigl(\alphat{1}_{j'}\Bigr)
\le \zeps n^{6b+2} \cdot g'\Bigl(\alphat{1}_{j'}\Bigr) =
\dfrac{\zeps n^{6b+2}}{2\sqrt{\alphat{1}_{j'}}}
 \le \dfrac{\zeps n^{6b+2}}{2},
\end{equation*}
where the last inequality follows from Invariant~\ref{inv:feasibility}, which implies $\alphat{1}_{j'} \ge 1$.  Combining the above bounds, in this case we have
\begin{align*}
\alpha_j - \alphat{1}_j \le 6\Delta\cdot \sqrtalphat{1}_j + 9\Delta^2 &\le
6\sqrtalphat{1}_j \cdot \frac{\zeps n^{6b+2}}{2} + 9\zeps n^{6b+2}
\\
&\le 3\sqrt{5n^{7}} \cdot \zeps n^{6b+2} + 9\zeps n^{6b+2}
\le 9\zeps n^{6b + 11/2}\, ,
\end{align*}
where the penultimate inequality follows from the feasibility invariant (Invariant~\ref{inv:feasibility}) and the preprocessing of Lemma~\ref{lem:dis} (that all squared-distances are at most $n^6$) which together imply that $\alpha_j \le \min_i (z_i + d(i,j)^2)  \leq  4n^7 + n^6 \leq 5n^7$ for all $j \in \clients$.

Combining all of the above cases, $\alpha_j$ can increase at most $\zeps n^{6b + 1}$, until $\alpha_j = \alphat{1}_j$ and then at most an additional $9\zeps n^{6b+11/2}$.  Thus, the total increase in $\alpha_j$ while $B(\theta) \le B(\mu) + b$ is at most $9\zeps n^{6b+11/2} + \zeps n^{6b + 1} \le \zeps n^{6b+6}$, as required.
\end{proof}

We now complete the proof of Proposition~\ref{prop:close-values}.  By Lemma \ref{lem:alpha-unchanged}, no $\alpha_j \ge 20^2\theta_s$ is changed by \sweep in any stage $s$, and so once $B(\theta) \ge B(20^2\theta_s)$ no $\alpha$-values are changed. By the Claim, we then have that in any call to \sweep in stage $s$, each client's $\alpha$-value is increased at most $\zeps n^{6(b + 1)}$ where
\[
b = B(20^2\theta_s) - B(\mu) \leq \lfloor\log_{1+\aeps}(20^2\theta_s/\mu)\rfloor + 1\,.
\] 
We now bound the above value $b$ for every stage $s$.

In stage $1$, we execute only a single call to \sweep (as shown in Lemma \ref{lem:no-changes-stage-1}) and in this call, $\mu = \min_{j \in U^{(0)}}\alphat{0}_j$. Since 
every $j \in U^{(0)}$ must have a tight edge to the facility $i^+$ in $\alphat{0}$, Lemma~\ref{lem:tightboundedratio} implies that $\nu \triangleq \max_{j \in U^{(0)}}\alphat{0}_j \le 20^2\mu$.   Then, recall that
\[
\theta_1 = (\nu + 2\zeps)(1+\aeps)^\sigma \le \nu (1+\aeps)^{K} \leq 20^2 \mu (1 + \aeps)^K,
\]
where we have used that $\nu \ge 1$ (by Invariant~\ref{inv:feasibility}), $\zeps < \aeps$ and $\sigma \le K/2 < K$.  Finally, recalling that $\const1 = \lceil\log_{1+\aeps}(20^4)\rceil$, we have:
\[
b \le \log_{1 + \aeps}(20^2\theta_1/\mu) + 1 \le \log_{1+\aeps}(20^2 \cdot 20^2 \cdot (1 + \aeps)^K) + 1 \le K + \const1 + 1\,.
\]

In stage $s > 1$, we have $\mu \ge \theta_{s-1}$ by Lemma~\ref{lem:no-changes-stage-later}.  Then, recall that $\theta_{s} = (1+\aeps)^K\theta_{s-1}$  Then, we have:
\[
b \le \log_{1 + \aeps}(20^2\theta_s/\mu) + 1 \le \log_{1+\aeps}(20^2 (1 + \aeps)^K) + 1 < K + \const1 + 1\,.
\]

In any case, the maximum increase in any client's $\alpha$-value is at most $\zeps n^{6(K + \const1 + 2)} = n^{-3}$ (recalling that by definition $\zeps = n^{-6(K + \const1 + 2) - 3}$).  As we have already observed above in the proof of the Claim, each $\alpha$-value can decrease at most $n$ times this amount.  Thus, no $\alpha$-value can decrease more than $n^{-2}$.
\end{proof}

\subsubsection{Bounding the number of clients in $\cB$}

We bound the number of clients in $\cB$ by showing that such clients need to
have an $\alphat{0}$-value close to a threshold $\theta_s$. We then select the thresholds so that only a tiny fraction of the clients can be in $\cB$. 

\begin{lemma}
\label{lem:bad-client-range}
Suppose that $j \in \cB$ for some $(\alpha,z)$ produced by \raiseprice.  Then, we must have $\tfrac{1}{81}\theta_s \le \alphat{0}_j < 25\cdot 20^4\theta_s$ for some $s$.
\end{lemma}
\begin{proof}
Consider a call to \raiseprice and let $\solt{0}$ be the input solution.  We denote by $\solt{\ell}$ the solution produced by the $\ell$th  call to \sweep in the execution of \raiseprice. We also use $\cB^{(\ell)}$ to refer to the set $\cB$ of potentially bad clients associated to the solution $\solt{\ell}$. That is, 
\[
\cB^{(\ell)} = \{j\in \clients: \mbox{ $j$ is  undecided in $\solt{\ell}$  and $2\sqrtalphat{\ell}_j < d(j,j') +  6\sqrtalphat{0}_{j'}$ for all clients $j' \in \clients$}\}.
\]
Note that  $\cB^{(0)} = \emptyset$ since every client is decided in  $\solt{0}$ by Invariant~\ref{inv:good-solution}. 
We prove the lemma by showing the following claim by
induction on $\ell$: 
\begin{claim*}
For each client $j\in \cB^{(\ell)}$ there is a client $j'$ such that $\sqrtalphat{\ell}_j \geq d(j, j') + \sqrtalphat{0}_{j'}$
and $\tfrac{1}{9}\theta_s \leq \alphat{0}_{j'} \leq 20^4 \theta_s$, for some $s$.
\end{claim*}  
Before proving the claim, let us show that it indeed implies the Lemma.  Suppose that $j \in \cB^{(\ell)}$ for some $\ell$. We start with the lower bound. Selecting $j' = j$ in the definition of $\cB^{(\ell)}$, we must have $2\sqrtalphat{\ell}_j < 6\sqrtalphat{0}_{j}$ and so $\alphat{\ell}_j < 9\alphat{0}_j$.  Now, consider the client $j'$ and stage $s$ guaranteed by the claim.  Then, we must have:
\[
\alphat{0}_j \geq \tfrac{1}{9}\alphat{\ell}_j \geq \tfrac{1}{9}\alphat{0}_{j'} \geq \tfrac{1}{81}\theta_s.
\]
For the upper bound, note that, because $j \in \cB^{(\ell)}$, $j$ is undecided in $\solt{\ell}$ and so by Lemma \ref{lem:sweep-preservation}, $\alphat{0}_j \leq \alphat{\ell}_j$.  In addition since $j \in \cB^{(\ell)}$ we must have (for the same client $j'$ and stage $s$ guaranteed by the claim):
\[
  2\sqrtalphat{\ell}_j < d(j,j') + 6\sqrtalphat{0}_{j'} \leq d(j,j') + \sqrtalphat{0}_{j'} + 5\cdot \sqrt{20^4\cdot \theta_s}
  \leq \sqrtalphat{\ell}_j +  5\cdot \sqrt{20^4\cdot \theta_s}.
\]
Thus, the upper bound $\alphat{0}_j \leq \alphat{\ell}_j < 25 \cdot 20^4\theta_s$ is satisfied as well. 
\end{proof}

\begin{proof}[Proof of the Claim]
The base case when $\ell=0$ is trivially true since $\cB^{(0)} = \emptyset$.  For the inductive step, we assume the induction hypothesis for $0, 1, \dots, \ell-1$ and prove it for $\ell$. Any client $j \in \cB^{(\ell)}$ is undecided and (by Property~1 of \sweep) must have been removed from $A$ by one of the Rules~$3$,~$4$, or~$5$. We divide the analysis based on these three cases.
\begin{description}
\item[Case 1: $j \in \cB^{(\ell)}$ was removed by Rule~$3$.] 
        In this case, we must have $j \in U$, by the definition of Rule~3.  Moreover, since all clients in $U^{(0)}$ are decided in every
        solution produced by \raiseprice (Corollary~\ref{cor:good-neighbor}), we must have that $\ell \geq 2$.  Thus, $j$ must have
        been undecided in the previously produced solution $\solt{\ell-1}$, and $\alphat{\ell}_j = \alphat{\ell-1}_{j} + \zeps$.  Then,
        $j \in \cB^{(\ell-1)}$, as well, since 
        \begin{equation*}
          2\sqrtalphat{\ell-1}_{j} < 2\sqrtalphat{\ell}_{j} < d(j, j') + 6\sqrtalphat{0}_{j'}\qquad \mbox{for all $j'\in \clients$.}
        \end{equation*}
        The statement then follows from the induction hypothesis and from $\alphat{\ell}_j \geq \alphat{\ell-1}_j$.

\item[Case 2: $j \in \cB^{(\ell)}$ was removed by Rule~$4$.]
        By the definition of Rule~4, we must have $\alphat{\ell}_j \geq \alphat{0}_j$
        and $\alphat{\ell}_j \geq \theta_s$, where $s$ is the stage in which $\alphat{\ell}$ was produced.
        In this case, we prove the claim for $j' = j$.  Clearly, we have $\sqrtalphat{\ell}_j \geq \sqrtalphat{0}_j = d(j,j) + 
        \sqrtalphat{0}_j$.  Next, we prove that $\tfrac{1}{9}\theta_s \leq \alphat{0}_{j} \leq 20^4\theta_s$.
        For the lower bound, we observe that since $j \in \cB^{(\ell)}$ we must have $2\sqrtalphat{\ell}_j < 6\sqrtalphat{0}_j$ and so 
        $\tfrac{1}{9} \theta_s \leq \tfrac{1}{9}\alphat{\ell}_j \leq \alphat{0}_j$. We now prove the upper bound.  First, note that $j$ was 
        not stopped by any client $j'$ in $\solt{0}$ since then we would have $2\sqrtalphat{\ell}_j \geq 2\sqrtalphat{0}_j \geq
        d(j,j')  + 6\sqrtalphat{0}_{j'}$ which would contradict that $j\in \cB^{(\ell)}$.  Then, since every client in
        $\solt{0}$ is decided (Invariant~\ref{inv:good-solution}), $j$ must have been witnessed by some facility $i$ in $\solt{0}$. 
        Since $j$ is currently undecided, by Corollary~\ref{cor:good-neighbor} we may further assume $i\neq i^+$ and thus $\zt{\ell}_i = 
        \zt{0}_i$.   Now suppose toward contradiction that $\alphat{0}_j > 20^4 \theta_s$. Lemma~\ref{lem:tightboundedratio}
        then implies $\alphat{0}_{j'} \ge 20^{-2}\alphat{0}_j > 20^2\theta_s$ for every $j' \in \Nt{0}(i)$.  Then, by
        Lemma~\ref{lem:alpha-unchanged}, $\alpha_{j'} = \alphat{0}_{j'}$ for
        all $j' \in \Nt{0}(i)$.   Thus, as $z_i = \zt{0}_i$, $i$ is still tight and a witness for $i$, which contradicts the assumption that
        $j\in \cB^{(\ell)}$, since $j$ is then decided.

 \item[Case 3: $j \in \cB^{(\ell)}$ was removed by Rule~$5$.] 
        Let $j_1, j_2, \dots, j_p$ be the clients in $\cB^{(\ell)}$ that were removed from $A$ by Rule~$5$ in the $\ell$th call to \sweep. 
        We index these clients according to the order in which they were removed from $A$. The previous
        cases already imply that the clients in $\cB^{(\ell)} \setminus \{j_1, \dots,
        j_p\}$ satisfy the induction hypothesis.  We now assume that it is true
        for the clients in $\cB^{(\ell)}_{<a} := \cB^{(\ell)} \setminus \{j_a, \dots, j_p\}$ and show that it
        is also true for client $j_a$, \ie for all clients in $\cB^{(\ell)} \setminus \{j_{a+1}, \dots, j_p\}$. 
        Consider client $j_a$.  Then, we must have $\sqrtalphat{\ell}_{j_a} \geq d(j_a,j') + \sqrtalphat{\ell}_{j'}$ for some 
        $j'$ that was previously removed from $A$.  Moreover, since $j_a$ is undecided, Property~\ref{item:sb-contained} implies that 
        $j'$ must also be undecided.  We now show that $j' \in \cB^{(\ell)}$.  Indeed, since $j'$ is undecided, if $j' \not \in \cB^{(\ell)}$, 
        there must be a $j''$ such that $2\sqrtalphat{\ell}_{j'} \geq d(j', j'') + 6 \sqrtalphat{0}_{j''}$.  But then 
        \begin{align*}
          2 \sqrtalphat{\ell}_{j_a} \geq 2d(j_a, j') + 2\sqrtalphat{\ell}_{j'}  \geq 2 d(j_a,j') + d(j', j'') + 6 \sqrtalphat{0}_{j''} \geq d(j_a, j'') + 6 \sqrtalphat{0}_{j''}\,,
        \end{align*}
        which contradicts the fact that $j_a\in \cB^{(\ell)}$.  Now, since $j' \in \cB^{(\ell)}$ was removed from $A$ before $j_a$, we in fact have $j' \in \cB^{\ell}_{<a}$, and so by assumption there exists some $j''$ and $s$ such that $\tfrac{1}{9}\theta_s \leq \alphat{0}_{j''} \leq 20^4\theta_s$ and $\sqrtalphat{\ell}_{j'} \geq d(j',j'') + \sqrtalphat{0}_{j''}$.  Thus
\[
\sqrtalphat{\ell}_{j_a} \geq d(j_a,j') + \sqrtalphat{\ell}_{j'} \geq d(j_a,j') + d(j',j'') + \sqrtalphat{0}_{j''} \geq d(j_a,j'') + \sqrtalphat{0}_{j''},
\]
and so the induction hypothesis holds for $j_a$ as well (using $j''$ and $s$).
    \end{description}
\end{proof}

The above lemma says that any client $j$ that becomes potentially bad in any
solution produced by \raiseprice (\ie in $j\in \cB$ for one produced solution)
must have $\alphat{0}_j$ close to a threshold. Our selection of the
shift-parameter $\sigma$ then ensures that this can only happen for a tiny
fraction of the clients. This allows us to   bound the connection
cost of clients in $\cB$ by an arbitrarily  small constant  fraction (depending on the
parameter $K$) of $\sum_{j\in \clients} \alpha_j$. However,   as stated
in  the second inequality of Property~$2$ in Definition~\ref{def:roundable}, we
need to bound the total connection cost of these clients as a tiny fraction of
$\opt_k$ instead of $\sum_{j\in \clients} \alpha_j$. As all we know is that
$\opt_k \geq \sum_{j\in \clients} \alpha_j - \lambda k$, this requires further
arguments that we present in the next section.   

We complete this section by formally showing that a client is unlikely to
become potentially bad over the randomness of the shift-parameter $\sigma$.   
For any given integer $\sigma \in [0,K/2)$, let
\begin{align*}
\cW(\sigma)  =\{j \in \clients : \text{ $81^{-1} \cdot 20^{-2}\cdot \theta_s \le \alphat{0}_j \le 25\cdot 20^{6} \cdot \theta_s$ for some $\theta_s$}\}\,.
\end{align*}
Note that by the above lemma, any client that is in $\cB$ in any solution
produced during the considered call to \raiseprice, is in
$\cW(\sigma)$.\footnote{To argue $\cB \subseteq \cW(\sigma)$, the bounds $81^{-1} \theta_s \le \alphat{0}_j \le 25\cdot 20^{4} \cdot \theta_s$ for some
$\theta_s$ would be sufficient in the definition of $\cW(\sigma)$. However, the
more relaxed bounds will be useful when analyzing dense clients in the next section.}
Note that each value $\alphat{0}_j$ is fixed at the beginning of \raiseprice, and there are only a 
(relatively) small number of choices for $\sigma$ such that any given $j$ is in $\cW(\sigma)$. 
Thus, if we choose $\sigma$ uniformly at random, the probability that any given $j \in \cW(\sigma)$ is
small.  The following corollary formalizes this intuition.
\begin{corollary}
\label{cor:small-prob}
If we select the shift-parameter $\sigma$ uniformly at random from $[0, K/2)$, 
  \begin{align*}
    \Pr[j\in \cW(\sigma)] \leq \gamma^4 \qquad \mbox{for any client $j$.}
  \end{align*}
\end{corollary}
\begin{proof}
  Suppose that we select an integer $\sigma$ uniformly at random from $[0,K/2)$. Then, note that by definition $\theta_s =(\max_{j \in U^{(0)}}\alpha_j + 2\zeps)(1 + \aeps)^{K\cdot(s-1) + \sigma}$ and so   $j \in \cW(\sigma)$ if and only if:
\begin{equation*}
K_1 + K (s-1) + \sigma \le \log_{1+\aeps}\alphat{0}_j \le K_2 + K (s-1) + \sigma\,,
\end{equation*}
for some $s$, where $K_1 = \log_{1+\aeps}(81^{-1} \cdot 20^{-2}) + \log_{1+\aeps}(\max_{j \in U^{(0)}}\alphat{0}_j + 2\zeps) $ and $K_2 = \log_{1+\aeps}(25\cdot20^{6})+ \log_{1+\aeps}(\max_{j \in U^{(0)}}\alphat{0}_j + 2\zeps)$. 
In other words, $\sigma$ needs to satisfy 
\begin{equation*}
  K_1 + K (s-1) - \log_{1+\aeps}\alphat{0}_j   \le - \sigma \le K_2 + K (s-1) - \log_{1+\aeps}\alphat{0}_j \mbox{ for some integer $s$.}
\end{equation*}
Notice that the difference between the upper bound and the lower bound is $K_2 - K_1  = \log_{1+\aeps}(81\cdot 25 \cdot 20^{8}) = \log_{1 + \aeps}C_0$ which by selection of $K$ is at most $\tfrac{\gamma^4}{2} K$.
Moreover, as $\sigma \in [0,K/2)$ there is at most one value of $s$  that can satisfy the above inequalities. 
  It follows that there are at most $\frac{\gamma^4}{2}K$ distinct values of $\sigma$ so that $j\in \cW(\sigma)$.   Thus, $j \in \cW(\sigma)$ with probability at most $\gamma^{4}$.  
\end{proof}

\subsection{Handling dense clients}
\label{sec:handl-dense-clients}
Corollary~\ref{cor:small-prob} implies that by carefully selecting our thresholds, we can ensure that only an arbitrarily small fraction $\gamma^4$ of clients $j$ can ever appear in $\cB$ throughout the execution of \raiseprice.  As briefly discussed previously, this is unfortunately insufficient for our purposes.  Specifically, in order to charge the extra service cost incurred by this small fraction of clients to $\opt_k \ge \sum_{j \in \clients}\alpha_j - \lambda k$, we need to handle carefully those clients $j$ for which most of $\alpha_j$ is contributing toward the opening cost $\lambda k$.

To cope with this difficulty, we introduce the notion \emph{dense} facilities and clients, as follows.  Recall that $\gamma \ll \aeps$ is a small constant.  We define the \emph{$\gamma$-close neighborhood} of a facility $i$ as 
\[
\cneigh(i) = \{ j \in \clients : d(j,i)^2 \le \gamma \alphat{0}_j \}.
\]  
Then, we say that a facility $i \in \ist{0}$ is \emph{dense} if 
\[
\sum_{j \in \cneigh(i)}\alphat{0}_j \ge (1 - \gamma)\zt{0}_i.
\]
We let $\dfac \subseteq \ist{0}$ be the set of all dense facilities, and then define the set of \emph{dense clients} as $\dclient = \bigcup_{i \in \dfac}\cneigh(i)$.  Note that the $\gamma$-close neighborhoods, dense facilities, and dense clients are all determined only by the input solution $(\alphat{0},\zt{0})$ and the integral solution $\ist{0}$ passed to \raiseprice.  

Intuitively, the dense clients are precisely those troublesome clients for which $\alphat{0}_j$ is much larger then the service cost of $j$ in $\ist{0}$.  In order to avoid paying $36 \alphat{0}_j$ for any such clients, we construct a set of \emph{special facilities} $\sfact{\ell}$ and \emph{special clients} $\sclientt{\ell}$ for each $\alphat{\ell}$ produced by our \raiseprice, as follows.  Let
\begin{align}
  \label{eq:sfacdef}
  \sfact{\ell} &= \{ i \in \dfac : |\cneigh(i) \cap \cB| \neq \emptyset \text{  and } \alphat{\ell}_{j} \ge \alphat{0}_{j} ~ \forall j \in \cneigh(i) \}\,, \quad \text{and}\notag \\
\sclientt{\ell}(i) &= \cneigh(i) \quad \text{for every $i \in \sfact{\ell}$}\,,
\end{align}
We will show that each solution $\cS^{(\ell)} = \rolt{\ell}$ is roundable, with the set of remaining bad clients given by $\bclient = \cB \setminus \dclient$. 

Recall that in Definition~\ref{def:roundable}, we have $\tau_i = \max_{j \in N(i) \cap \sclient(i)}\alpha_j$ for any facility $i \in \sfac$ and $\tau_i = t_i$ for all other facilities.  Note that as $\solt{0}$ by Invariant~\ref{inv:good-solution} is a completely decided solution, by our choice of $\sfac$ and $\sclient$, we have $\sfact{0} = \emptyset$ in the roundable solution $\rolt{0}$.
Therefore, the conflict graph  $\Gft{0}$ of $\rolt{0}$ does not contain any special facilities, and so $\tau_i = t_i$ for each facility $i \in \Gft{0}$.   Moreover, recall that $\ist{0}$ was the maximal independent set of $\Gft{0}$ computed in the previous call to \graphupdate.  In particular, $|\ist{0}| > k$ and $\ist{0}$ does not contain any special facilities.

The following simple lemma is now a direct consequence of our definitions.
\begin{lemma}
\label{lemma:singlepayment}
Suppose that $j \in \cneigh(i)$ for some $i \in \ist{0}$.  Then, $\betat{0}_{i'j} = 0$ for all other $i' \in \ist{0}$. Moreover, for every client $j\in \clients$, $\alphat{0}_j \geq \sum_{i\in \ist{0}} \betat{0}_{ij}$. 
\end{lemma}
\begin{proof}
  We start by proving that if $j \in \cneigh(i)$ for some $i \in \ist{0}$, then $\betat{0}_{i'j} = 0$ for all other $i' \in \ist{0}$.
  Consider some facility $i \in \ist{0}$, and suppose that $j \in \cneigh(i)$.  Further, suppose that for some other facility $i' \in \Gft{0}$ we have $\beta_{i'j} > 0$.  Note that this implies (since no facility is special) that $j$ is adjacent to both $i$ and $i'$ in the client-facility graph that generated $\Gft{0}$. We shall show that $i' \not\in \ist{0}$.  Indeed, we must have:
  \begin{equation*}
  d(i,i') \le d(i,j) + d(j,i') < \sqrt{\gamma} \cdot \sqrtalphat{0}_j + \sqrtalphat{0}_j < \sqrt{\delta} \cdot \sqrtalphat{0}_j \le \sqrt{\delta t_i} = \sqrt{\delta \tau_i}\,.
\end{equation*}
Thus, there is an edge between $i$ and $i'$ in the conflict graph $\Gft{0}$, and since $i \in \ist{0}$, we have $i' \not\in \ist{0}$.

We shall now prove that $\alphat{0}_j \geq \sum_{i\in \ist{0}} \betat{0}_{ij}$ for any client $j\in \clients$.
Again using that no facility is special, we have that $j$'s neighborhood in the
client-facility graph that generated $\Gft{0}$ is equal to  the
set of tight facilities that $j$ is paying for.  Therefore, we have that $\alphat{0}_j - \sum_{i \in \ist{0}} \betat{0}_{ij} \geq d(j,\ist{0})^2/\rho$ (which implies $\alphat{0}_j \geq \sum_{i\in \ist{0}} \betat{0}_{ij}$) by the exact same arguments as ``Case $s=1$'' and ``Case $s>1$''  in the proof of Theorem~\ref{thm:euckmeans}.
\end{proof}

The next lemma shows that we can indeed charge the total $\alphat{0}$-value of all non-dense clients to $\opt_k$.
\begin{lemma}
\label{lem:alpha-nondense}
$\sum_{j \in \clients \setminus \dclient}\alphat{0}_j \le \gamma^{-3}\cdot \opt_k$.
\end{lemma}
\begin{proof}
We partition the clients in $\clients \setminus \dclient$ into two sets:
  \begin{align*}
    \clients_{>\gamma} = \{j\in \clients\setminus \dclient : d(j, \ist{0})^2 > \gamma\cdot \alphat{0}_j\} \mbox{ and } \clients_{\leq \gamma} = \{j\in \clients\setminus \dclient : d(j, \ist{0})^2 \leq \gamma\cdot \alphat{0}_j\}. 
  \end{align*}
By definition,
  \begin{align}
    \label{eq:claimtriv}
    \sum_{j\in \clients_{>\gamma}}  d(j, \ist{0})^2 >  \gamma \cdot \sum_{j\in \clients_{>\gamma}}  \alphat{0}_j\,.
  \end{align}
  To bound the remaining clients consider the following fractional token argument: each client $j\in \clients_{>\gamma}$ distributes  $\betat{0}_{ij} = [\alphat{0}_j - d(j,i)^2]^+$  tokens to each facility $i\in \ist{0}$.  Lemma~\ref{lemma:singlepayment} says that $\sum_{i \in \ist{0}}\betat{0}_{ij} \leq \alphat{0}_j$ for every client $j$, and so the total number of tokens distributed is at most $\sum_{j \in \clients_{>\gamma}}\alphat{0}_j$.

Now note that every $j \in \clients_{\leq \gamma}$ is in $\cneigh(i)$ for some $i \in \ist{0}$.  By Lemma \ref{lemma:singlepayment} we thus have $\betat{0}_{i'j} = 0$ for every $i' \neq i$ in $\ist{0}$.  Moreover, $i$ must \emph{not} be dense, since otherwise $j$ would be in $\dclient$.  Hence, we have  
  \begin{align}
    \label{eq:some1}
    \sum_{j \in \cneigh(i)}\betat{0}_{ij} \leq 
\sum_{j\in \cneigh(i)} \alphat{0}_j \leq (1-\gamma) \zt{0}_{i}\,.
  \end{align}
Moreover, there must be at least $\gamma\zt{0}_i$ tokens assigned to $i$, because it is a tight facility with respect to $\alphat{0}$ and by Lemma~\ref{lemma:singlepayment}, every client $j \not\in \clients_{> \gamma} \cup  \cneigh(i)$ must have $\betat{0}_{ij} = 0$.  That is,
\begin{align*}
\sum_{j\in \clients_{> \gamma}} \betat{0}_{ij}  = \sum_{j\in \clients \setminus \cneigh(i)} \betat{0}_{ij} \geq \gamma \zt{0}_i\,.
\end{align*}
  Thus,
\begin{align*}
  \sum_{j\in \clients_{\leq \gamma}} \alphat{0}_j = \sum_{i \in \ist{0} \setminus \dfac} \sum_{j\in \cneigh(i)} \alphat{0}_j
\leq \sum_{i \in \ist{0} \setminus \dfac}\zt{0}_i     = \frac{1}{\gamma}\sum_{i \in \ist{0} \setminus \dfac} \gamma \cdot \zt{0}_i 
    \leq \frac{1}{\gamma}  \sum_{j\in \clients_{>\gamma}} \alphat{0}_j\,,
  \end{align*}
where the first inequality follows from~\eqref{eq:some1}, and the last inequality from the fact that each facility $i\in \ist{0} \setminus \dfac$ received at least $\gamma \zt{0}_i$ tokens and the total amount of distributed tokens was at most $\sum_{j\in \clients_{>\gamma}} \alphat{0}_j$.

Hence, 
  \begin{align*}
\sum_{j \in \clients \setminus \dclient} \alphat{0}_j = \sum_{j\in \clients_{>\gamma}} \alphat{0}_j  + \sum_{j\in \clients_{\leq \gamma}} \alphat{0}_j & \leq   \left( 1+ \frac{1}{\gamma}\right)  \sum_{j\in \clients_{>\gamma}} \alphat{0}_j \\
    & < \frac{1}{\gamma} \cdot \left( 1+ \frac{1}{\gamma}\right)  \sum_{j\in \clients_{>\gamma}} d(j, \ist{0})^2 \\
    & \leq  \frac{1}{\gamma} \cdot \left( 1+ \frac{1}{\gamma}\right)\bigl(\rho + 1000\epsilon\bigr)\opt_k\,,
  \end{align*}
  where the penultimate inequality follows from~\eqref{eq:claimtriv} and the last inequality from Theorem~\ref{thm:rounding-approx-ratio}, as $|\ist{0}| > k$.
\end{proof}

Lemma~\ref{lem:alpha-nondense} shows that we can relate the total $\alpha$-value of all non-dense clients to $\opt_k$, as desired. Moreover, we have the following corollary.
\begin{corollary}
  \label{cor:W-mass-small}
  If we select the shift-parameter $\sigma$ uniformly at random from $[0, K/2)$,
    \begin{align*}
      \E\left[\sum_{j\in \cW(\sigma) \setminus \dclient} \alphat{0}_j\right] \leq \gamma \cdot \opt_k\,.
    \end{align*}
In particular, if we set $\cW = \cW(\sigma)$ for the value $\sigma$ that minimizes $\sum_{j \in \cW(\sigma) \setminus \dclient}\alphat{0}_j$ then, we have
\begin{equation*}
\sum_{j \in \cW \setminus \dclient} \alphat{0}_j \leq  \gamma \cdot \opt_k\,.
\end{equation*}
\end{corollary}
\begin{proof}
For the first claim, we note that
  \begin{align*}
    \E\left[\sum_{j\in \cW(\sigma) \setminus \dclient} \alphat{0}_j\right]  = \sum_{j\in \clients \setminus \dclient}  \Pr[j\in \cW(\sigma)] \cdot\alphat{0}_j 
     \leq \gamma^4  \sum_{j\in \clients \setminus \dclient}  \alphat{0}_j 
     \leq  \gamma \cdot \opt_k\,,
  \end{align*}
where the first inequality follows from by Corollary~\ref{cor:small-prob} and the last inequality follows from Lemma~\ref{lem:alpha-nondense}.  The second claim now follows since the minimum of left-hand side over all $\sigma \in [0,K/2)$ is at most its expected value over a randomly chosen $\sigma \in [0,K/2)$.
\end{proof}

\begin{remark}
  \label{rem:derandomize-shift}
  The only property of the selection of $\sigma$ that we use is that $\sum_{j\in \cW\setminus \dclient} \alphat{0}_j \leq \gamma \cdot \opt_k$. It easy to find the $\sigma$ that minimizes $\sum_{j\in \cW(\sigma) \setminus \dclient} \alphat{0}_j$ (since the number of options is constant) at the start of a call to \raiseprice and thus the selection of the shift-parameter can be derandomized.
\end{remark}

Now, we show how to obtain a better bound than that given by Lemma~\ref{lem:roundable-alpha-lower-bound} for dense clients $\dclient \cap \cB$.  Specifically, we show how to bound the cost of all clients in $\dclient \cap \cB$ using the facilities of $\sfac$.  This will allow us to eventually obtain a roundable solution.  

\begin{lemma}
\label{lem:roundable-dense-bound}
For any $j \in \dclient \cap \cB$, either:
\begin{itemize}
\item There exists a tight facility $i \in \facilities$ such that $(1+\sqrt{\delta} + 10\aeps)\sqrtalpha_j \ge d(j,i) + \sqrt{\delta t_i}$.
\item There exists a special facility $i \in \sfac$ such that $(1+\sqrt{\delta} + 10\aeps)\sqrtalpha_j \ge d(j,i) + \sqrt{\delta \stime_{i}}$.
\end{itemize}
\end{lemma}
\begin{proof}
  Consider a client $j_0 \in \dclient \cap \cB$.  Since $j_0 \in \dclient$ there must be some $i^\star \in \dfac$ such that $j_0 \in \cneigh(i^\star)$.  Moreover, since $j_0 \in \cB$, $j_0$ is undecided and so by Lemma \ref{lem:sweep-preservation} we must have $\alpha_{j_0} \ge \alphat{0}_{j_0}$.
  
Suppose first that $i^\star \in \sfac$. Then $\stime_{i^\star} = \max_{j\in N(i^\star) \cap \sclient(i^\star)} \alpha_j$.  Since $i^\star \in \sfac$ we have $\alpha_j \geq \alphat{0}_j$ for all $j\in \cneigh(i^\star)$.  We claim that $\stime_{i^\star}\leq(1+\aeps)^2 \alpha_{j_0}$. Indeed, otherwise there is a client $j \in N(i^\star) \cap \sclient(i^\star) = N(i^\star) \cap \cneigh(i^\star)$ such that $\sqrtalpha_j > (1+\aeps)\sqrtalpha_{j_0}$ and so
\begin{align*}
(1+\aeps)\sqrtalpha_j& > (1+\aeps)\sqrtalpha_{j_0} + \tfrac{\aeps}{2}\cdot \sqrtalpha_{j_0} + \tfrac{\aeps}{2}\cdot \sqrtalpha_{j} & \\
& \geq (1+\aeps)\sqrtalpha_{j_0} + \tfrac{\aeps}{2}\cdot \sqrtalphat{0}_{j_0} + \tfrac{\aeps}{2}\cdot \sqrtalphat{0}_{j}  & \mbox{\small ($\alpha_{j_0} \ge \alphat{0}_{j_0}$ and $\alpha_j \ge \alphat{0}_j$)} \\
& \geq (1+\aeps)\sqrtalpha_{j_0} + \tfrac{\aeps}{2\sqrt{\gamma}}\cdot d(j_0,i^\star) + \tfrac{\aeps}{2\sqrt{\gamma}}\cdot d(j,i^\star)  & 
\mbox{\small ($j, j_0 \in \cneigh(i^\star)$)} \\
& \geq (1+\aeps)\sqrtalpha_{j_0} + (1+\aeps)d(j_0,i^\star) + (1+\aeps)d(j,i^\star)  & \mbox{\small ($\gamma \ll \aeps$ and so $\tfrac{\aeps}{2\sqrt{\gamma}} \ge (1 + \aeps)$)} \\
& \geq (1+\aeps)\Bigl(\sqrtalpha_{j_0} + d(j_0,j)\Bigr)\,, &
\end{align*}
contradicting Invariant~\ref{inv:containment}, since the $\alpha$-ball of $j$ would then strictly contain the $\alpha$-ball of $j_0$.  Hence, $\sqrt{\stime_{i^\star}} \le (1 + \aeps) \sqrtalpha_{j_0}$. Furthermore, $d(j_0,i^\star) \le \sqrt{\gamma} \sqrtalphat{0}_{j_0} \le \sqrt{\gamma} \sqrtalpha_{j_0}$ and  therefore 
  \[(1 + \sqrt{\delta} + 3\aeps)\sqrtalpha_{j_0} \ge d(j_0,i^\star) + \sqrt{\delta \stime_{i^\star}}\,.\]

  On the other hand, if $i^\star \not\in \sfac$ then (by the definition of $\sfac$)
there must be some $j \in \cneigh(i^\star)$ with $\alpha_{j} < \alphat{0}_{j}$.  By Lemma \ref{lem:sweep-preservation}, $j$ must be decided.  Then, $j \not\in \cB$ and so by Lemma \ref{lem:roundable-alpha-lower-bound} there exists some tight facility $i$ such that $(1 + \sqrt{\delta} + \aeps)\sqrtalpha_{j} \ge d(j,i) + \sqrt{\delta t_i }$. Moreover, applying the same argument as above, we must have $\sqrtalphat{0}_{j} \leq (1+\aeps) \sqrtalphat{0}_{j_0}$, since otherwise in $\alphat{0}$, the $\alpha$-ball of $j$ would strictly contain the $\alpha$-ball of $j_0$, contradicting Invariant~\ref{inv:containment}.  Then,  we have:
\begin{align*}
d(j_0,i) + \sqrt{\delta t_i} &\le d(j_0,i^\star) + d(i^\star,j) + d(j,i) + \sqrt{\delta t_i} \\
&\le \sqrt{\gamma}\sqrtalphat{0}_{j_0} + \sqrt{\gamma}\sqrtalphat{0}_{j} + (1+\sqrt{\delta} + \aeps)\sqrtalpha_{j} \\
&\le \aeps\sqrtalphat{0}_{j_0} + \aeps(1+\aeps)\sqrtalphat{0}_{j_0} + (1 + \sqrt{\delta} + \aeps)(1 + \aeps)\sqrtalphat{0}_{j_0} \\
&< (1 + \sqrt{\delta} + 10\aeps)\sqrtalpha_{j_0}\,, 
\end{align*}
where for the final inequality we used that $j_0$ is undecided and so by Lemma~\ref{lem:sweep-preservation} we have $\alpha_{j_0} \ge \alphat{0}_{j_0}$.
\end{proof}

\subsection{Showing that each solution is roundable and completing the analysis}
\label{sec:select-thresh-thet}

We start by showing that each solution $(\alpha,z)$ produced by Algorithm~\ref{alg:1} satisfies the properties of Definition~\ref{def:roundable}.

\begin{proposition}
  \label{prop:roundable}
Every solution $(\alpha,z)$ produced by Algorithm~\ref{alg:1} is roundable.
\end{proposition}
\begin{proof}
By construction, each solution $(\alpha,z)$ produced by Algorithm~\ref{alg:1} is feasible with respect to $\dualf[\lambda+\zeps]$ and $\zeps < \frac{1}{n}$. In addition, we have $\lambda \leq z_i \leq \lambda+ \zeps \leq \lambda + 1/n$ for all $i\in \facilities$. 
It remains to show that Properties \ref{item:roundable-2} and \ref{item:roundable-3} of Definition \ref{def:roundable} are satisfied.  Recall the definitions of $\dfac$ and $\dclient$, and define $\sfac$ and $\sclient$ as in~\eqref{eq:sfacdef}.   Further let $\bclient = \cW \setminus \dclient$. 

Now we show that Property~\ref{item:roundable-2} holds for $\cS = (\alpha,z,\sfac,\dclient)$ with respect to the set $\bclient$.  By Lemma \ref{lem:bad-client-range} (which shows that a client $j \in \cB$ only if $\frac{1}{81}\theta_s \le \alphat{0}_j \le 25 \cdot 20^{4}\theta_s$ for some stage $s$) we have $\cB \subseteq \cW$.  Thus $\cB \setminus \dclient \subseteq \bclient$.  Now, by Lemma \ref{lem:roundable-alpha-lower-bound} for all $j \in \clients \setminus \cB$ there exists some tight facility $i$ such that
\[
(1 + \sqrt{\delta} + 10\epsilon)^2\alpha_j \ge \left(d(j,i) + \sqrt{\delta  t_i}\right)^2
\ge \left(d(j,i) + \sqrt{\delta \stime_i}\right)^2.
\]
Similarly, by Lemma \ref{lem:roundable-dense-bound}, for all $j \in \cB \cap \dclient$, there exists either some tight facility $i$ or some special facility $i \in \sfac$ such that
\[
(1 + \sqrt{\delta} + 10\epsilon)^2\alpha_j 
\ge \left(d(j,i) + \sqrt{\delta \stime_i}\right)^2\,.
\]
Finally, we consider the remaining clients in $\cB \setminus \dclient \subseteq \bclient$.  Lemma~\ref{lem:roundable-alpha-lower-bound} shows that for every client $j \in \clients$ there is some tight facility $i$ such that
\[
  36\alphat{0}_j \ge\left(d(j,i) + \sqrt{\delta t_i}\right)^2 \geq \left(d(j,i) + \sqrt{\delta \stime_i}\right)^2.
\]
For each client $j$, let $w(j)$ be this specified tight facility $i$.  Then,
\begin{equation*}
\sum_{j \in \bclient}\left(d(j,i) + \sqrt{\delta \stime_{w(j)}}\right)^2
\le 36\sum_{j \in \bclient} \alphat{0}_j 
\le 36 \gamma \cdot \opt_k,
\end{equation*}
where the last inequality follows from Corollary~\ref{cor:W-mass-small}. 

Finally, we show that Property~\ref{item:roundable-3} must hold. Consider some $i \in \sfac$.  By
definition of $\sfac$,  we must have $j \in \cB$ for some $j \in \cneigh(i)$.
Then, by Lemma \ref{lem:bad-client-range} we must have $\frac{1}{81}\theta_s
\le \alphat{0}_j \le 25 \cdot 20^{4}\theta_s$ for some $s$.  By 
Lemma~\ref{lem:tightboundedratio} (which bounds the ratio to be at most $19^2 < 20^2$ between $\alpha_j$ and $\alpha_{j'}$ for any pair of clients $j,j'$ that share a tight edge to some common facility $i$), we must have $\alphat{0}_{j'} \in \cW$ for any $j' \in \Nt{0}(i)$.  
Moreover, by Lemma~\ref{lemma:singlepayment} (which shows that each dense client pays for at most one dense facility in $\ist{0}$), we have $\betat{0}_{ij} = 0$ for all $j \in \dclient \setminus \cneigh(i)$.  Altogether, then we have $\Nt{0}(i) \subseteq \cW$ and $\Nt{0}(i) \cap \dclient = \cneigh(i)$ and so
\[
\Nt{0}(i) \setminus \cneigh(i) = \Nt{0}(i) \setminus \dclient \subseteq \cW \setminus \dclient\,,
\]
for every $i \in \sfac$.  Notice that Lemma~\ref{lemma:singlepayment} also implies that for each dense facility $i \in \dfac$ there is some dense client that only pays for that facility in $\ist{0}$, so indeed $|\sfac| \leq |\dfac| \leq n$. It remains to prove that $\sum_{i\in \sfac} \sum_{j\in \sclient} \beta_{ij} \geq \lambda |\sfac| - \gamma \cdot \opt_k$.    To that end, recall that Invariant~\ref{inv:good-solution} implies that $\sfact{0} = \emptyset$.
Thus, every $i \in \ist{0}$ must have been tight in $\alphat{0}$, and hence $\sum_{j \in \Nt{0}(i)}\betat{0}_{ij} \ge \zt{0}_i \ge \lambda$.  Combining these
observations, for every $i \in \sfac$ we have:
\begin{equation}
\sum_{j \in \sclient(i)}\beta_{ij} = \sum_{j \in \cneigh(i)}\beta_{ij} \ge \sum_{j \in \cneigh(i)}\betat{0}_{ij} \ge \lambda - \sum_{j \in \Nt{0}(i) \setminus \cneigh(i)} \betat{0}_{ij} \ge \lambda - \sum_{j \in \cW \setminus \dclient} \betat{0}_{ij}\,,\label{eq:special-fac-price}
\end{equation}
where the first inequality follows from the definition of $\sfac$, which requires that for any $i \in \sfac$, $\alpha_j \geq \alphat{0}_j$ for all $j \in \cneigh(i)$.  By Invariant~\ref{inv:good-solution}, every client is decided in $\alphat{0}$ and so, in particular, $\sfact{0} = \emptyset$ and $\ist{0}$ contains no special facilities.  Then, by Lemma~\ref{lemma:singlepayment}, $\sum_{i \in \sfac}\betat{0}_{ij} \le \sum_{i \in \ist{0}}\betat{0}_{ij} \le \alphat{0}_j$ for all $j$.  Summing \eqref{eq:special-fac-price} over all $i \in \sfac$ we thus have
\begin{equation*}
\sum_{i \in \sfac}\sum_{j \in \sclient(i)}\beta_{ij} \ge |\sfac|\lambda - \sum_{j \in \cW \setminus \dclient}\sum_{i \in \sfac}\betat{0}_{ij} \ge
|\sfac|\lambda - \sum_{j \in \cW \setminus \dclient}\alphat{0}_{j} \ge
|\sfac|\lambda - \gamma \cdot \opt_k,
\end{equation*}
where the final inequality follows from Corollary~\ref{cor:W-mass-small}.
\end{proof}

The following theorem completes the analysis.

\begin{theorem}
\raiseprice runs in polynomial time and produces a polynomial number of close roundable solutions.
\label{thm:polynomial-time}
\end{theorem}
\begin{proof} 
  That the produced solutions are close follows from
  Proposition~\ref{prop:close-values} and the produced solutions are roundable
  follows from Proposition~\ref{prop:roundable}. We continue to bound the
  running time and the number of produced solutions.  \raiseprice produces one
  solution for each call to \sweep.  In Appendix~\ref{sec:Ashkan} we argue
  (similarly as we did for \quasisweep) that \sweep can be implemented in
  polynomial time, and it is clear that the remaining operations in \raiseprice
  can be implemented in polynomial time.  Thus, to prove both claims, it
  suffices bound the number of calls to \sweep in \raiseprice. For that purpose, define:
\[
M = \lambda + \max_{j \in \clients, i \in \facilities}d(j,i) \le 4n^7 + n^6  < n^8\,.
\]
Using  our preprocessing (Lemma \ref{lem:dis}), we note that at all times during any call to \raiseprice, we have $\alpha_j \le M$, since otherwise $\alpha$ would be infeasible (contradicting Invariant~\ref{inv:feasibility}).

Let us now bound the number of calls to \sweep in each stage.  In stage 1, we
make only 1 call to \sweep, as shown in Lemma \ref{lem:no-changes-stage-1}.  In
each stage $s > 1$, \raiseprice calls \sweep only until $\alpha_j \ge \theta_s$
and $\alpha_j \ge \alphat{0}_j$ for every undecided client $j$.
Consider a call to $\sweep$ in stage $s>1$ and let $(\alpha,z)$ be the produced solution. Let $j$  be the undecided client with the smallest $\alpha$-value in $(\alpha,z)$ (breaking ties in the order of removal from the set $A$). If $j$ was removed by Rule~$4$, we have $\alpha_j \ge \theta_s$
and $\alpha_j \ge \alphat{0}_j$ and so every undecided client has an $\alpha$-value of at least $\theta_s$ which implies the termination of stage $s$. Otherwise, as $j$ has the smallest $\alpha$-value of undecided clients it cannot be removed by Rule~$5$ (by Property~\ref{item:sb-contained}) and so it must have been removed by Rule~$3$. Therefore, by the definition of that rule, $j$ was undecided in the previous iteration and its $\alpha$-value has increased by $\zeps$ in the considered call to \sweep. By the above, we have that either stage $s$ terminates or the smallest $\alpha$-value of the undecided clients increases by at least $\zeps$. Therefore, the stage must terminate after at most  $\zeps^{-1}M=n^{O(\aeps^{-1}\gamma^{-4})}$ calls to \sweep since no $\alpha$-value is larger than $M$.

Finally, let us bound the number of stages executed in \raiseprice.  By Lemma~\ref{lem:no-changes-stage-later} after stage $s$, all clients with $\alpha_j < \theta_s$ are decided.  Then, for $s = (K\aeps)^{-1}\Theta(\log n) = \gamma^{4}O(\log n)$ we have $\theta_s > M$ and so all clients must be decided.
\end{proof}

\bibliographystyle{abbrv}
\bibliography{library}

\begin{thebibliography}{10}

\bibitem{DBLP:conf/esa/ArcherRS03}
A.~Archer, R.~Rajagopalan, and D.~B. Shmoys.
\newblock Lagrangian relaxation for the k-median problem: New insights and
  continuity properties.
\newblock In {\em Proc. 11th ESA}, pages 31--42, 2003.

\bibitem{Arthur:2006:SKM:1137856.1137880}
D.~Arthur and S.~Vassilvitskii.
\newblock How slow is the k-means method?
\newblock In {\em Proc. 22nd SoCG}, pages 144--153, 2006.

\bibitem{Arthur:2007:KAC:1283383.1283494}
D.~Arthur and S.~Vassilvitskii.
\newblock K-means++: The advantages of careful seeding.
\newblock In {\em Proc. 18th SODA}, pages 1027--1035, 2007.

\bibitem{AGKMMP}
V.~Arya, N.~Garg, R.~Khandekar, A.~Meyerson, K.~Munagala, and V.~Pandit.
\newblock Local search heuristics for $k$-median and facility location
  problems.
\newblock {\em SIAM J. Comput.}, 33(3):544--562, 2004.

\bibitem{Awasthi:2010:SYP:1917827.1918395}
P.~Awasthi, A.~Blum, and O.~Sheffet.
\newblock Stability yields a {PTAS} for k-median and k-means clustering.
\newblock In {\em Proc. 51st FOCS}, pages 309--318, 2010.

\bibitem{DBLP:conf/compgeom/AwasthiCKS15}
P.~Awasthi, M.~Charikar, R.~Krishnaswamy, and A.~K. Sinop.
\newblock The hardness of approximation of euclidean k-means.
\newblock In {\em Proc. 31st SoCG}, pages 754--767, 2015.

\bibitem{Balcan2009}
M.-F. Balcan, A.~Blum, and A.~Gupta.
\newblock Approximate clustering without the approximation.
\newblock In {\em Proc. 20th SODA}, pages 1068--1077, 2009.

\bibitem{BA10}
J.~Byrka and K.~Aardal.
\newblock An optimal bifactor approximation algorithm for the metric
  uncapacitated facility location problem.
\newblock {\em SIAM J. Comput.}, 39(6):2212--2231, 2010.

\bibitem{DBLP:conf/soda/ByrkaPRST15}
J.~Byrka, T.~Pensyl, B.~Rybicki, A.~Srinivasan, and K.~Trinh.
\newblock An improved approximation for \emph{k}-median, and positive
  correlation in budgeted optimization.
\newblock In {\em Proc. 26th SODA}, pages 737--756, 2015.

\bibitem{CG}
M.~Charikar and S.~Guha.
\newblock Improved combinatorial algorithms for facility location problems.
\newblock {\em SIAM J. Comput.}, 34(4):803--824, 2005.

\bibitem{CS04}
F.~A. Chudak and D.~B. Shmoys.
\newblock Improved approximation algorithms for the uncapacitated facility
  location problem.
\newblock {\em SIAM J. Comput.}, 33(1):1--25, 2004.

\bibitem{DBLP:journals/corr/Cohen-AddadKM16}
V.~Cohen{-}Addad, P.~N. Klein, and C.~Mathieu.
\newblock The power of local search for clustering.
\newblock {\em CoRR}, abs/1603.09535, 2016.

\bibitem{DBLP:conf/compgeom/FeldmanMS07}
D.~Feldman, M.~Monemizadeh, and C.~Sohler.
\newblock A {PTAS} for k-means clustering based on weak coresets.
\newblock In J.~Erickson, editor, {\em Proc. 23rd SoCG}, pages 11--18, 2007.

\bibitem{DBLP:journals/corr/FriggstadRS16}
Z.~Friggstad, M.~Rezapour, and M.~R. Salavatipour.
\newblock Local search yields a {PTAS} for k-means in doubling metrics.
\newblock {\em CoRR}, abs/1603.08976, 2016.

\bibitem{DBLP:journals/corr/abs-0809-2554}
A.~Gupta and K.~Tangwongsan.
\newblock Simpler analyses of local search algorithms for facility location.
\newblock {\em CoRR}, abs/0809.2554, 2008.

\bibitem{JMM03}
K.~Jain, M.~Mahdian, E.~Markakis, A.~Saberi, and V.~V. Vazirani.
\newblock Greedy facility location algorithms analyzed using dual fitting with
  factor-revealing {LP}.
\newblock {\em J. ACM}, 50:795--824, 2003.

\bibitem{Jain:2002:NGA:509907.510012}
K.~Jain, M.~Mahdian, and A.~Saberi.
\newblock A new greedy approach for facility location problems.
\newblock In {\em Proc. 34th STOC}, pages 731--740, 2002.

\bibitem{DBLP:journals/jacm/JainV01}
K.~Jain and V.~V. Vazirani.
\newblock Approximation algorithms for metric facility location and
  \emph{k}-median problems using the primal-dual schema and lagrangian
  relaxation.
\newblock {\em J. {ACM}}, 48(2):274--296, 2001.

\bibitem{DBLP:journals/comgeo/KanungoMNPSW04111}
T.~Kanungo, D.~M. Mount, N.~S. Netanyahu, C.~D. Piatko, R.~Silverman, and A.~Y.
  Wu.
\newblock A local search approximation algorithm for k-means clustering.
\newblock {\em Comput. Geom.}, 28(2-3):89--112, 2004.

\bibitem{DBLP:journals/corr/LeeSW15a}
E.~Lee, M.~Schmidt, and J.~Wright.
\newblock Improved and simplified inapproximability for k-means.
\newblock {\em CoRR}, abs/1509.00916, 2015.

\bibitem{DBLP:journals/iandc/Li13}
S.~Li.
\newblock A 1.488 approximation algorithm for the uncapacitated facility
  location problem.
\newblock {\em Inf. Comput.}, 222:45--58, 2013.

\bibitem{LiS16}
S.~Li and O.~Svensson.
\newblock Approximating k-median via pseudo-approximation.
\newblock {\em {SIAM} J. Comput.}, 45(2):530--547, 2016.

\bibitem{LV92B}
J.~Lin and J.~S. Vitter.
\newblock Approximation algorithms for geometric median problems.
\newblock {\em Inf. Process. Lett.}, 44:245--249, 1992.

\bibitem{Lloyd:2006:LSQ:2263356.2269955}
S.~Lloyd.
\newblock Least squares quantization in {PCM}.
\newblock {\em IEEE Trans. Inf. Theor.}, 28(2):129--137, Sept. 2006.

\bibitem{matouvsek2000approximate}
J.~Matou{\v{s}}ek.
\newblock On approximate geometric k-clustering.
\newblock {\em Discrete \& Computational Geometry}, 24(1):61--84, 2000.

\bibitem{Ostrovsky2013}
R.~Ostrovsky, Y.~Rabani, L.~J. Schulman, and C.~Swamy.
\newblock The effectiveness of {L}loyd-type methods for the k-means problem.
\newblock {\em J. ACM}, 59(6):28:1--28:22, Jan. 2013.

\bibitem{STA97}
D.~B. Shmoys, E.~Tardos, and K.~Aardal.
\newblock Approximation algorithms for facility location problems (extended
  abstract).
\newblock In {\em Proc. 29th STOC}, pages 265--274, 1997.

\bibitem{Vattani2011}
A.~Vattani.
\newblock k-means requires exponentially many iterations even in the plane.
\newblock {\em Discrete {\&} Computational Geometry}, 45(4):596--616, 2011.

\bibitem{Vazirani:2001:AA:500776}
V.~V. Vazirani.
\newblock {\em Approximation Algorithms}.
\newblock Springer-Verlag New York, Inc., New York, NY, USA, 2001.

\bibitem{Williamson:2011:DAA:1971947}
D.~P. Williamson and D.~B. Shmoys.
\newblock {\em The Design of Approximation Algorithms}.
\newblock Cambridge University Press, New York, NY, USA, 1st edition, 2011.

\end{thebibliography}

\appendix

\section{Implementation of \quasisweep}
\label{sec:quasi-runtime}
In Section~\ref{sec:quasi-polyn-time-algor}, we presented the procedure \quasisweep in a continuous fashion. We now describe a discrete, polynomial time implementation of \quasisweep.  As presented, \quasisweep maintains only the $\alpha$-values of each client, the value of $\theta$, and the set $A$ of active clients.  We suppose we are increasing $\theta$ at the speed of $1$, so that the value of $\theta$ corresponds to the current time.  Then, \quasisweep changes each $\alpha$-value at the speed of either 0, 1, or $-|A|$.  Moreover, this speed does not change until one of the following events happens:
\begin{enumerate}[label={Event \arabic*:},ref={\arabic*},leftmargin=*,labelindent=\parindent]
\item \label{re1} Client $j$ joins $A$: this can happen only if $\alpha_j = \theta$.
\item \label{re2} $\theta$ changes buckets: this can only happen when $\theta$ has reached the border of a bucket.
\item \label{re3} Facility $i$ becomes tight: this can happen if $(1)$ no client with a tight edge to $i$ is decreasing, and $(2)$ some client in $A$ has a tight edge to $i$.
\item \label{re4} Client $j$ gains a tight edge to facility $i$: this can happen only if $j\in A$.
\item \label{re5} Client $j \notin A$ changes buckets and enters the same bucket as $\theta$: this can happen only if $\alpha_j$ is being decreased. 
\end{enumerate}
Note that we remove a client from $A$ either immediately after it is added to $A$, at the time that some facility becomes tight, or at the time that it gains a tight edge to some (tight) facility. Therefore, we do not need to add an event for removing a client from $A$, since it only happens if one of the above events happen.

The polynomial time \quasisweep now works as follows: In each step, we find the next time that any one of the above events happens, then increase/decrease each $\alpha$-value according to its current speed to obtain a new set of values at this time.  Then, we update $\theta$, $A$, and our set of speeds and continue.  We need to show how we can efficiently compute the next event that happens, and also we need to prove that the number of such events is polynomial.  In what follows, we compute the time until each event above happens, assuming that it is the next event that happens.  Then the next event that actually happens is the event with the minimum such time (breaking ties arbitrarily).  

We now consider each of the above events in turn:
\begin{itemize}
\item The time until the Event~\ref{re1} may happen next is $\min_{\alpha_j>\theta} \alpha_j - \theta$.  Also we have exactly $n$ occurrences of this event. 

\item The time until Event~\ref{re2} may happen next is the difference between $\theta$ and the border of the next bucket, i.e., $B_{next}-\theta$, where $B_{next} = (1 + \epsilon)^{B(\theta)}$.  We have at most $O(\epsilon^{-1}\log(n))$ such events.  

\item For Event~\ref{re3},  if some client with a tight edge to $i$ is decreasing then (non-tight) facility $i$ cannot become tight (due to the choice of the speed of decrease). If no decreasing client has a tight edge to facility $i$, then the time that $i$ may become tight is
 \[
\frac{(\lambda + \epsilon_z) - \sum_{j\in \cD} [\alpha_j-d(j,i)^2]^+}{|N(i)\cap A|}. 
\]
Notice that the numerator is the current slack of facility $i$ and the denominator is the speed at which this slack decreases.
Moreover, there are at most $m =|\facilities|$ such events, since if a facility becomes tight, it will stay tight (as we discussed in our description of \quasisweep).   

\item The time until Event~\ref{re4} may happen for some edge $(j,i)$ is $d(j,i)^2 - \alpha_j$ if $\alpha_j < d(j,i)^2$, and there are at most $nm$ such events, since if an edge becomes tight, it remains tight afterwards.  

\item Finally, Event~\ref{re5} may happen only for those clients $j$ with $B(\alpha_j) > B(\theta)$.  For any such $j$, the time until Event~\ref{re5} happens is $(\alpha_j-B_{next})/|A|$. This event can happen also at most $n$ times, since once $B(\alpha_j) = B(\theta)$, $j$ is no longer decreased.  Note that when $\alpha_j$ is decreasing, we consider it to change buckets at the moment that it lies on the lower border of its current bucket (\ie at the moment that $1 + \log_{1 + \aeps}\alpha_j = B_{next}$).  It is easy to verify that still $(1+\aeps)\alpha_j \leq \alpha_{j'}$ for any $j$ and $j'$ placed in the same bucket by this rule.
\end{itemize}
\noindent
From the above, it is clear that the number of events are polynomial, and also that the next event can be computed in polynomial time.

 \section{Implementation of \sweep}
\label{sec:Ashkan}
In this section, we present a polynomial time implementation of the \sweep procedure. Our general approach is the same as described \quasisweep in Appendix~\ref{sec:quasi-runtime}, besides the set of events.  Recall that the polynomial time algorithm for \quasisweep is as follows: repeatedly find the next event that happens, then update the $\alpha$-values.  We increase $\theta$ at a rate 1, so that $\theta$ corresponds naturally to our notion of time.  Let $\theta^{(0)}$ denote the value of $\theta$ at the time that the preceding event happened.
 
We now focus on the events, explain them in detail, show the number of times that each can occur, and discuss the way that we can find each one of them.  Let $A$ be the set of active clients (as in \sweep) and let $D$ denote the set of all clients $j$ whose value $\alpha_j$ is being decreased.  Then, $\alpha_j$ is changed at a rate of 1 for every $j \in A$, $-|A|$ for every $j \in D$, and 0 for all other clients.  We now consider the events that can cause $A$ and $D$ to change.  For each such event, we show how to compute (in polynomial time) the time at which it would occur, assuming that $A$ and $D$ have not yet changed.  The next time that the behavior of \sweep can change (and the next time that an event actually occurs) is then the minimum over all of these event times.  

Given this next time $\theta$, we first update all $\alpha$-values according to their current rate of change.  Then, we compute the set $A$ of active clients, as follows.  We first add to $A$ all clients $j$ with $\alpha_j = \theta$.  Next, we remove clients from $A$ according to Rules 1-5 in \sweep.  Using the updated $\alpha$-values and the updated set $A$, we then compute the set of decreasing clients as in \sweep.   Note that until $\theta$ has increased neither the set $A$ nor any $\alpha$-values will change.  Thus, we can assume without loss of generality that the next event occurs when $\theta > \theta_0$.  For any $j \in A$ (after we update $A$) we therefore consider $\alpha_j$ to be \emph{strictly} greater than its present value when computing the set of potentially tight facilities (and hence decreasing clients): i.e.\ if $j \in A$, we consider $j \in N(i)$ if $\alpha_j \geq d(j,i)^2$ and we consider $\alpha_j > \alphat{0}_j$ if $\alpha_j \geq \alphat{0}_j$.  After computing the set of decreasing clients, we finally set $\theta_0 = \theta$ and continue.

Let us now describe how to compute the events that may cause $A$ or $D$ to change, and argue that they occur at most a polynomial number of times.  We consider the following basic events:
\begin{enumerate}[label={Event \arabic*:},ref={\arabic*},leftmargin=*,labelindent=\parindent]
\item \label{pev1} $\theta$ changes bucket or some client $j \in D$ enters the same bucket as $\theta$.
The time at which $\theta$ changes bucket can be computed exactly as described in Event 2 in our discussion of the running time of \quasisweep, and this happens at most once for each bucket.  Similarly, we can compute the next time that some client $j \in D$ enters the same bucket as $\theta$ as described in Event 5 in our discussion of \quasisweep (using the fact that here also, all clients of $D$ are decreasing at a rate of $|A|$).  Just as in \quasisweep, here also no client is decreased after entering the same bucket as $\theta$, and so this event can happen at most once per client. 
\item \label{pev2} $\alpha_j$ becomes equal to some constant $C$.  This can occur only for a client $j \in D$ with $\alpha_j > C$ or $j \in A$ with $\alpha_j < C$.  Since no value $\alpha_j$ is decreased by \sweep once it has been increased, this event can happen at most twice for any $j$ and $C$: once when $j \in D$ and once when $j \in A$.  Moreover, while $A$ and $D$ (and so the rate $r_j$ at which $\alpha_j$ is changing) remain constant, the time at which this event will occur is the $\theta$ satisfying
$r_j(\theta - \theta_0) = C - \alpha_j$.
\item \label{pev3} $\theta = \alpha_j$ for some $j$ that has not yet been added to $A$.  This event occurs only if $B(\alpha_j) = B(\theta)$ and so $\alpha_j$ is not presently decreasing or increasing.  Once $\theta = \alpha_j$ we place $j \in A$ and so this event can occur at most once for each $j$.  The time at which this event occurs is then, by definition, $\theta = \alpha_j$.
\item \label{pev4} $\theta = (C_1 + C_2\sqrt{\alpha_{j'}})^2$ for some constants $C_1$ and $C_2$ and some client $j'$ that has already been removed from $A$.  Once a client $j'$ has been removed from $A$, it is not subsequently changed by \sweep and so $(C_1 + C_2\sqrt{\alpha_{j'}})^2$ is a constant.  Thus, this event can happen at most once for each $j$, $C_1$, and $C_2$.
\item \label{pev5} A facility $i$ becomes tight and a witness for some $j \in A$.  As discussed in our analysis of \sweep, once $i$ is a tight and a witness for some $j$ it will remain so until the end of \sweep.  Thus, this event occurs at most once for each $j$ and $i$.  This event can occur only if the contributions to $i$ are increasing, in which case we must have $N(i) \cap D = \emptyset$.  The time that $i$ becomes tight is then the value $\theta$ satisfying:
\[
z_{i}=\sum_{j\in \cD \setminus A} [\alpha_j-d(i,j)^2]^++\sum_{j\in A}[\theta-d(i,j)^2]^+\,.
\]
\end{enumerate}

Now, let us show how these events capture the potential changes in $A$ and $D$, in turn.  First, suppose that $A$ changes.  If some client is added to $A$, we must have $\theta = \alpha_j$ and so Event~\ref{pev3} occurs.  Now suppose that some client $j$ is removed from $A$.  We consider the rules for removing clients from $A$ one by one:
\begin{itemize}
\item{Rule 1:}  In this case, either $j$ gains a tight edge to some tight facility $i$ or a new facility $i$ becomes tight and a witness for $j$.  In the first case, $\alpha_j = d(i,j)^2$, and so Event~\ref{pev2} occurs with $j \in A$ and $C = d(i,j)^2$.  In the second case, Event~\ref{pev5} occurs.
\item{Rule 2:} In this case, $j$ is stopped by some client $j'$.  Any such $j'$ must have already been removed from $A$.  Also, we have $2\sqrt{\theta} = 2\sqrt{\alpha_j} = d(j,j') + 6\sqrt{\alpha_{j'}}$.  Then, Event~\ref{pev4} occurs with $C_1 = d(j,j')/2$ and $C_2 = 3$.
\item{Rule 3:} In this case, we had $j \in U$ at the start of \sweep and $\alpha_j = \alphat{0}_j + \zeps$.  Then, Event~\ref{pev2} occurs with $j \in A$ and $C = \alphat{0}_j + \zeps$.
\item{Rule 4:} In this case, we have that $\alpha_j \geq \alphat{0}_j$ and $\alpha_j \geq \theta_s$.  When this first happens, Event~\ref{pev2} occurs with $j \in A$ and either $C = \alphat{0}_j$ or $C = \theta_s$.
\item{Rule 5:} In this case, there is a client $j'$ that has already been removed from $A$ such that $\sqrtalpha_j \geq d(j,j')+\sqrtalpha_{j'}$.  Then, Event~\ref{pev4} occurs with $j \in A$, $C_1 = d(j,j')$ and $C_2 = 1$.
\end{itemize}
From the above discussion, the events that can potentially cause $A$ to change will occur at most a polynomial number of times.\footnote{Here, and later we implicitly use that all $\alphat{0}$, $d(j,j')$, $d(j,i)$ are all constant throughout \sweep, and there are at most a polynomial number of distinct such values.  That is, we consider Event 2 and Event 4 with only a polynomial number of values for $C$, $C_1$, and $C_2$.}

Let us now consider how the set $D$ may change while the set $A$ remains constant.  We first consider the case in which some client $j$ is removed from $D$.  Since $j \in D$ presently, we must have that $j \in N(i)$ and $\alpha_j = t_i$ for some potentially tight facility $i$ with $N(i) \cap A \neq \emptyset$.  Then, $\alpha_j$ will stop decreasing only if one of the following happen:
\begin{itemize}
\item Client $j$ enters the same bucket as $\theta$.  Then, Event~\ref{pev1} occurs.  
\item Client $j$ is removed from $N(i)$.  Then, Event~\ref{pev2} occurs with $j \in D$ and $C = d(j,i)^2$.  
\item Facility $i$ becomes no longer potentially tight.  For this case, we first claim that there must in fact be some $j \in N(i)$ with $\alpha_j > \alphat{0}_j$.  Suppose otherwise; then we must have $\alpha_j \leq \alphat{0}_j$ for all $j \in N(i)$ and so $N(i) \subseteq \Nt{0}(i)$.  Since $N(i) \cap A \neq \emptyset$, there must be some $j' \in \Nt{0}(i) \cap A$.  Moreover, since $i$ is potentially tight $\alpha_{j'} = \alphat{0}_{j'}$.  But then, since $j' \in A$ we would then consider $\alpha_{j'} > \alphat{0}_{j'}$ for the purpose of computing this event, as described in our initial discussion.  In summary, as long as $j$ is decreasing due to some potentially tight facility $i$, we must have some client $j'$ with $\alpha_{j'} > \alphat{0}_{j'}$.  Then, $i$ remains potentially tight until $\alpha_{j'}$ decreases to $\alphat{0}_{j'}$.  When this happens, Event~\ref{pev2} occurs with $j' \in D$ and $C = \alphat{0}_{j'}$.
\end{itemize}
From the above discussion, the number of events that may cause a client $j$ to be removed from $D$ can occur at most a polynomial number of times for each value of $A$.  In particular, any client can be removed from $D$ at most a polynomial number of times in total.

Finally, let us consider the case in which a client $j$ is added to $D$.  Note that if $j \in A$ we have $B(\alpha_j) = B(\theta)$ and so $j$ cannot decrease.  Thus, we must have $j \not\in A$.  Then $j$ can begin decreasing only in the following cases:
\begin{itemize}
\item Some $i$ such that $j \in N(i)$ with $\alpha_j = t_i$ and $A \cap N(i) \neq \emptyset$ becomes potentially tight.  Similar to the discussion above, in this case, we must have that $\alpha_{j'} = \alphat{0}_{j'}$ for some $j' \in N(i) \cap A$.  Then, Event~\ref{pev2} occurs with $j' \in A$ and $C = \alphat{0}_{j'}$.
\item A client $j' \in A$ is added to $N(i)$ for some already potentially tight facility $i$ such that $j \in N(i)$ with $\alpha_j = t_i$.  In this case, we must have $\alpha_{j'} = d(j,i)^2$.  Then, Event~\ref{pev2} occurs with $j' \in A$ and $C = d(j,i)^2$.
\item The value $\alpha_{j}$ becomes equal to $t_i$ for some already potentially tight facility $i$.  In this case, let $j' \in N(i)$ be any client with $\alpha_{j'} = t_i$ presently.  Then either $j'$ is removed from $N(i)$, in which case Event~\ref{pev2} occurs with $j' \in A$ and $C = d(j,i)^2$, or $\alpha_{j'}$ is decreased until it is equal to $\alpha_j$.  This last case occurs at the time $\theta$ satisfying $|A|(\theta - \theta_0) = \alpha_{j'} - \alpha_j$.  We now argue that this event occurs at most a polynomial number of times.  Indeed, whenever this event occurs we add some $j \not\in D$ to the set $D$, and we have previously shown that any client $j$ can be removed from (and hence added back to) $D$ at most a polynomial number of times.
\end{itemize}

The above shows that we can calculate the next event in polynomial time and that there are in total at most polynomially many events. It follows that \sweep can be implemented to run in time that is polynomial in the number of clients and facilities.

 \section{Bounding the Distances}\label{appdistances}
Here we prove the following:
\begin{lemma}By losing a factor $(1+100/n^2)$ in the approximation guarantee, we can assume that the squared-distance between any client and any facility is in $[1,n^6]$, where $n=|\cD|$.
\end{lemma}
\begin{proof}
We prove that for a given instance of the $k$-means problem, $\cI=(\cF,\cD,d,k)$, we can in polynomial time output an instance $\cI\rq{}=(\cF,\cD,d\rq{},k)$ such that:
\begin{itemize}
\item The squared distance between any client and any facility is in $[1,n^6]$ in $\cI\rq{}$, i.e., for any $i\in\cF, j\in\cD$, we have $1 \leq d\rq{}(i,j)^2 \leq n^6$.
\item For any constant $\rho$, any $\rho$-approximate solution for $\cI\rq{}$ is a $\rho(1+100/n^2)$-approximate solution for $\cI$.
\end{itemize}
In what follows, we first prove the lemma for the case that $d$ can be any metric distance function, then we prove it for the case in which $d$ must be a Euclidean metric function. 
\paragraph{Metric Distance:}
We focus on the case that $d$ is a metric distance. To that end, we create $3$ instances $\cI_1, \cI_2, \cI\rq{}$ with distances $d_1,d_2,d\rq{}$ respectively. Choose $M$, such that $\opt(\cI) \leq M \leq 100\cdot\opt(\cI)$. We can use the algorithm presented in~\cite{DBLP:journals/corr/abs-0809-2554} to find such $M$. First, let $d_1(i,j) = \sqrt{\frac{n^3}{M}}d(i,j)$ for all $i\in\cF, j\in\cD$.  This results in $\opt(\cI_1)=\opt(\cI) \frac{n^3}{M}$, so $n^3/100 \leq \opt(\cI_1)\leq n^3$. Second, for any $i\in\cF, j\in\cD$ let $d_2(i,j) = \min (d_1(i,j),n^2)$. Consider any constant-factor $\rho$-approximate solution, for $\cI_1$. This solution cannot use any of the edges that we updated in the previous step, since the cost of this edge is more than $n^4 \geq n\cdot \opt(\cI_1)$. Similarly, any $\rho$-approximate solution for $\cI_2$ cannot use any such edge. 
Therefore, $\opt(\cI_2{}) = \opt(\cI_1{})$. Third, for any $i\in\cF, j\in\cD$, assign $d\rq{}(i,j) = \max(d_2(i,j),1)$. Since this step might increase the cost of any solution by at most $n$, $\opt(\cI_2) \leq \opt(\cI\rq{}) \leq \opt(\cI_2{})+n$. Now it is clear that  for any $i\in\cF, j\in\cD$, $1 \leq d\rq{}(i,j)^2 \leq n^4$. We need to show that any good solution for $\cI\rq{}$ is also a good solution for $\cI$. Note that during all these steps, we focused on the distances between clients and facilities. To guarantee that  $d\rq{}$ is metric, we make the exact same changes on the pairs of facilities and pairs of clients as well.  

Consider a $\rho$-approximate solution for $\cI\rq{}$. We know that the cost of this solution is at most $\rho\cdot OPT(\cI\rq{})$. Now consider the same solution for $\cI_2$. Since the cost of any solution for $\cI_2$ is no more than its cost for $\cI\rq{}$, the cost of this solution for $\cI_2$ is at most $\rho\cdot OPT(\cI\rq{}) \leq \rho\cdot (OPT(\cI_2)+n)$. Also the cost of the same solution for $\cI_1$ equals to its cost for $\cI_2$ so it is at most $\rho(OPT(\cI_2)+n) = \rho(OPT(\cI_1) +n) \leq \rho\cdot OPT(\cI_1)(1+100/{n^2})$,
where the last inequality is due to the fact that $n^3/100 \leq \opt(\cI_1)$. Thus the cost of the same solution for $\cI$ is at most  $\frac{M}{n^3}(\rho\cdot \opt(\cI_1)(1+100/n^2))=\rho(1+100/n^2)\cdot\opt(\cI)$. The lemma then follows by noting that $d\rq{}$ is metric since we only rescaled, increased the minimum distance, and decreased the maximum distance of the given metric $d$.\\
\paragraph{Euclidean Metric Distance:}
Now assume that the given distance function is Euclidean. We assume that clients and facilities are points in some $\ell$ dimensional Euclidean space. We first create a solution $\cI_1$, making sure that the $\opt(\cI_1)$ is bounded by a polynomial.  As in the previous case, we can use~\cite{DBLP:journals/corr/abs-0809-2554} to find an $M$ such that $\opt(\cI) \leq M \leq 100\cdot\opt(\cI)$.  We then divide each coordinate of each point by $\sqrt{\frac{n^3}{M}}$.   We get that $d_1(i,j) = \sqrt{\frac{n^3}{M}}d(i,j)$ for all $i\in\cF, j\in\cD$ and $n^3/100 \leq\opt(\cI_1)\leq n^3$. Now we cluster the points in $\cD \cup \cF$ such that the distance between any two points in different clusters is $\Omega(n)\cdot\opt(\cI_1)$. To do that, we create each cluster as follows: Pick any client $j$ that is not part of any cluster and add it to cluster $S$; we call $j$ the center of cluster $S$. While there exists a client $j\rq{}\in \cD$ that is not part of any cluster and the distance between $j\rq{}$ to its closest client in $S$ is less than $n^2/4$ add $j\rq{}$ to $S$. Let $S_1,\dots,S_s$ be the clusters that we create. This gives a partition of the clients. Now we add a facility $i$ to cluster $S_\ell$, if there exists a client $j\in S_\ell$, such that $d(i,j)< n^2/8$. This ensures that each facility is at most part of one cluster, since the distance between two clients in different clusters is more than $n^2/4$.

It is easy to see that our clusters have the following properties.
\begin{enumerate}
\item \label{prepdis1} $d_1(j,j\rq{})^2< n^6/16$ for any two clients $j,j\rq{}$ in the same cluster. This is because when we add any client to a cluster, the maximum distance in the cluster increases by less than $n^2/4$ and so $d_1(j,j\rq{})< n^3/4$ at the end of the process.
\item \label{prepdis2}  $d_1(i,j)^2 \leq n^6/8$ for any client $j$ and facility $i$ in the  same cluster.  Indeed, by the triangle inequality we have that $d_1(i,j) \leq d_1(i,j_1)+d_1(j_1,j)$ where $j_1$ is the client such that $d_1(i,j_1) < n^2/4$.  We have $d_1(j_1,j) < n^3/4$ by the previous property.
\item \label{prepdis3} $d_1(i,i\rq{})^2 \leq n^6/8$ for any two facilities $i,i\rq{}$ in the same cluster. Similarly to the previous case, let $j,j\rq{}$ be the closest client in this cluster to $i,i\rq{}$ respectively. By the triangle inequality we have that $d_1(i,i\rq{}) \leq d_1(i,j)+d_1(j,j\rq{})+d_1(j\rq{},i\rq{}) \leq n^2/4+n^3/4+n^2/4$.

\item \label{prepdis4} $d_1(i,j)^2 \geq n^4/64 \geq (n/64)\cdot \opt (\cI_1)$ for any facility $i$ and client $j$ not in the same cluster.
\end{enumerate}  
 We remove all the facilities that are not part of any cluster, since no client can be connected to them in any solution with approximation guarantee better than $n/64$ (this follows from the above property~\ref{prepdis4}). From above properties~\ref{prepdis1},~\ref{prepdis2}, and~\ref{prepdis3} it is clear that the squared-distance between any two points in the same cluster is at most $n^6/8$.  Then, the squared-distance between any point (whether corresponding to a client $j$ or a facility $j$) in some cluster and the centroid of that cluster is also at most $n^6/8$.

Next, we translate each cluster of points so that its centroid lies at the origin, which we denote by $\vec{0}$.  This preserves the distances between each pair of points in the same cluster.  Consider any two points $p_1,p_2$ in distinct clusters.  Then, 
$d_1(p_1,p_2)^2 \leq 2(d_1(p_1,\vec{0})^2+d_1(\vec{0},p_2)^2) \leq n^6/2$.  Now, we add $s$ new dimensions, one for each cluster. For each $1 \leq i \leq s$, we assign $n^2$ to the $i^{th}$ new coordinate for the points in $i^{th}$ cluster and $0$ to the rest. Let $\ell\rq{} = \ell+s$ the number of the coordinates that the points have right now. Now consider two points and the value of their coordinates, $j = (j_1,j_2,...,j_{\ell\rq{}}),j\rq{}=(j\rq{}_1,j\rq{}_2,...,j\rq{}_{\ell\rq{}})$, we have \[d(j,j\rq{})^2=\sum_{k=1}^{\ell\rq{}} (j_k-j\rq{}_k)^2=\sum_{k=1}^{\ell} (j_k-j\rq{}_k)^2+ \sum_{k=\ell+1}^{\ell\rq{}} (j_k-j\rq{}_k)^2\leq n^6/2+2n^4. \] 
This guarantees that the maximum squared distance between any two points remains less than $n^6/2+2n^4$.   Also, it still holds that any solution with an approximation guarantee better than $n/64$ can only connect the clients in a cluster to the facilities in the same cluster, since the squared distance between any facility and any client in different clusters is at least $2n^4$, which is more than $(n/64)\cdot\opt(\cI_1)$. Now we need to make sure that the distance between the facilities and the clients is at least one. To that end, we add one new dimension and assign one for facilities and zero for clients in this coordinate. Similarly to the analysis of the general metric, we can show that any $\rho$-approximate solution for the new instance is also a $\rho(1+100/n^2)$-approximate for $\cI$, since we increase the cost of any solution by at most $n$. Note that the last step does not increase the distance-squared between any two points by more than one so the maximum distance-squared between any two points is at most $n^6/2+2n^4+1 \leq n^6$. 

Clearly, the running time of this procedure is $\poly(n)$.
\end{proof}

\end{document}